\documentclass[12pt]{article}

\usepackage{amsmath}
\usepackage{graphicx,psfrag,epsf}
\usepackage{enumerate}
\usepackage{natbib}
\usepackage{url} % not crucial - just used below for the URL 
\usepackage{hyperref}
\usepackage{graphicx} 
\usepackage{amsthm,amssymb}
\usepackage{bm}% bold math
\usepackage{hyperref}% add hypertext capabilities
\usepackage{dsfont}
\usepackage[caption=false]{subfig}
\usepackage{tabularx}
\usepackage{float}
\usepackage{algorithm}
\usepackage{algpseudocode}
\usepackage[section]{placeins}%This prevents placing floats before a section.
\usepackage{multirow}
\usepackage{todonotes}

\newtheorem{theorem}{Theorem}[section]

\newtheorem{lemma}[theorem]{Lemma}

\newtheorem{remark}[theorem]{Remark}
\newtheorem{proposition}[theorem]{Proposition}

\DeclareMathOperator{\Tr}{Tr}

\begin{document}

	\def\spacingset#1{\renewcommand{\baselinestretch}%
		{#1}\small\normalsize} \spacingset{1}

	%%%%%%%%%%%%%%%%%%%%%%%%%%%%%%%%%%%%%%%%%%%%%%%%%%%%%%%%%%%%%%%%%%%%%%%%%%%%%%

%	\if1\blind
%	{
		\title{ Maximizing the algebraic connectivity in multilayer networks with arbitrary interconnections}
		\author{ Ali Tavasoli$^{1}$,  Ehsan Ardjmand$^{2}$, Heman Shakeri$^{3}$
		\footnote{$^1$Department of Mechanical Engineering,
		Payame Noor University, Tehran, Iran. $^2$Department of Analytics and Information Systems, College of Business, Ohio University, OH, USA. $^3$(Corresponding author) School of Data Science, University of Virginia, Charlottesville, Virginia, USA.}}

		\maketitle
	
	\begin{abstract} 

The role of multilayer networks on the emergence of several real world phenomena has been impacted network science research in recent years. The second smallest eigenvalue of the Laplacian matrix, known as algebraic connectivity,  is determinative in characterizing  properties such as diffusion speed and robustness.
  In this paper, we go beyond the special structure of  one-to-one interconnection and study multilayer networks with arbitrary interconnections and investigate the problem of maximizing algebraic connectivity by allocating interlink weights subject to a limited total budget $c$.
 We show that our formulated optimization problem is impacted by a threshold budget $c^*$ below which the maximum algebraic connectivity reaches a known upper-bound that is subject to regular optimal weights--that may or may not be  uniform depending on the interlayer structure.
 For efficient numerical approaches in regions of no analytical solution, we cast the problem into a convex optimization and considered the primal-dual setting to enable exploration from several perspectives. Particularly, a geometric transformation of dual variables leads to a graph embedding problem that is easier to interpret and is related to optimum diffusion phases, as well as to interlayer and intralayer interactions, in each region. Allowing arbitrary interconnections entails regions of multiple transitions, thus we observe more diverse diffusion phases with respect to the one-to-one interconnection case. We derive several analytical results in multilayer networks with all-pairs interconnection possibility, and investigate the associations between the optimal weights and Fiedler vector components of each layer. 
  The algebraic connectivity divided by the number of nodes in a subgraph, which we call specific algebraic connectivity, plays an important role in explaining the results. Finally, we study the placement of a limited number of interlinks by greedy heuristics, using the subgraph Fiedler vector components. 
 
	\end{abstract}
	
	\noindent%
  
	\maketitle
\section{Introduction}
We live in a world of networks that are seldom isolated and often function strongly based on each other \citep{Boccalettia2006}. Such interacting networks can be found in every discipline, with social \citep{Cozzo2015Triadic,Estrada2014}, biological \citep{Sahneh2012}, transportation \citep{Domenico2014Navigability,Estrada2014}, supply chain \citep{borgatti2009social,kim2011structural}, engineering \citep{Buldyrev2010,Mesbahi2010}, and  sport game \citep{buldu2018using} networks representing only a few examples of systems that can perform highly interconnected dynamics. The operation of such interdependent systems may be considered in multiple layers, thereby motivating their multilayer name. While studying multilayer networks, their interlayer structural property is a focal subject area \citep{Arenas2014multilayer,boccaletti2014structure} due to its impact on different dynamical and functional features of such networks; for example, percolation \citep{Buldyrev2010,Son2012Percolation,Yagan2012OptAlc,Kryven2019Percolation}, robustness \citep{Gao2012Robustness,Min2014Robustness,Yagan2019Robustness}, epidemic spreading \citep{Saumell2012Epidemic,Dickison2012Epidemics,Yagan2012Contagion,Yagan2019Contagion}, synchronization \citep{Aguirre2014Synchronization,Wang2019Synchronization}, diffusive behavior \citep{Arruda2018spreading,Cencetti2019Diffusive}, and controllability \citep{Moothedath2019Controllability}. Consequently,  significant effort have  spurred towards optimizing the design of inter-structures  \citep{Yagan2012OptAlc, Li2015RobustAllocation, Tejedor2018Directed,Moothedath2019Controllability,Pan2019Spreading,Yang2019Synch,Chattopadhyay2019Robustness}.  \\

The second smallest eigenvalue of Lagrangian of the graph represents the connectivity of networks  \citep{PietBook} and was appropriately coined by \citet{Fiedler1973}, the algebraic connectivity of a graph. Moreover, the eigenvector corresponding to algebraic connectivity is named Fiedler vector that plays a key role in spectral partitioning of networks \citep{PietBook}. %When the subgraph components are connected the algebraic connectivity of the multilayer graph is positive, and its magnitude reflects the level of connectedness of the graph. 
Algebraic connectivity increases monotonically by adding links \citep{Fiedler1973}  and can be considered as a measure of network robustness \citep{Jamakovic2007}. 
Moreover, various bounds in graph partitioning, optimal graph labeling, min-sum problems, or bandwidth optimization can be obtained using the second smallest eigenvalue as a key factor \citep{Juvan1993Bandwidth,Helmberg1995Bounds}. The convergence speed of various processes such as mixing Markov chains on graphs \citep{bremaud2013markov}, reaching consensus in multi-agent systems \citep{Jadbabaie2003NearestNeighbor,Olfati-Saber2004Consensus,Olfati-SaberFlock,Jadbabaie2007Flocking}, synchronization of coupled oscillators \citep{Strogatz2001, Arena2008Sync}, or diffusion dynamics on networks \citep{GomezDiffusion2013} are controlled by the second smallest eigenvalue of the Laplacian. \\

Algebraic connectivity of multilayer networks has recently been topic of several works, with the effect of interlinks being studied more than any other thing. When the interconnection follows a one-to-one interconnection, with varying interlink weights, the algebraic connectivity grows linearly with increasing weights up to a critical threshold, say $c^*$, and then enters a nonlinear region afterwards \citep{GomezDiffusion2013}.  The existence of such threshold gives rise to structural transition in interdependent networks \citep{Radicchi2013}. Below the threshold $c^*$ the individual networks are structurally distinguishable, while above that the multilayer network acts as a whole \citep{Radicchi2013}. When at least one of the network components has vanishing algebraic connectivity, the structural transition disappears \citep{sahneh2014exact, Mieghem2019Regular} and, hence, components of such interconnected network topologies become indistinguishable despite very weak coupling between them \citep{sahneh2014exact}. On the other hand, subgraphs are coupled more difficultly when their algebraic connectivity values are close to each other \citep{sahneh2014exact, shakeri2020designing}. Some bounds \citep{Radicchi2013,sahneh2014exact,Mieghem2019Regular} and exact expressions \citep{sahneh2014exact} for the coupling threshold were derived. \citet{Radicchi2014Supercritical} compares the existence of transition threshold to thermodynamic behavior of substances in transition from normal to supercritical fluid. Moreover, \citet{Mieghem2019Regular} discuss the physical meaning of $c^*$ in terms of the minimum cut. Abrupt structural transition is not limited to varying coupling strength as it can be equivalently observed, for instance, by varying the number of interconnections \citet{Gregorio2014algebraic} or by layer degradation \citep{Cozzo2019LayDegrd}. \\
%In literature, interlinks, or interlayer links, are links that connect two different networks, while intralinks, or Intralayer links, denote those within a single network. 

Another characteristic phenomenon in interconnected networks is supper-diffusion \citep{GomezDiffusion2013}, where under certain conditions the diffusion in the interconnected network takes place faster than in either of the networks separately. \citet{sahneh2014exact} place the superdiffusion with respect to the coupling threshold and report diversified behaviors in interconnected networks. While diffusion is monotonic with respect to coupling strength in undirected networks \citep{GomezDiffusion2013}, their directed counterparts \citep{Tejedor2018Directed} can exhibit a nonmonotonic behavior resulting in a faster diffusion at an intermediate degree of coupling than when the two layers are fully coupled. \\

In this paper, we are searching for optimal inter-structures that maximize the algebraic connectivity. In a single layer  graph, with variable edge-weights subject to a total budget, \citet{boyd2004fastestChain} and \citet{GoringShadowSeperator} show that maximizing algebraic connectivity  corresponds to a dual semidefinite optimization problem  and   the optimal solutions of the dual are related to the eigenvectors of the optimal algebraic connectivity. Additionally,  the dual problem can be interpreted as an embedding of the single-layer graph in $\mathbb{R}^n$ (optimal realization of the graph in Euclidean space), and the optimal embedding has structural properties tightly connected to the separators of the graph. We pose similar problems in multilayer networks and similarly, our goal is to  allocate weights on the interlayer links so as to maximize the smallest positive eigenvalue of the Laplacian. In this paper, this is achieved  by formulating the primal-dual program \citep{boyd2004convex} and deriving its properties, and then extracting the equivalent graph realization problem and identifying its features with respect to the multilayer network's structure. \\

The significant part of above works on interdependent networks has been paid on a one-to-one interconnection between nodes of different layers. Particularly, the multilayer graph is usually a multiplex where the number of nodes in each layer is the same and the interconnection matrix $B = pI$, with $I$ being the identity matrix and $p$ the coupling strength \citep{Mieghem2019Regular}. However there are many real world examples breaking the especial one-to-one interconnection pattern so that a node in a layer may interact with multiple nodes from other layer \citep{Mieghem2017Interdepen, Arena2018Multiple, Mieghem2019Regular}. However, the functionality of such special structures is very limited. Figure \ref{fig:begin} shows a small multilayer network where only one change in interlayer graph with respect to a one-to-one interconnection impacts its spectral properties; in particular, supperdiffusion occurs by fairly small budgets, this is impossible in a corresponding multiplex network whatever the amount of budget.     
\begin{figure}[!htb]
	\begin{center}
		\includegraphics[clip,width=.3\columnwidth]{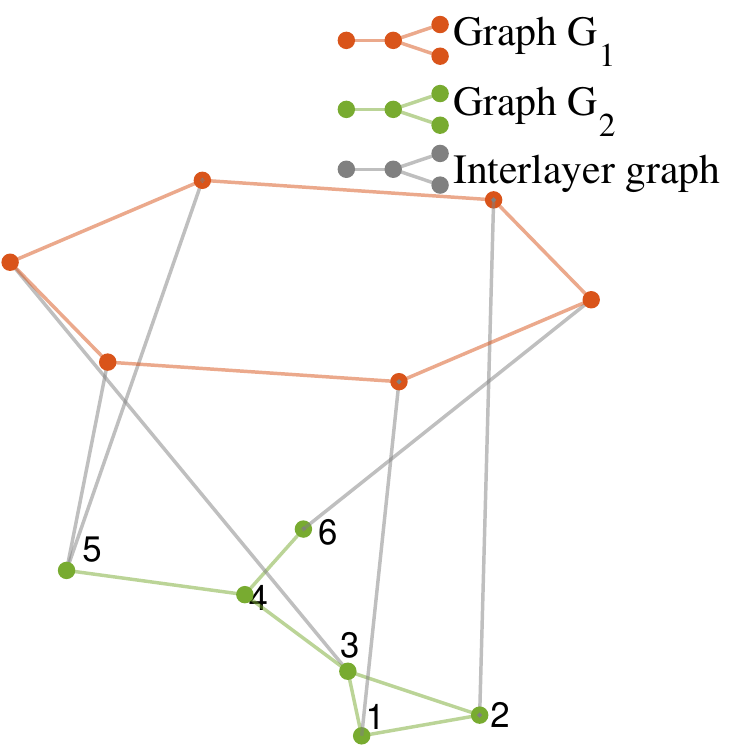} 
		\caption{A small interconnected network.}
		\label{fig:begin}
	\end{center}
\end{figure}
Furthermore, an effective strategy to make an interdependent system more robust is to bring the superposition of the layers as close as possible to an all-to-all topology \citep{Radicchi2013}. These motivate the present research after the recent related work \citep{shakeri2020designing} where the authors investigated the special case of multilayer networks with one-to-one interconnection. We relax the assumptions regarding one-to-one structures and allow for multilayer structures where the number of nodes in each layer can be different and the pattern of interconnections is arbitrary.  \\

Examining the algebraic connectivity in a primal-dual setting allows for a versatile approach that reflects a multi-sided view of the problem. This finds more importance when working under an interweaving environment where the presence of several nonlinear phenomena gives birth to diversified dynamical behaviors \citep{sahneh2014exact}. While our primal problem goes into diffusion speed \citep{GomezDiffusion2013} over the multilayer, we show that its dual is related to characteristic valuation \citep{Radicchi2013} and, hence, informative of diffusion phase. This duality enables us to understand the optimal diffusion route in each region. Here, the primal and dual optimization problems are impacted by structural transition  \citep{Arena2018Multiple}. Our analytical approach shows that primal problem gives a regular weight distribution \citep{Mieghem2017Interdepen,Mieghem2019Regular}, whenever feasible, before the (first) transition threshold, where the dual problem reveals an intralayer optimal diffusion phase.  For larger total budget values, the weight distribution becomes nonregular and different phases of interlayer diffusion are activated. This is related to interlink and intralink Fiedler cuts of the multilayer graph before and after transition \citep{Gregorio2014algebraic}. However, here, due to multiple structural transitions, we note more diverse  optimal diffusion phases. 

When there is no restriction on the interconnection pattern, we derive several analytical results for maximum algebraic connectivity and optimal weight distribution before threshold, as well as for different transition thresholds and the conditions under which supperdiffusion is possible. Furthermore, we find out a positive correlation between optimal weights and the components of subgraph Fiedler vectors after the initial transition. We observe the role of a quantity equal to the algebraic connectivity divided by number of nodes in a subgraph, which we call specific algebraic connectivity, in sorting different results. When interconnections are restricted to a given admissible set, we note the conditions under which optimal weights are not uniform before the transition $c^*$. This is different from conditions of multilayer networks with one-to-one interconnection \citep{shakeri2015PRL,shakeri2020designing} where optimal weights before $c^*$ are always uniform. Another problem pursued in this work is well-interconnected networks where our design parameter changes from interlink weights to interconnection pattern, which plays a key role in characterizing multilayer spectral properties. \citet{Mieghem2017Interdepen} and \citet{Mieghem2019Regular} investigate the multilayer spectra under some special inter-structures. However, it is not definite  which interconnection pattern will lead to maximum algebraic connectivity for a given number of interlinks. We address this problem through a simple greedy approach \citep{Boyd2006Growing}, and observe that nodes having substantially different components in a subgraph Fiedler vector are determinative in achieving well-interconnected multilayer networks.  \\
%nodes with most positive and negative values in Fiedler vector of the subgraph with larger algebraic connectivity are assigned multiple interlinks. These determining nodes are interconnected to nodes which are far from each other in the other layer, i.e. the nodes that their corresponding Fiedler vector components are quite different.  

We organize the remainder of the paper  as follows. Section \ref{sec:Model} describes multilayer networks with arbitrary interconnections. We formulate the maximum algebraic connectivity problem in Section \ref{Sec:MaxLam2}, and derive some of its main properties. In Section \ref{sec:SDP}, we analyze primal-dual setting for maximizing algebraic connectivity in multilayer networks, and consider maximum algebraic connectivity of multilayers with all-pairs interconnection possibility in Section \ref{sec:all-pair}. Furthermore, we investigate the condition when the interlinks can be chosen only from a given admissible set in Section \ref{sec:admissible}. In Section \ref{sec:Well}, we change our optimum design parameter from interlink weights to interconnections pattern and suggest well-interconnected multilayer networks. Finally Section \ref{sec:Conclusion} is devoted to concluding remarks and a discussion on the applications. \\

\section{Multi-layer network with arbitrary interlayer connections}\label{sec:Model} 

Let $G=\left(V,E\right)$ represent an undirected network and by $V=\left\{ 1,\ldots,n\right\} $ and $E\subset{V\choose 2}$, we denote the set of nodes and links. 
For a link $e$ between nodes $i$ and $j$, i.e., $e:\{i,j\}\in E$, we define a nonnegative value $w_{ij}$ as the weight of the link.
Given $G$  a multilayer network,  let $G_1 = \left\{ V_1,E_1\right\}$ and $G_2=\left\{ V_2,E_2\right\}$, $|V_1|=n$, $| V_2|=m$, represent the layers, and a bipartite graph $G_3=\left\{ V,E_3\right\}$ with $E_3\subseteq \{\{i,j\}: i\in V_1, j\in V_2\}$ are connecting the layers. The whole multilayer network holds total number of nodes $N=n+m$. Throughout the paper, we use the term \textit{intralayer} links for $E_1$ and $E_2$, and \textit{interlayer} links for $E_3$. \\
%We assume, without loss of arbitraryity, that the number of nodes in $G_1$ is not less than that in $G_2$, $n\geq m$. 

The links in $G_3$ bridge $G_1$ and $G_2$ and should be chosen strategically, for instance in a way that minimizes the disruption of the flow of information, electric power or goods, or to avoid failures against attackers and possible errors that can  fragment the system or cause  cascading phenomena \citep{Buldyrev2010}.
The edge weights of $G_3$ are design parameters. \\

The Laplacian matrix is defined as 
\begin{equation}
\label{eq:Laplacian}
L(w):=
\sum_{\{i,j\}\in E_1\cup E_2}B_{ij}+\sum_{\{i,j\}\in E_3}w_{ij}B_{ij},
\end{equation}
where $B_{ij}:= (\delta_i-\delta_j)(\delta_i-\delta_j)^T$, for each link $\{i,j\}$, and $\delta_i$ is the delta function at vertex $i$. In particular, we think of the Laplacian matrix of $G$ as a function of the interlayer weights $w$.
Enumerating the vertices in $V_1$ followed by the vertices in $V_2$, we can write $L(w)$ in block form in terms of the Laplacian matrices of the layers, $L_1$ and $L_2$, as follows
\begin{equation} \label{eq:LaplacMatrix}
\begin{aligned}
L(w)=
\begin{bmatrix}
L_1+\text{diag}\left(W\boldsymbol{1}_m\right) & -W \\
-W^T & L_2+\text{diag}\left(W^T\boldsymbol{1}_n\right)
\end{bmatrix}
\end{aligned}
\end{equation}
We will first assume that it is possible to connect any node of $G_1$ to any node in $G_2$ and $W_{n\times m}=\left[w_{ij}\right]$ consists of nonnegative weights. We use $\boldsymbol{1}_n$ and $\boldsymbol{1}_m$ to denote the $n$- and $m$-dimensional all ones vectors, respectively. Recall that the Laplacian matrix $L(\omega) $ is positive semidefinite and has (at least for connected networks) one zero eigenvalue with eigenvector $\boldsymbol{1}=[1,\dots,1]^{T}$, the vector of all ones of appropriate length. The eigenvalues of $ L(\omega) $ are ordered as $ 0=\lambda_1(\omega)\leq\lambda_2(\omega)\leq\lambda_3(\omega)\leq\dots\leq\lambda_{n}(\omega) $.\\

Our goal is to  allocate weights on the interlayer links, subject to a total budget $c$ such that $\sum w_{ij} = c$,  to maximize the smallest positive eigenvalue of the Laplacian. After formulating the primal-dual program and deriving its properties, the equivalent graph realization problem in is extracted and its features with respect to the multilayer network's structure are identified. One difference between the present research and the recent related work \citep{shakeri2020designing} is that there the authors considered a special case where both individual networks hold the same number of nodes and each node is interconnected exactly to one node in the other network (one-to-one interconnection). However, in this paper, we relax the assumptions regarding  one-to-one interconnections and allow for multilayer structures where the number of nodes in each layer can be different and the pattern of interconnections is arbitrary.

\section{Maximum algebraic connectivity in arbitrary multilayer networks}\label{Sec:MaxLam2}
%While the previous works have considered homogeneous interlinks, we consider heterogeneous case. 
%The algebraic connectivity, i.e. 
The smallest positive Laplacian eigenvalue is called  the algebraic connectivity of graphs and is characterized as
\begin{equation}\label{eq:lambda2}
\begin{gathered}
\lambda_2\left(L\right)=\underset{\substack{v^T\boldsymbol{1}=0 \\ {v\neq 0}}}{\text{min}} \ \ \frac{v^TLv}{\| v\|^2}
\end{gathered}
\end{equation}
We  maximize $\lambda_2(L)$ by distributing a total budget $c$ over the interlayer edges, formulated as
\begin{equation}\label{eq:maxLam}
\begin{gathered}
F\left(c\right):=\underset{\substack{w\geq0 \\ w^T\boldsymbol{1}=c}}{\text{max}} \  \lambda_2\left[L\left(w\right)\right]
\end{gathered}
\end{equation}
We rewrite \eqref{eq:lambda2}  by separating the components of $v$,
\begin{equation}\label{eq:v1v2}
\begin{gathered}
v_1^T\left(L_1+\text{diag}\left(W\boldsymbol{1}_m\right)\right)v_1-2v_1^TWv_2+v_2^T\left(L_2+\text{diag}\left(W^T\boldsymbol{1}_n\right)\right)v_2 \\ -\lambda_2\left(L\right)\left(\|v_1\|^2+\|v_2\|^2\right)\geq 0, \ \ \ \forall \  v_1^T\boldsymbol{1}_n=-v_2^T\boldsymbol{1}_m
\end{gathered}
\end{equation}
where $\left[v_1^T \ v_2^T\right]^T = v$, $v_1\in\mathbb R^n$ and $v_2\in\mathbb R^m$.
We further decompose  $v_1$ and $v_2$,
\begin{equation}\label{eq:vdecompos}
\begin{gathered}
v_1=\alpha\boldsymbol{1}_n+u_1, \ \ v_2=-\frac{\alpha n}{m}\boldsymbol{1}_m+u_2, \ \ \ \forall u_1\in\mathbb R^n, \ u_2\in\mathbb R^m, \ \ \  u_1^T\boldsymbol{1}_n=u_2^T\boldsymbol{1}_m=0
\end{gathered}
\end{equation}
where $\alpha$ is a scalar. Substituting \eqref{eq:vdecompos} in \eqref{eq:v1v2} gives the following inequality that is used in the subsequent  lemmas.
\begin{equation}\label{eq:lam2main}
\begin{gathered}
\alpha^2\left(1+\frac{n}{m}\right)\left[\left(1+\frac{n}{m}\right)c-n\lambda_2\right]  +2\alpha\left(1+\frac{n}{m}\right)\left[u_1^TW\boldsymbol 1_m-\boldsymbol 1_n^TWu_2\right] \\ +u_1^T\left(L_1+\text{diag}\left(W\boldsymbol{1}_m\right)\right)u_1-2u_1^TWu_2+u_2^T\left(L_2+\text{diag}\left(W^T\boldsymbol{1}_n\right)\right)u_2 -\lambda_2\left(\|u_1\|^2+\|u_2\|^2\right)\geq 0, \\ \forall \ \alpha, \ u_1^T\boldsymbol{1}_n=u_2^T\boldsymbol{1}_m=0
\end{gathered}
\end{equation}
\begin{lemma}\label{lem:lam2bound1}
	The maximum algebraic connectivity function $F\left(c\right)$ in \eqref{eq:maxLam} is upper-bounded as
	\begin{equation}\label{eq:lam2bound}
	\begin{gathered}
	F\left(c\right)\leq\left(\frac{1}{n}+\frac{1}{m}\right)c
	\end{gathered}
	\end{equation}
\end{lemma}
\begin{proof}
	Since the inequality \eqref{eq:lam2main} must hold for every $\alpha$, it follows the coefficient of $\alpha^2$ must be nonnegative, so that
	\[
	\left(1+\frac{n}{m}\right)c-n\lambda_2\geq 0
	\]  
	which is satisfied only for $\lambda_2\leq\left(\frac{1}{n}+\frac{1}{m}\right)c$.
\end{proof} 

The  upper-bound  in \eqref{eq:lam2bound} is independent of the number and pattern of interconnections.
%however, theattaining this bound 
% and thus if the bound is attainable (reachable), the maximum algebraic connectivity is the same for any interconnections satisfying the reaching conditions. However, the reaching conditions depend on number and pattern of interconnections, as discussed below. 
 In a  multilayer with one-to-one interconnection \citep{shakeri2015PRL}, $n=m$, the bound \eqref{eq:lam2bound} is verified as $F\left(c\right)\leq\frac{2c}{n}$.

\begin{lemma}\label{lem:UpBound}
	The upper-bound \eqref{eq:lam2bound} is attainable only if the following regulatory conditions are satisfied
	\begin{equation}\label{eq:RegulCond0}
	\begin{gathered}
	W\boldsymbol 1_m\in\text{span}\{\boldsymbol 1_n\}, \ \ W^T\boldsymbol 1_n\in\text{span}\{\boldsymbol 1_m\}
	\end{gathered}
	\end{equation}
	In other words, $W$ has constant row sum and column sum.
\end{lemma}
\begin{proof}
	When the upper-bound  is reached, i.e. $\lambda_2=\left(\frac{1}{n}+\frac{1}{m}\right)c$, inequality  \eqref{eq:lam2main} reads
	\begin{equation*}
	\begin{gathered}
	2\alpha\left(1+\frac{n}{m}\right)\left[u_1^TW\boldsymbol 1_m-\boldsymbol 1_n^TWu_2\right] \\ +u_1^T\left(L_1+\text{diag}\left(W\boldsymbol{1}_m\right)\right)u_1-2u_1^TWu_2+u_2^T\left(L_2+\text{diag}\left(W^T\boldsymbol{1}_n\right)\right)u_2 -\lambda_2\left(\|u_1\|^2+\|u_2\|^2\right)\geq 0, \\ \forall \ \alpha, \ u_1^T\boldsymbol{1}_n=u_2^T\boldsymbol{1}_m=0
	\end{gathered}
	\end{equation*} 
	Since this inequality must hold for every $\alpha$, the coefficient of $\alpha$ must vanish
	\begin{equation}\label{eq:alfaCoef}
	\begin{gathered}
	u_1^TW\boldsymbol 1_m-\boldsymbol 1_n^TWu_2=0
	\end{gathered}
	\end{equation}
	Setting $u_2=0$ in \eqref{eq:alfaCoef} yields $u_1^TW\boldsymbol 1_m=0$, which in addition to $u_1^T\boldsymbol 1_n=0$ imply the vector $W\boldsymbol 1_m$ belongs to the space spanned by $\boldsymbol 1_n$.  Similarly, setting $u_1=0$ in \eqref{eq:alfaCoef} will yield $\boldsymbol 1_n^TWu_2=0$, and thus the vector $\boldsymbol 1_n^TW$ belongs to the space spanned by $\boldsymbol 1_m$. 
\end{proof}

\begin{remark}
	Knowing that the total budget $c$ satisfies $\boldsymbol 1_n^TW\boldsymbol 1_m=c$, the regularity condition \eqref{eq:RegulCond0} implies
	\begin{equation}\label{eq:RegulCond}
	\begin{gathered}
	W\boldsymbol 1_m=\frac{c}{n}\boldsymbol 1_n, \ \ \ W^T\boldsymbol 1_n=\frac{c}{m}\boldsymbol 1_m
	\end{gathered}
	\end{equation}
indicating the nodes in an individual layer are all assigned the same total interlayer weight (i.e. equal weighted interlayer degree for all nodes of a layer). This generally does not imply identical weights for all interlinks. \citet{shakeri2015PRL} show uniform optimal weights in a  one-to-one interconnection structure when $c\leq c^*$. Another case where regularity is feasible with uniform weights is an all-pairs interconnection pattern \citep{Mieghem2017Interdepen, Mieghem2019Regular}. A more general case where regularity is accompanied with uniform interlink weights is when interlayer connectivity follows $k$-to-$k$ coupling scheme ($k$ integer) \citep{Mieghem2019Regular}.
%	Therefore, for upper-bound \eqref{eq:lam2bound} to be attainable, one condition is that, the interconnection pattern should satisfy the regularity conditions under uniform weights.   
\end{remark}

	When regularity is feasible, one can check that $\lambda=\left(\frac{1}{n}+\frac{1}{m}\right)c$ is always an eigenvalue of $L$ with eigenvector $\begin{bmatrix} m\boldsymbol{1}_n \\ -n\boldsymbol{1}_m \end{bmatrix}$. Recall that this eigenvalue is zero when $c=0$, and that the algebraic connectivity is bounded by this eigenvalue (Lemma \ref{lem:lam2bound1}). In fact, a perturbation approach for sufficiently small values of budget $c$--starting from zero and increasing--indicates that the maximum algebraic connectivity is $\lambda=\left(\frac{1}{n}+\frac{1}{m}\right)c$ with the corresponding Fiedler vector	$\begin{bmatrix} m\boldsymbol{1}_n \\ -n\boldsymbol{1}_m \end{bmatrix}$. 
This eigenvalue increases linearly with $c$ and by increasing the coupling budget and at some point,  there exists a transition threshold $c^*$ where the second and third smallest eigenvalues coalesce\footnote{For nonvanishing subgraph algebraic connectivities, this concurrence of different eigenvalues is inevitable because the eigenvalues of $L$ are continuous functions of the coupling strength $c$ \citep{sahneh2014exact, Mieghem2019Regular}.}.   %Therefore, $\lambda=\left(\frac{1}{n}+\frac{1}{m}\right)c$ will remain the second smallest eigenvalue of $L$ until it reaches the third eigenvalue at a threshold budget $c^*$. 
After $c^*$, the  regularity conditions is broken and the optimal weights will be generally non-regular and the maximum algebraic connectivity is a nonlinear function of $c$;  Lemma \ref{lem:Opt0} summarizes this. 
%An interesting consequence is that there is no need to solve the optimization problem, which can be found analytically, when the available budget is less than $c^*$ and regularity is feasible. \\

\begin{lemma}\label{lem:Opt0}
	Suppose the regularity conditions given by \eqref{eq:RegulCond0} are feasible. For budget values $c$ not greater than the threshold $c^*$, $c\leq c^*$, the solution to the maximum algebraic connectivity problem \eqref{eq:maxLam} is $\lambda^*=\left(\frac{1}{n}+\frac{1}{m}\right)c$ with corresponding eigenvector $\begin{bmatrix} m\boldsymbol{1}_n \\ -n\boldsymbol{1}_m \end{bmatrix}$, that is only attainable  by  regularity conditions \eqref{eq:RegulCond0}.
\end{lemma}
Further bounds can be derived based on the algebraic connectivity of the individual layers and their corresponding eigenvectors. Specifically, denoting by $\lambda_2^{(1)}$ and $\lambda_2^{(2)}$ the second smallest eigenvalue of $L_1$ and $L_2$, respectively, and by $u_2^{(1)}$ and $v_2^{(2)}$ the corresponding Fiedler vectors we have (see Appendix \ref{AppBounds}):

\begin{align}\label{eq:lam2bound1-2}
\lambda_2<\lambda_2^{(1)}+{u_2^{(1)}}^T\text{diag}\left(W\boldsymbol 1_m\right)u_2^{(1)}\\
\lambda_2<\lambda_2^{(2)}+{v_2^{(2)}}^T\text{diag}\left(W^T\boldsymbol 1_n\right)v_2^{(2)}
\end{align}

Attempting to maximize the minimum of the right hand sides of \eqref{eq:lam2bound1-2}, suggests assigning most weight to the node with largest entries in the corresponding Fiedler vectors. This signifies the importance of Fiedler vectors of the layers in determining the optimal weights that will be discussed further in the paper.

\section{Primal and dual semidefinite programmings}\label{sec:SDP}

The problem in \eqref{eq:maxLam} is a convex optimization problem with a concave objective \citep{Boyd2006FastestMP} and linear constraints. For $c\le c^*$, with regularity feasible, the solution to \eqref{eq:maxLam} is determined analytically. However, for $c>c^*$, the optimal weights are generally nonregular and  numerical solutions of are required. 
%To this end, we recast the problem defined by \eqref{eq:lambda2} and \eqref{eq:maxLam} as a primal semidefinite programming (SDP) \citep{Boyd96SDP, boyd2004convex} and investigate the associated dual problem. 
%While the primal problem is associated with $\lambda_2$ that tells us about the speed of diffusion over the multilayer network, the dual problem, or the equivalent embedding problem, is associated with Fiedler vector and informative of optimum diffusion phase. 
Our approach in this section closely follows \citep{Boyd2006FastestMP, GoringShadowSeperator, shakeri2020designing}. 

\subsection{Primal SDP}\label{sec:Primal}
We recast  \eqref{eq:maxLam} as a semidefinite programming (SDP) problem \citep{GoringShadowSeperator,goring2011rotational}, 
\begin{equation}\label{eq:lambda_2Primal}
\begin{aligned}
& \underset{w_{ij}, \lambda_2, \mu}{\text{maximize}}
& & \lambda_2 \\
& \text{subject to}
& & \sum_{ij\in E_3}w_{ij}B_{ij}+L_0+\mu \boldsymbol{1} \boldsymbol{1}^T - \lambda_2 I\succeq 0 \\
& & & \sum_{ij\in E_3}w_{ij}= c \\
& & & w_{ij}\geq 0 
\end{aligned}
\end{equation}
where $L_0 = \sum_{ij\in E_1\cup E_2}B_{ij}$ is the Laplacian for the disjoint union of the layers. The semidefinite constraint ensures when the optimal solution $ (w_{ij}^*, \lambda_2^*, \mu^*) $ is attained, $\lambda_2^*$ is the smallest eigenvalue of $\sum_{ij\in E_3}w_{ij}^*B_{ij}+L_0+\mu^* \boldsymbol{1} \boldsymbol{1}^T $ or equivalently the second smallest eigenvalue of $\sum_{ij\in E_3}w_{ij}^*B_{ij}+L_0 $. For nonnegative budget $c$, the feasible set of the primal problem is not empty, and the primal problem attains its optimal solution \citep{shakeri2020designing}.

\subsection{Dual SDP and  equivalent graph embeddings}\label{sec:Dual}
The primal problem \eqref{sec:Primal} maximizes the diffusion speed over the multilayer network. We study the  Fiedler eigenspace to explore the route to this fastest spread over the network. The dual of \eqref{eq:lambda_2Primal} provides valuable information about the fastest diffusion mode. The dual  is obtained by the usual Lagrangian approach \citep{Boyd96SDP, boyd2004convex}:
\begin{equation}\label{lambda_2Dual}
\begin{aligned}
& \underset{\xi\in\mathbb{R}, X\in\mathbb{R}^{n\times n}}{\text{maximize}}
& & c\xi -\langle X, L_0\rangle  \\
& \text{subject to}
& & \langle X,I\rangle = 1 \\
& & & \langle X,ee^T\rangle = 0\\
& & & \langle X,B_{ij}\rangle \leq -\xi ~~\quad\forall\lbrace i,j\rbrace\in E_3\\
& & & X\succeq 0 
\end{aligned}
\end{equation}
where  $\langle X, L_0\rangle =\Tr(L_0^TX) = \sum_{\lbrace i,j\rbrace \in E_1\cup E_2}x_{ii}+x_{jj}-2x_{ij}$. The feasible set of the dual problem is not empty, and strong duality holds for the primal and dual problems \eqref{eq:lambda_2Primal} and \eqref{lambda_2Dual}\citep{shakeri2020designing}. Therefore the dual problem attains its optimal solution, and optimal values of the primal and dual problems are the same. 

\citet{Boyd2006FastestMP} use the Gram representation of the dual matrix $X = U^TU$, where $U\in \mathbb{R}^{n\times n}$ to rewrite \eqref{lambda_2Dual} as:
\begin{equation}\label{Lambda2Embedding}
\begin{aligned}
& \underset{\xi\in \mathbb{R}, u_i\in \mathbb{R}^n}{\text{maximize}}
& & c\xi -\sum_{\lbrace i,j\rbrace\in E_1\cup E_2}\| u_i-u_j \|^2 \\
& \text{subject to}
& & \sum_{i\in V}\|u_i\|^2 = 1\\
& & & \sum_{i\in V}u_i = 0\\
& & &  \| u_i - u_j\|^2\leq -\xi \quad \ \forall \lbrace i,j\rbrace\in E_3\\
\end{aligned}
\end{equation}

That is a realization of the graph in $\mathbb R^n$ such that the distances between nodes are minimized and the barycenter is in the origin, but since the sum of the squared norms equals one, not all nodes can be embedded in the origin. By the following lemma some main structural properties of optimal graph can be derived from the geometric dual problem.

\begin{lemma}\label{lem:EmbProj}
	The projections of optimal embedding onto (nonzero) one-dimensional subspaces yield eigenvectors for the algebraic connectivity.
\end{lemma} 
The following proposition is a result of Lemmas \ref{lem:Opt0} and \ref{lem:EmbProj}.  
\begin{proposition}\label{lem:clump}
	Assume the regularity condition \eqref{eq:RegulCond0} is feasible. For budget values up to the threshold $c^*$, $c\leq c^*$, the optimal solution of the embedding problem is given as
	\begin{equation}\label{eq:UniformEmbed}
	\begin{aligned}
	u_i^*=\begin{cases}
	mh & \text{if}\ i \in V_1 \\ 
	-nh, & \text{if}\ i \in V_2
	\end{cases}
	\end{aligned}
	\end{equation} 
	where $\boldsymbol{h}=\langle h\rangle$ is a one-dimensional subspace.
\end{proposition}
The embedding \eqref{eq:UniformEmbed} implies each layer clumps together at the opposite sides with respect to the other layer, while distanced from the origin inversely proportional to the number of its nodes. In this case, the Fiedler cut distinguishes the individual layers \citep{Gregorio2014algebraic}. The condition in \eqref{eq:UniformEmbed} is similar to the momentum balance condition \citep
{Shames1996} of two masses attached to opposite sides of a rotating rigid uniform rod. 
%rotating about a pivot, here the origin, where the masses are installed on distances inversely proportional to their values

Lemma \ref{lem:EmbProj} is a result of the complementary condition in the primal-dual formulation \citep{boyd2004convex}, and is first established by \citet{Boyd2006FastestMP} for single-layer networks and  extended to multilayer networks by \citet{shakeri2020designing}. It relates the solution of the dual problem to the eigenspace of maximum $\lambda_2$ and helps understanding how to improve diffusion speed \citep{GomezDiffusion2013}, and how synchronization \citep{Strogatz2001, Arena2008Sync} can happen faster. 
This is particularly of interest in multilayer networks where processes show  richer dynamics \citep{Buldyrev2010, GomezDiffusion2013, Radicchi2013, Zhang2015Explosive}. 
%Dual problem shed light on diffusion phases in and between layers; for example, Proposition \ref{lem:clump} indicates an intralayer optimal diffusion mode happens within the layers for  $c\leq c^*$. 
%We will observe more diverse behavior in next sections, including the situation when optimal diffusion in one specific layer is prominently through interlinks while through intralinks in the other layer. 

%\citet{SusannaThesis} establishes a reformulation of the graph realization problem that is interpreted as maps of eigenvectors to the second smallest eigenvalue of the Laplacian.
%, while, conversely, the Fiedler vector of solution to the primal problem yields the corresponding optimal graph embedding.
\begin{remark} \label{rem:EmbDim}
	The  multiplicity of $\lambda_2(L)$ sets an upper-bound on the dimension of realization \citep{Helmberg2010}. This can be understood from Proposition \ref{lem:EmbProj} and  the dimension of the eigenspace corresponding to $\lambda_2$.
	%This situation has been depicted for two BA networks in Figure \ref{fig:EmbDim}.
\end{remark}

For connected single-layer networks, we scale the weights in \eqref{lambda_2Dual} by $c\lambda_2\neq0$ and obtain a scaled version of the primal-dual problem for multilayer networks \citep{GoringShadowSeperator}
\begin{equation}\label{eq:PrimalScaled}
\begin{aligned}
& \underset{\hat{w}_{i,j}\in\mathbb{R}^{E_3}}{\text{minimize}}
& & \sum_{\lbrace i,j\rbrace\in E_3}\hat{w}_{ij} \\
& \text{subject to}
& & c\sum_{i,j\in E_3}\hat{w}_{ij}B_{ij}+(\sum_{\lbrace i,j\rbrace\in E_3}\hat{w}_{ij})L_0+\hat{\mu} \boldsymbol{e} \boldsymbol{e}^T - I\succeq 0 \\
& & & \hat{w}_{ij}\geq 0 \ \ \forall\lbrace i,j\rbrace\in E_3
\end{aligned}
\end{equation}
where $ \hat{w}_{ij} =\frac{w_{ij}}{c\lambda_2} $ and the scaled dual (embedding) problem \eqref{Lambda2Embedding} is written as 
\begin{equation}\label{eq:EmbeddingScaled}
\begin{aligned}
& \underset{\hat{u}_i\in\mathbb{R}^n}{\text{maximize}}
& & \sum_{i\in V}\| \hat{u}_i \|^2 \\
& \text{subject to}
& & c\| \hat{u}_i - \hat{u}_j\|^2+\sum_{\lbrace k,l\rbrace \in E_1}\| \hat{u}_k - \hat{u}_l\|^2+\\
& & &  \sum_{\lbrace k,l\rbrace \in E_2}\| \hat{u}_k - \hat{u}_l\|^2\leq 1 \ \ \forall\lbrace i,j\rbrace\in E_3\\
& & & \sum_{i\in V}\hat{u}_i= 0
%& & & \hat{U}\succeq 0
\end{aligned}
\end{equation}
\citet{GoringShadowSeperator} used the scaled primal-dual formulation to   prove the Separator-Shadow Theorem. Moreover, \citet{goring2011rotational} introduced the rotational dimension of a graph which is the maximal minimum dimension of an optimal graph realization, and proved it is bounded based on the  graph tree-width \citep{Diestel2017}. The importance of these theorems is that they establish the existence of low dimensional optimal realizations. %\citet{Shakeri2019} investigated graph realization structural properties in multiplex networks. %Moreover, the geometric dual problem discussed in this subsection holds interesting physical interpretations. %In particular, \citet{Boyd2006FastestMP} and \citet{GoringShadowSeperator} suggest a physical interpretation of embedding problem in single layer networks by considering each node as a point mass and each edge as a light rope that connects the points, and \citet{Shakeri2019} extend this interpretation under multiplex setting. \\

\section{Multilayer networks with all-pairs interconnection possibility}\label{sec:all-pair}
In this section, we remove all topological constraints and allow all-to-all interconnection 
in  multilayer networks.
Firstly, we study the spectral properties for identical weights assigned to all interlinks and secondly, investigate the primal and dual problems. Furthermore,  we explore the conditions under which the optimal weights are nonuniform and illustrate that the optimal strategy assigns zero weights to some edges despite possible all-pairs interconnections. 
Several results of this section emphasize the role of algebraic connectivity divided by number of nodes in a layer, which we call specific algebraic connectivity, and Fiedler vectors of subgraphs in optimum diffusion determination. 
\subsection{Uniform weights}\label{Sec:UnifWeight}
Considering uniform  weights for the interlinks, $w_{ij}=c/nm$ and $W=\frac{c}{nm}J$ with $J$ the ${n\times m}$ all ones matrix, results in the following eigenvalue problem 
\[
L\begin{bmatrix} u \\ v \end{bmatrix}=\lambda\begin{bmatrix} u \\ v \end{bmatrix}
\]
with $n-1$ solutions:
\begin{equation}\label{eq:LamUnif1}
\begin{gathered}
\lambda=\lambda_i^{(1)}+\frac{c}{n}, \ \ \ \begin{bmatrix} u \\ v \end{bmatrix}=\begin{bmatrix} u_i^{(1)} \\ 0 \end{bmatrix}
\end{gathered}
\end{equation}
where $\lambda_i^{(1)}$, $i=2,...,n$, are nonzero eigenvalues of $L_1$ and $u_i^{(1)}$ are the corresponding eigenvectors. 
Indeed, for uniform weights $W=\frac{c}{nm}J=\frac{c}{nm}\boldsymbol{1}_n\boldsymbol{1}_m^T$ and by the Laplacian \eqref{eq:LaplacMatrix} we have 
\[
L\begin{bmatrix} u_i^{(1)} \\ 0 \end{bmatrix}=\begin{bmatrix} L_1+\frac{c}{n}I_n &&& -\frac{c}{nm}\boldsymbol{1}_n\boldsymbol{1}_m^T \\ -\frac{c}{nm}\boldsymbol{1}_m\boldsymbol{1}_n^T &&& L_2+\frac{c}{m}I_m \end{bmatrix}\begin{bmatrix} u_i^{(1)} \\ 0\end{bmatrix}=\begin{bmatrix} L_1u_i^{(1)}+\frac{c}{n}u_i^{(1)} \\ -\frac{c}{nm}\boldsymbol{1}_m\boldsymbol{1}_n^Tu_i^{(1)}\end{bmatrix}=\left(\lambda_i^{(1)}+\frac{c}{n}\right)\begin{bmatrix} u_i^{(1)} \\ 0 \end{bmatrix}
\]
where $I_n$ and $I_m$ are respectively $n\times n$ and $m\times m$ identity matrices and we know that $\boldsymbol{1}_n^Tu_i^{(1)}=0$ for any eigenvector $u_i^{(1)}$ associated with a nonzero Laplacian eigenvalue in a connected graph. Similarly, $m-1$ nonzero of $L$ are given as
\begin{equation}\label{eq:LamUnif2}
\begin{gathered}
\lambda=\lambda_j^{(2)}+\frac{c}{m}, \ \ \ \begin{bmatrix} u \\ v \end{bmatrix}=\begin{bmatrix} 0 \\ v_j^{(2)} \end{bmatrix}
\end{gathered}
\end{equation}
where $\lambda_j^{(2)}$ is the $j$-th eigenvector of $L_2$ and $v_j^{(2)}$ the corresponding eigenvector. 

The remaining nonzero eigenvalue with its corresponding eigenvector are
\begin{equation}\label{eq:LamUnif0}
\begin{gathered}
\lambda=\left(\frac{1}{n}+\frac{1}{m}\right)c, \ \ \ \begin{bmatrix} u \\ v \end{bmatrix}=\frac{1}{\sqrt{2nm}}\begin{bmatrix} m\boldsymbol{1}_n \\ -n\boldsymbol{1}_m \end{bmatrix}
\end{gathered}
\end{equation}
hence, all eigenvalues increase linearly with $c$. The second largest eigenvalue is obtained as
\begin{equation}\label{lambda2_unif}
\begin{aligned}
\lambda_2\left[L\right]=\text{min}\left[\left(\frac{1}{n}+\frac{1}{m}\right)c, \ \lambda_2^{(1)}+\frac{c}{n}, \ \lambda_2^{(2)}+\frac{c}{m}\right]
\end{aligned}
\end{equation}
Varying $c$ results in  transitions  among the three linear functions in \eqref{lambda2_unif}. Figure \ref{fig:UniformCase1} illustrates a case, where $n>m$, $\lambda_2\left[L_1\right]>\lambda_2\left[L_2\right]$, and $\lambda_2\left[L_1\right]/n>\lambda_2\left[L_2\right]/m$, with two transitions occurring in
\begin{equation}\label{eq:thresholds}
\begin{aligned}
c^*=n\lambda_2^{(2)}, \ \ \ c^{**}=\left(\frac{1}{m}-\frac{1}{n}\right)^{-1}\left(\lambda_2^{(1)}-\lambda_2^{(2)}\right)
\end{aligned}
\end{equation}

Figure \ref{fig:UniformWeights} also illustrates Cases 2 and 3 where each holds only one transition. The transition in Case 3 is $c^*$ in \eqref{eq:thresholds}, and the only transition in Case 2 is
\begin{equation}\label{eq:threshold2}
\begin{aligned}
c^*=m\lambda_2^{(1)}
\end{aligned}
\end{equation} 

For the special case of subgraphs with the same number of nodes, $n=m$, the algebraic connectivity before $c^*$ is $2c/n$, and after $c^*$  is equal to $\lambda_2^{(0)}+c/n$ where $\lambda_2^{(0)}$ denotes the minimum algebraic connectivity of subgraphs, $\lambda_2^{(0)}=\text{min}(\lambda_2^{(1)},\lambda_2^{(2)})$. The threshold is $c^*=n\lambda_2^{(0)}$, in this special case. \\
  
\begin{figure*}
	\centering
	\subfloat[\label{fig:UniformCase1}]{\includegraphics[clip,width=.3\columnwidth]{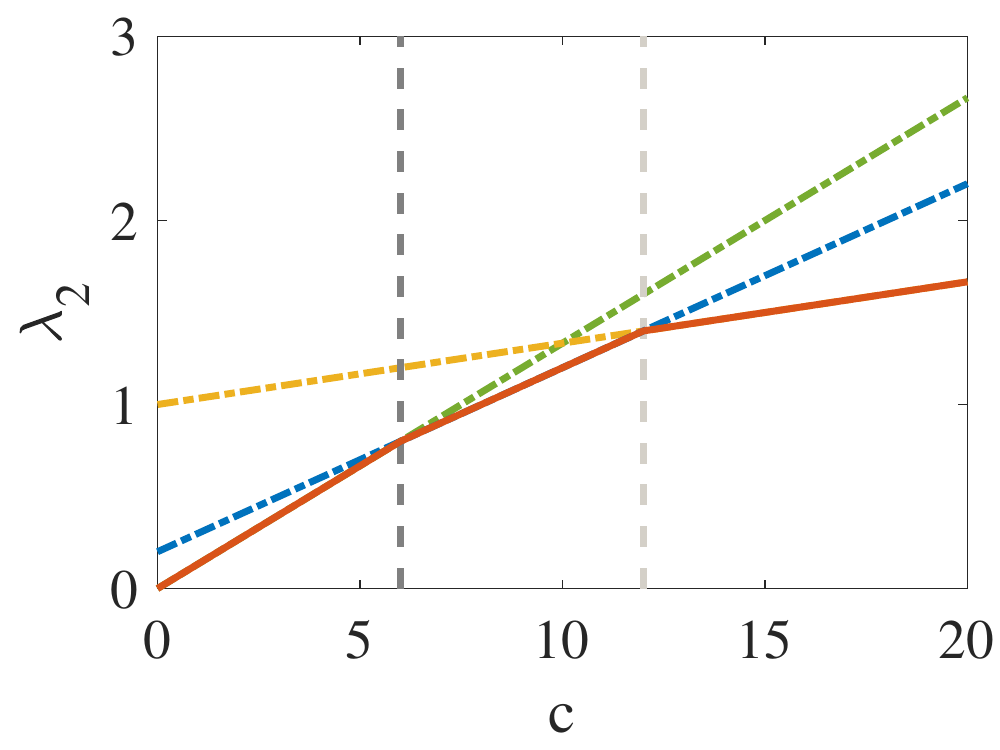}} 
	\subfloat[\label{fig:UniformCase2}]{\includegraphics[clip,width=.3\columnwidth]{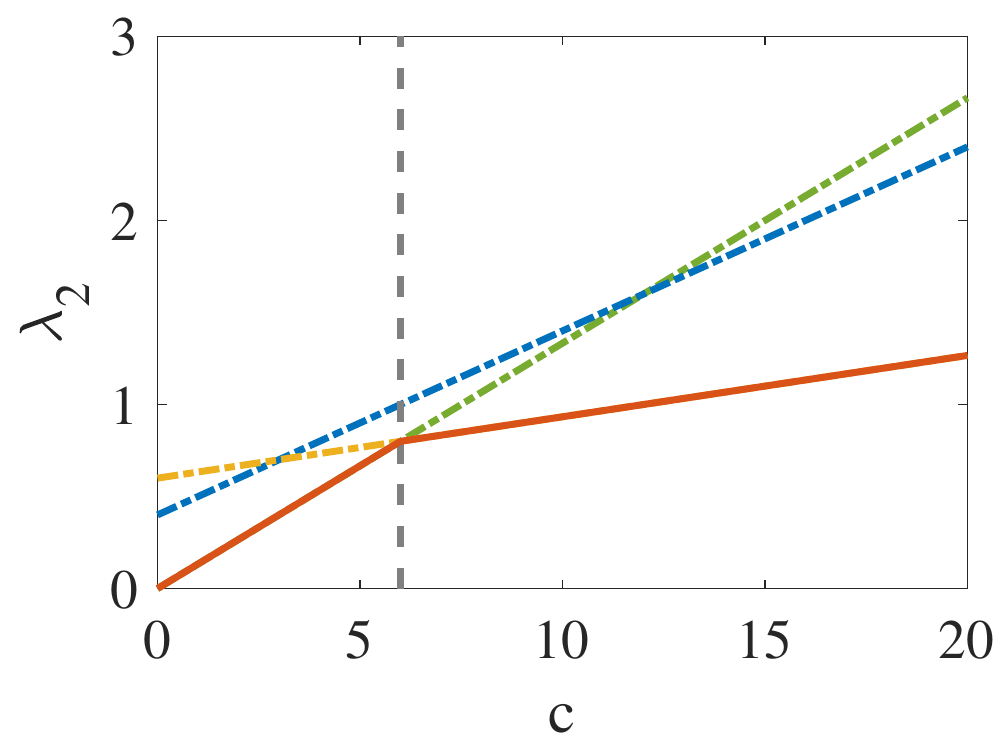}}
	\subfloat[\label{fig:UniformCase3}]{\includegraphics[clip,width=.3\columnwidth]{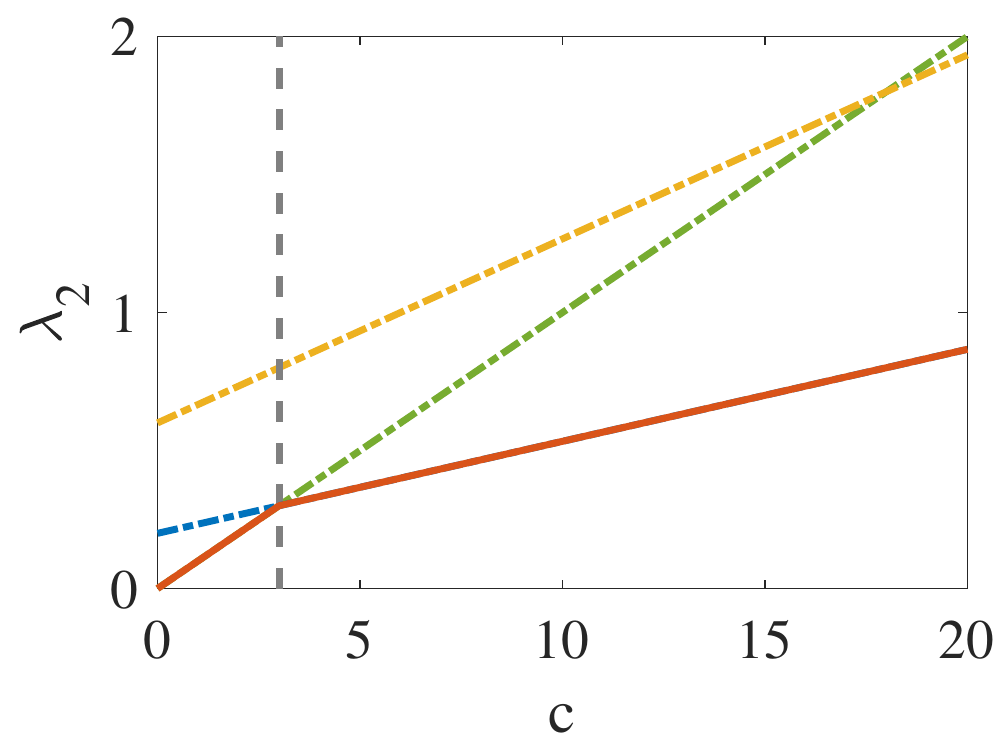}}\\
	\subfloat[\label{}]{\includegraphics[clip,width=.3\columnwidth]{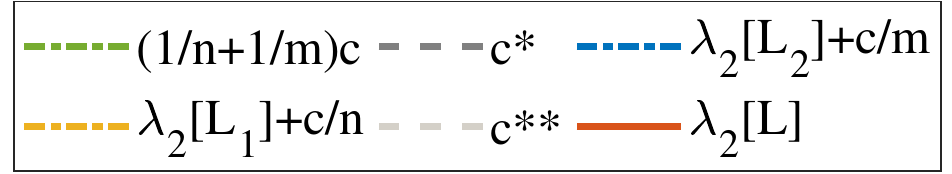}}
	\caption{Algebraic connectivity of supra-Laplacian $L$ as function of total budget $c$ for uniform weight distribution. (a) Case 1: $n>m$,  $\lambda_2\left[L_1\right]/n>\lambda_2\left[L_2\right]/m$ (thus $\lambda_2\left[L_1\right]>\lambda_2\left[L_2\right]$), (b) Case 2: $n>m$, $\lambda_2\left[L_1\right]>\lambda_2\left[L_2\right]$,  $\lambda_2\left[L_1\right]/n<\lambda_2\left[L_2\right]/m$, (c) Case 3 : $n<m$, $\lambda_2\left[L_1\right]>\lambda_2\left[L_2\right]$ (thus $\lambda_2\left[L_1\right]/n>\lambda_2\left[L_2\right]/m$).}
	\label{fig:UniformWeights} 
\end{figure*}

In all three cases illustrated in Figure \ref{fig:UniformWeights}, individual network components play no role on supra-Laplacian algebraic connectivity for  $c\leq c^*$. By increasing $c$ beyond the threshold $c> c^*$, the algebraic connectivities of  individual layers $G_1$ and $G_2$ per unit node, i.e. $\lambda_2\left[L_1\right]/n$ and $\lambda_2\left[L_2\right]/m$ are the main parameters characterizing the algebraic connectivity of the whole network.
We call the algebraic connectivity per unit node the \textit{specific algebraic connectivity}. The term ``specific'' is motivated by its use in thermodynamics \citep{Sonntag2009} where it refers to a quantity per unit mass. A specific quantity is an intensive property because it does not depend on substance amount. Here, for instance, a complete graph with $n$ nodes has algebraic connectivity equal to $n$, while its specific algebraic connectivity is unity, thus independent of graph number of nodes $n$. As a feature of three cases of Figure \ref{fig:UniformWeights}, the subgraph with larger specific algebraic connectivity determines the supra-Laplacian algebraic connectivity only in Case 1 when $c>c^{**}$, and under all other circumstances of $c>c^*$ it is the subgraph with smaller specific connectivity that determines the supra-Laplacian algebraic connectivity.
%	; or equivalently, the subgraph with less nodes determines the supra-Laplacian algebraic connectivity only in Case 1 when $c^*<c\leq c^{**}$, and under all other circumstances of $c>c^*$ it is the subgraph with more nodes that determines the supra-Laplacian algebraic connectivity. We will see more use of specific algebraic connectivity for describing the results of a geometric dual optimization problem in next sections.
%For intermediate budget $c^*< c\leq c^{**}$, the supra-Laplacian algebraic connectivity originates from algebraic connectivity of subgraph with less nodes, which is replaced with subgraph of more nodes for $c>c^{**}$. On the other hand, when the subgraph with more nodes holds the smaller specific algebraic connectivity in Figures \ref{fig:UniformCase2} and \ref{fig:UniformCase3}, the supra-Laplacian algebraic connectivity for $c>c^*$ is rooted in this subgraph algebraic connectivity. As conclusion, 

\emph{Super-diffusion} happens when diffusion in the interconnected network spreads faster than in each individual network if isolated. In multilayer networks with one-to-one interconnection, super-diffusion is possible only if the algebraic connectivity of the average Laplacian $L_{ave}=\frac{1}{2}\left(L_1+L_2\right)$ is greater than the algebraic connectivity values of individual networks \citep{gomez2013diffusion}. This is possible provided the Fiedler vectors of $G_1$ and $G_2$ are far from being parallel (close-to-orthogonal) \citep{sahneh2014exact}. In contrast, in all-to-all interconnection configurations, since algebraic connectivity increases linearly with $c$, super-diffusion always occurs for sufficiently large budgets, and no restriction is imposed on Fiedler vectors of individual networks. In fact, having close-to-orthogonal Fiedler vectors in one-to-one interconnected networks means that links of $G_2$ connect those nodes that are far from each other in $G_1$, and vice versa \citep{sahneh2014exact}. Since this condition is satisfied in all-to-all interconnection regime, super-diffusion is always feasible. However, to explore  whether  super-diffusion is possible before the threshold $c^*$, we investigate if the following inequality is satisfied: $$\lambda_2\left[L\right]>\text{max}\left(\lambda_2^{\left(1\right)},\lambda_2^{\left(2\right)}\right)$$ 
Figure \ref{fig:UniformWeights} shows that in cases 1 and 3,
\begin{equation}\label{eq:SupDifuCase1}
\begin{gathered}
\frac{\lambda_2^{\left(2\right)}}{m}<\frac{\lambda_2^{\left(1\right)}}{n}<\left(\frac{1}{n}+\frac{1}{m}\right)\lambda_2^{\left(2\right)}
\end{gathered}
\end{equation}
and for Case 2, 
\begin{equation}\label{eq:SupDifuCase2}
\begin{gathered}
\frac{\lambda_2^{\left(1\right)}}{n}<\frac{\lambda_2^{\left(2\right)}}{m}<\left(\frac{1}{n}+\frac{1}{m}\right)\lambda_2^{\left(1\right)}
\end{gathered}
\end{equation}

The conditions \eqref{eq:SupDifuCase1} and \eqref{eq:SupDifuCase2} set restriction on the difference between the algebraic connectivity values of the network components, and thus  correlation between structural properties of the layers is required to allow super-diffusion for  $c\le c^*$. 
Presence of significant difference between the algebraic connectivity values of individual networks postpones super-diffusion until large interconnection strengths. Conversely, for close values super-diffusion can occur for values of $c<c^*$, where the network components function distinctly, thus allowing advantage of interconnections while preserving the autonomy of each subsystem \citep{sahneh2014exact}. \\

\subsection{Maximum algebraic connectivity: a perturbation analysis}\label{Sec:Perturb}
Lemmas \ref{lem:lam2bound1} and \ref{lem:UpBound} explain that the maximum algebraic connectivity is upper-bounded and the upperbound \eqref{eq:lam2bound}  is attainable subject to the regularity conditions \eqref{eq:RegulCond0}. 
Section \ref{Sec:UnifWeight} shows that with uniform weights--which satisfy the regularity conditions--the algebraic connectivity can reach its upper-bound \eqref{eq:lam2bound} in $c\leq c^*$. 
Therefore, for $c\le c^*$, the uniform weight distribution is the solution of maximum eigenvalue problem \eqref{eq:maxLam}. Although,  Lemma \ref{lem:Opt0} gives the solution of maximum algebraic connectivity analytically when $c\leq c^*$ and regularity is feasible, we resort to the solution of primal and dual SDP formulations discussed in Section \ref{sec:SDP} for other situations. To this end, we conduct a perturbation analysis to shed some light on how the maximum algebraic connectivity behaves after $c>c^*$ and then investigate the SDP results for all-to-all interconnection possibility in sections \ref{sec:PrimAll} and \ref{sec:DuAll}.

For the perturbation analysis, we consider a nominal budget $c_0$ and the associated eigenvalue problem  $L_0x_0=\lambda_0x_0$ where $L_0$ is the nominal Laplacian with an eigenvalue $\lambda_0$ and corresponding eigenvector $x_0$. Then, consider the perturbed quantities $c=c_0+\epsilon c', L=L_0+\epsilon L', \lambda=\lambda_0+\epsilon \lambda', x=x_0+\epsilon x'$, where $\epsilon$ is sufficiently small and the quantities with prime show the levels of perturbations. Now, the eigenvalue problem $Lx=\lambda x$ is written as
\begin{equation}\label{eq:PerturbEigenProb}
\begin{gathered}
\left(L_0+\epsilon L'\right)\left(x_0+\epsilon x'\right)=\left(\lambda_0+\epsilon \lambda'\right)\left(x_0+\epsilon x'\right)
\end{gathered}
\end{equation} 
Equating the coefficients of $\epsilon$ in both sides of \eqref{eq:PerturbEigenProb}, we have
\begin{equation}\label{eq:PerturbEps1}
\begin{gathered}
L_0x'+L'x_0=\lambda_0x'+\lambda'x_0
\end{gathered}
\end{equation}
Inner product of \eqref{eq:PerturbEps1} by $x_0$ yields the following expression for the perturbed value $\lambda'$
\begin{equation}\label{eq:PerturbLam}
\begin{gathered}
\lambda'=\frac{x_0^TL'x_0}{\|x_0\|^2}
\end{gathered}
\end{equation}

Figure \ref{fig:UniformCase1} shows the above analysis for Case 1 at the threshold budget, $c_0=c^*$. We know that before $c^*$ the optimal weights are the uniform ones and that at threshold $c^*$ the second and third eigenvalues coalesce and the third eigenvalue (before $c^*$) becomes the algebraic connectivity (right after $c^*$). Therefore, we repeat the perturbation analysis for the third eigenvalue $\lambda_3$ at $c^*$; using the results of Section \ref{Sec:UnifWeight} for Case 1, for $c_0=c^*$,
\[
\lambda_0=\lambda_3=\lambda_2^{\left(2\right)}+\frac{c^*}{m}, x_0=\begin{bmatrix}
0 \\ v_2^{\left(2\right)}
\end{bmatrix}, L'=\begin{bmatrix}
\text{diag}\left(W'\boldsymbol{1}_m\right) & -W' \\
-{W'}^T & \text{diag}\left({W'}^T\boldsymbol{1}_n\right)
\end{bmatrix}.
\]
where $W'$ is the perturbed weight matrix due to the perturbed budget $c'$.
% and should be determined by optimal mechanism.
Substituting the above $x_0$ and $L_0$ into \eqref{eq:PerturbLam} and after normalizing the eigenvector $\|v_2^{\left(2\right)}\|=1$, we obtain the following expression for the increment in algebraic connectivity beyond $c^*$
\begin{equation}\label{eq:PerturbLam2}
\begin{gathered}
\lambda'_2={v_2^{\left(2\right)}}^T\text{diag}\left({W'}^T\boldsymbol{1}_n\right)v_2^{\left(2\right)}
\end{gathered}
\end{equation}
Considering $\lambda=\lambda_0+\epsilon\lambda'$, we can approximate the algebraic connectivity for budget values just above $c^*$. Therefore, maximizing the perturbed value \eqref{eq:PerturbLam2} can be regarded as an equivalent problem for maximizing the algebraic connectivity for sufficiently small increments of $c$ beyond $c^*$.

	We mention two highlights in maximizing \eqref{eq:PerturbLam2}; first, diagonal  elements of $\text{diag}\left({W'}^T\boldsymbol{1}_n\right)$ represent the total weight assigned to each node in Layer 2 and thus $\lambda'_2$ only depends on the interlink weights of nodes in Layer 2. 
%	Consequently, there is no dependence on interlayer weights assigned to nodes of Layer 1 when maximizing the increment $\lambda'_2$ to algebraic connectivity beyond $c^*$. 
	Thus, the optimization mechanism will make no difference between nodes in Layer 1, and continue allocating equal interlink weights to all nodes in this layer
%	\todo{see commented}
%	--similar to before $c^*$ and beyond $c^*$--
	and the nodes in Layer 1 still experience uniform total interlink weights. Second, the increment in algebraic connectivity in \eqref{eq:PerturbLam2} is correlated with the Fiedler vector of $L_2$. Consequently, maximizing \eqref{eq:PerturbLam2}, require assigning the weights in vector ${W'}^T\boldsymbol 1_n$ to nodes with larger $\left(v_2^{\left(2\right)}\left(i\right)\right)^2$, or  larger absolute value $|v_2^{\left(2\right)}\left(i\right)|$.\\ 

In summary, for sufficiently small increments of budget beyond $c^*$, while the optimal weights in the Layer with larger specific connectivity, remain uniform, the optimal weights in the other Layer  are positively correlated with Fiedler vector. However, this situation turns conversely for larger budget values above the second threshold $c^{**}$. This will be illustrated in next subsection. \\

We can get similar results for Cases 2 and 3 in Figures \ref{fig:UniformCase2} and \ref{fig:UniformCase3}. In particular for Case 2, note that for $c_0=c^*=m\lambda_2^{\left(1\right)}$ the third eigenvalue is $\lambda_2^{\left(1\right)}+\frac{c}{n}$. The corresponding eigenvector is $x_0=\begin{bmatrix}
u_2^{\left(1\right)} \\ 0
\end{bmatrix}$. Using the above perturbation approach, we get the following increment for algebraic connectivity just above $c^*$
\begin{equation}\label{eq:PerturbLam2Case2}
\begin{gathered}
\lambda'_2={u_2^{\left(1\right)}}^T\text{diag}\left(W'\boldsymbol{1}_m\right)u_2^{\left(1\right)}
\end{gathered}
\end{equation}
Likewise, by increasing the total budget $c$ beyond $c^*$ in Case 2, the optimal weight distribution for nodes in the layer with smaller specific connectivity (Layer 1) will be positively correlated with Fiedler vector, and optimal weights for nodes in layer with larger specific connectivity (Layer 2) will  remain uniform. For Case 3, we get the same relation given in \eqref{eq:PerturbLam2}.\\
%i.e. the larger $|u_2^{\left(1\right)}\left(i\right)|$ \textcolor{red}{the larger the total weight assigned to node $i$ in Layer 1;}
%\todo{repetition?}\textit{Consequently, in all three cases of Figure \ref{fig:UniformWeights}, by increasing $c$ beyond $c^*$, while the optimal weights in layer with larger specific connectivity will still remain uniform, the optimal weights in the other layer with smaller specific connectivity will be positively correlated with Fiedler vector components of this layer.} However, for Case 1, this situation turns conversely after a budget value about the second threshold $c^{**}$. This will be illustrated in next subsection.  \\

%Now we turn our attention to large budgets, $c\rightarrow\infty$. For this extremely large budget case, the interlayer links are influential enough to dominate the structural characteristics of individual layers \citep{GomezDiffusion2013, Ribalta2013Spectral}. In an all-pairs interconnection pattern, the mulrilayer network will be totally under the effect of interlayer links and other properties such as the intralayer links (all with unit weight) can hardly play significant role. In such conditions, different nodes with different intralinks will not be discriminated by an optimal strategy. Therefore (almost) uniform weight distribution is expected for extraordinarily large budgets and the maximum algebraic connectivity converges to that of uniform weights as $c$ goes to infinity (i.e. to the corresponding case in Figure \ref{fig:UniformWeights}).     

\subsection{Primal problem}\label{sec:PrimAll}

The SDP \eqref{eq:lambda_2Primal} can be  solved efficiently using convex solvers, e.g. CVX package by \citet{grant2009cvx}. Figure \ref{fig:Primal1} shows the optimal results for an example of Case 1. Figure \ref{fig:Primal1Lam} compares the maximum algebraic connectivity with uniform weights case.
%Although, for $c<c^*$ the optimal value for the algebraic connectivity can be  obtained via uniform weights, Figure \ref{fig:Primal1Lam} shows that this value can increase  after the threshold\todo{trivial?!}. 
In Figures \ref{fig:Primal1W1} and \ref{fig:Primal1W2}, we observe while the optimal weights assigned to nodes in Layer 1 remain uniform for budgets up to $c^{**}$, they start to become inhomogeneous in Layer 2 right above $c^*$. Here, Layer 1 holds the larger specific algebraic connectivity and Layer 2 holds
 the smaller one. Figures \ref{fig:Primal1W2Case1c5} and \ref{fig:Primal1W2Case1c20},  further clarify that, after $c^*$, the optimal weights in Layer 2 are positively correlated with Fiedler vector components of this layer--these results are in accordance with our perturbation results in Subsection \ref{Sec:Perturb}. \\

For $c>c^{**}$, Figure \ref{fig:Primal1W1} illustrates that weight distribution in Layer 1 becomes heterogeneous and many nodes experience zero interlink weights. Investigation of optimum weights in Figure \ref{fig:Primal1W1Case1c30} reveals that the nodes with no interlinks in Layer 1 are those with smallest components in Fiedler vector. 
	On the contrary, optimal weights in Layer 2 again become (almost) uniform for some budget $c>c^{**}$ (see Figure \ref{fig:Primal1W2}) where for strong coupling strength $c$, the intralinks in the layer with smaller specific algebraic connectivity (here Layer 2) lose their significance against very strong interlinks. This causes the nodes in this layer to be not effectively discriminated by an optimal mechanism, and uniform weights are expected.

	Figure \ref{fig:Primal1Case2} shows the SDP results for Case 2.
We observe that, while the weights in Layer 1 with smaller specific connectivity become inhomogeneous for  budgets $c>c^*$, they remain (almost) uniform in Layer 2 with larger specific connectivity. Figure \ref{fig:Primal1W1Case2} reports a positive correlation between the optimal weights in Layer 1 and the corresponding Fiedler vector components. Figures \ref{fig:Primal1Case3} and \ref{fig:Primal1W2Case3} show the results corresponding to Case 3. These results are the same of Case 2 by changing Layers 1 and 2.  \\
%In each subplot, the node undergoing the largest weight is that with largest absolute value in Fiedler vector, and the nodes with no interlinks, if any, are those associated with smallest absolute values in Fiedler vector. An investigation of Figure \ref{fig:Primal1W1Case2c10} illustrates that optimal weight distribution in Layer 1 for budget value just above $c^*$ still tends to preserve its uniformity before $c\leq c^*$, with the difference that here some few nodes experience smaller weights while the node with largest component in Fiedler vector undergoing the largest weight. This pattern dies away with increasing $c$ so that more inhomogeneous conditions are observed in Figures \ref{fig:Primal1W1Case2c20} and \ref{fig:Primal1W1Case2c50}.
\begin{figure*}
	\centering
	\subfloat[\label{fig:Primal1Lam}]{\includegraphics[clip,width=.3\columnwidth]{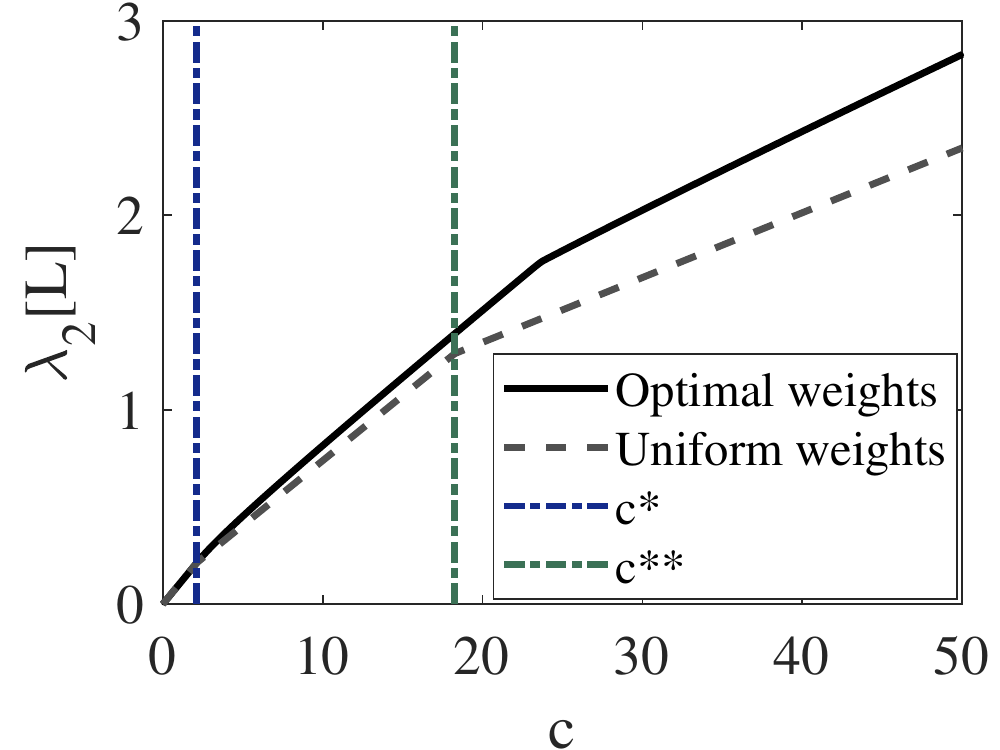}} \ \ \ \ \
	\subfloat[\label{fig:Primal1W1}]{\includegraphics[clip,width=.3\columnwidth]{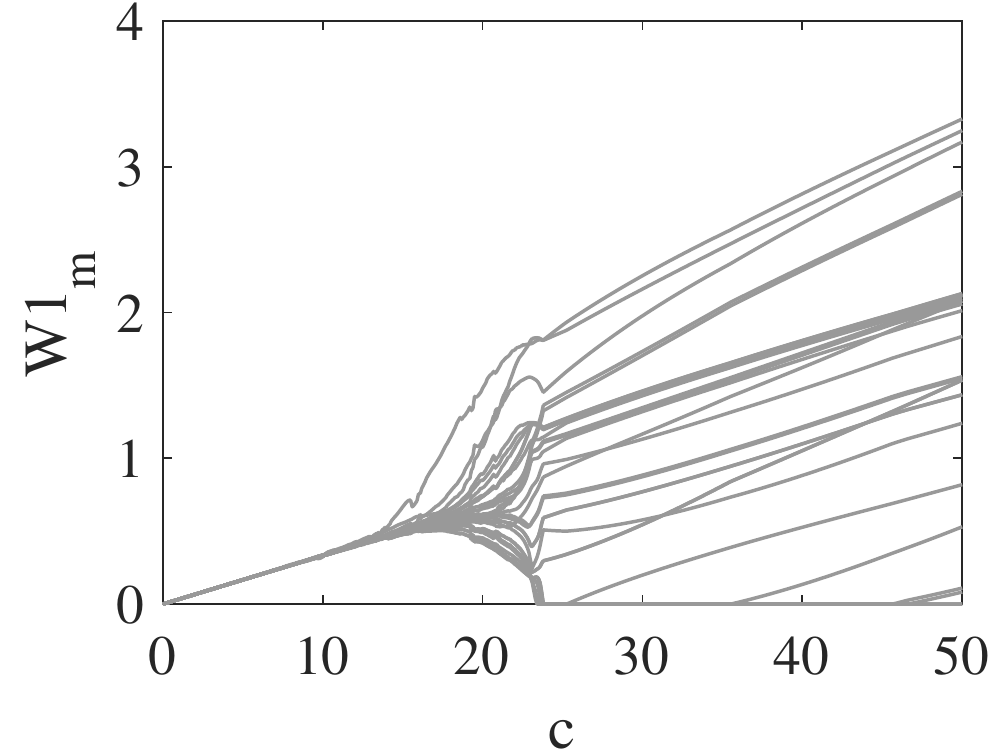}} \ \ \ \ \
	\subfloat[\label{fig:Primal1W2}]{\includegraphics[clip,width=.3\columnwidth]{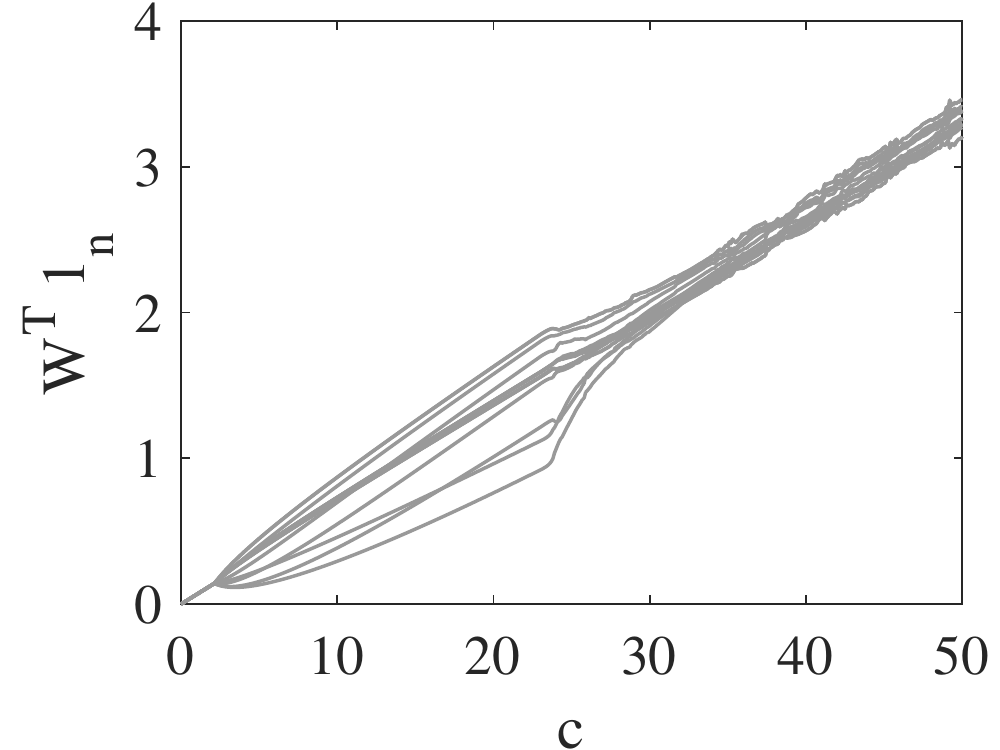}} 
	\caption{SDP results for an example of Case 1 in Figure \ref{fig:UniformWeights}: (a) Algebraic connectivity of supra-Laplacian $L$ as function of total budget $c$, (b) optimal interlayer weights assigned to nodes in Layer 1, and (c) optimal interlayer weights assigned to nodes in Layer 2, for two Geo networks with $n=30, m=15, \lambda_2^{(1)}=0.6798,\lambda_2^{(2)}=0.0712, c^*=2.1373, c^{**}=18.2554$.}
	\label{fig:Primal1} 
\end{figure*}
\begin{figure*}
	\centering
	\subfloat[\label{fig:Primal1W2Case1c5}]{\includegraphics[clip,width=.3\columnwidth]{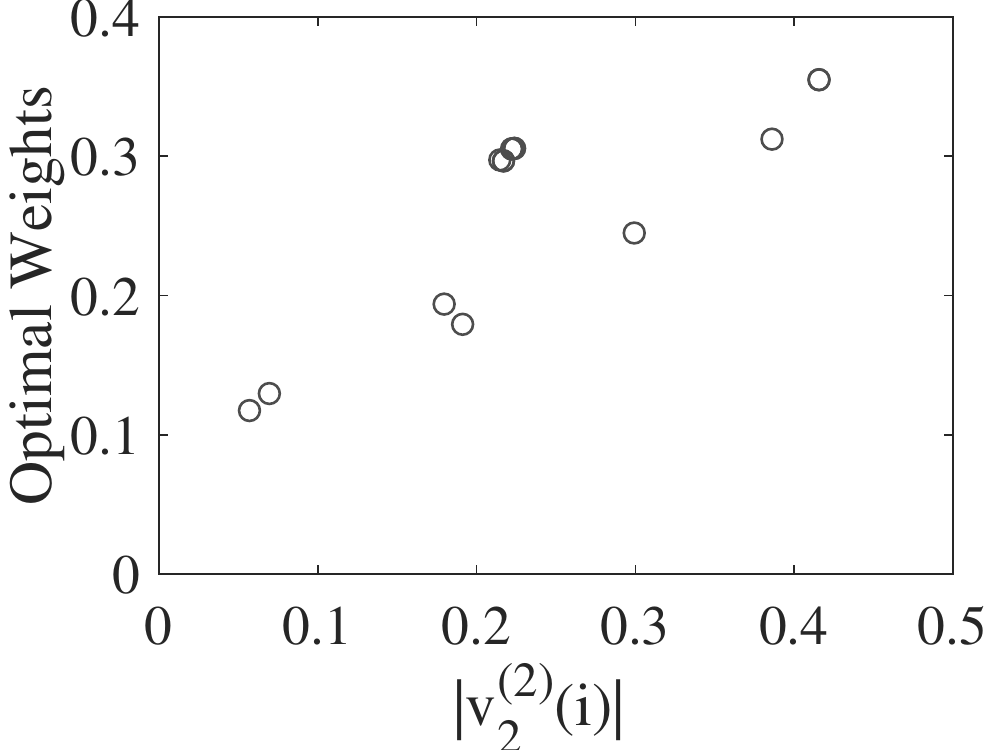}} \ \ \ \ \
	\subfloat[\label{fig:Primal1W2Case1c20}]{\includegraphics[clip,width=.3\columnwidth]{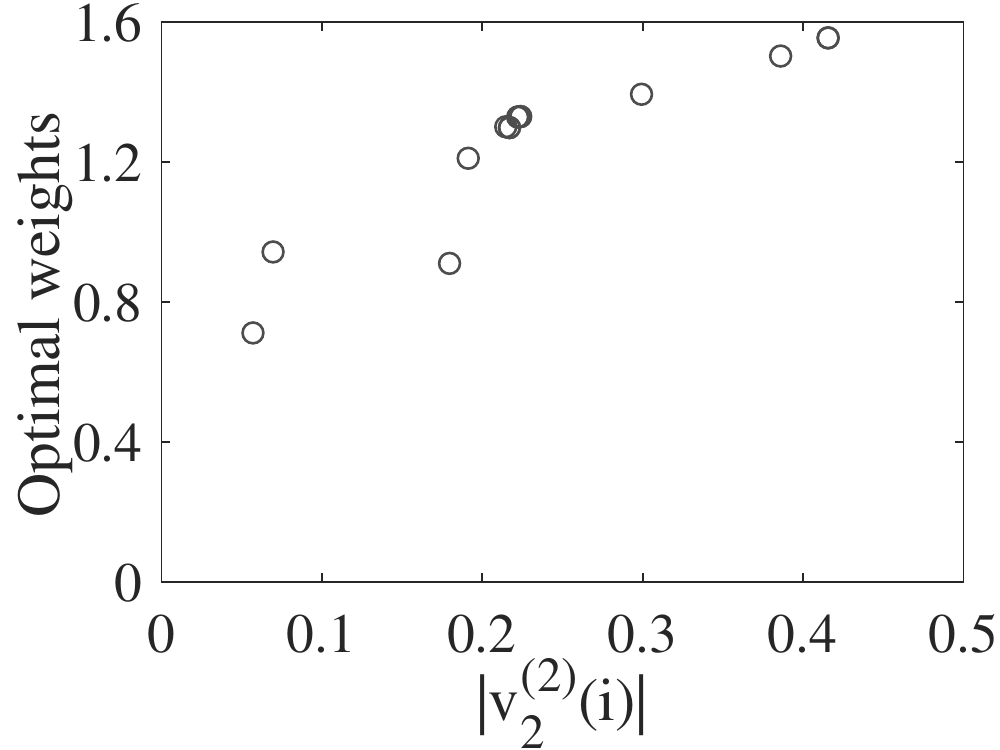}} \ \ \ \ \
	\subfloat[\label{fig:Primal1W1Case1c30}]{\includegraphics[clip,width=.3\columnwidth]{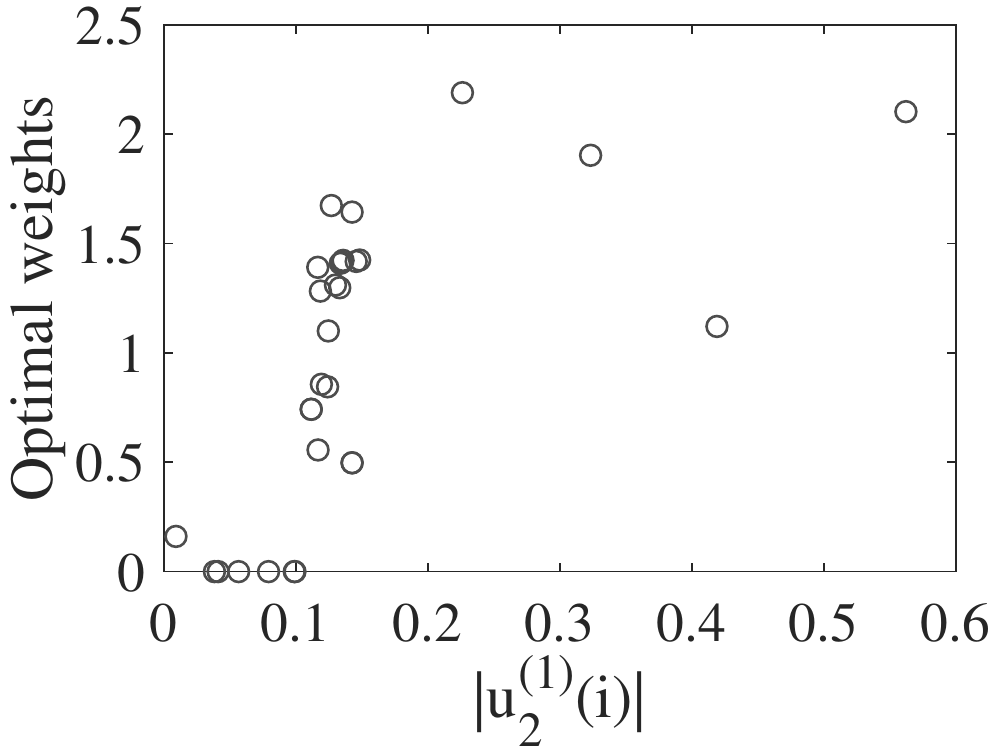}} 
	\caption{Optimal weights as function of Fiedler vector components corresponding to Figure \ref{fig:Primal1} in (a) Layer 2 for $c=5$, (b) Layer 2 for $c=20$, and (c) Layer 1 for $c=30$.}
	\label{fig:Primal1W2Case1} 
\end{figure*}
\begin{figure*}
	\centering
	\subfloat[\label{fig:Primal1LamCase2}]{\includegraphics[clip,width=.3\columnwidth]{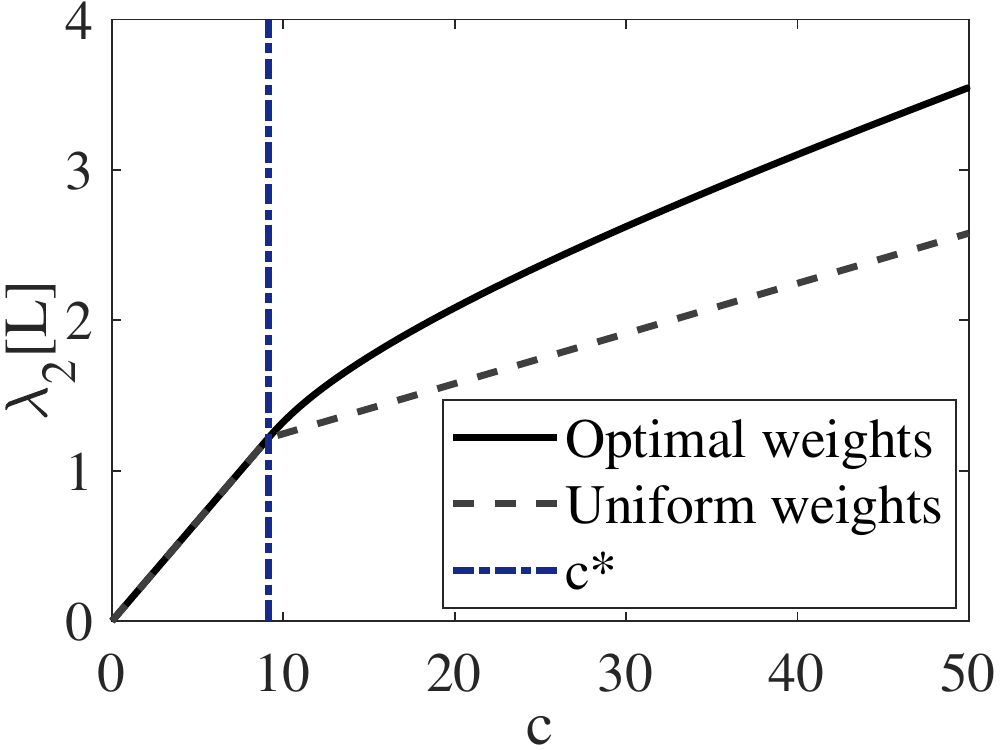}} \ \ \ \ \
	\subfloat[\label{fig:Primal1W1Case2}]{\includegraphics[clip,width=.3\columnwidth]{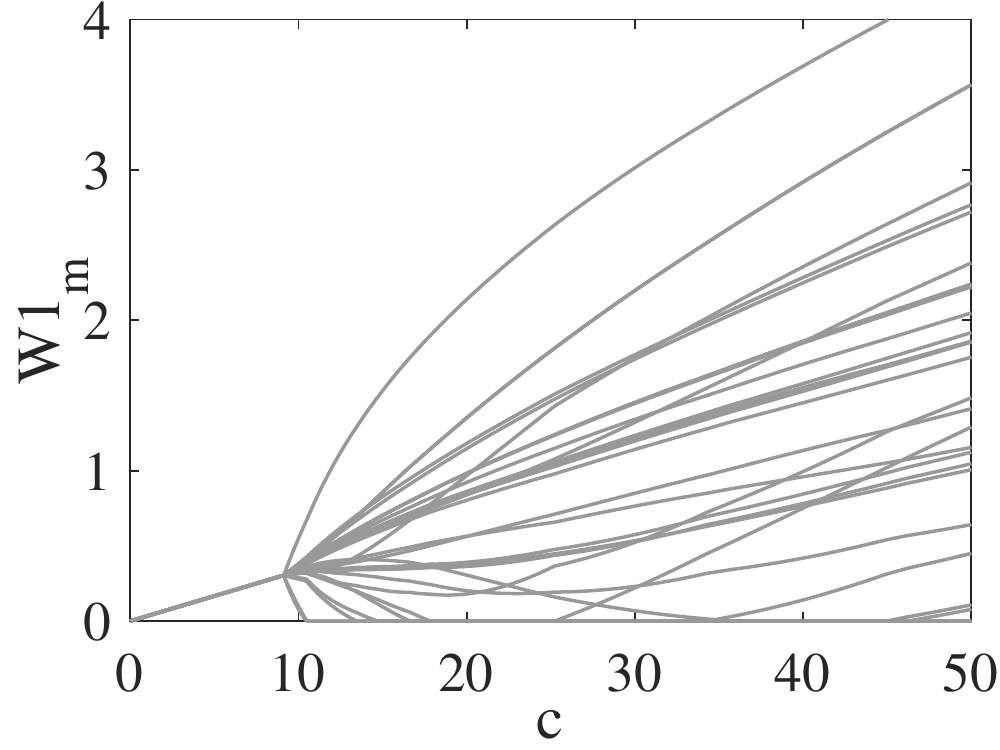}} \ \ \ \ \
	\subfloat[\label{fig:Primal1W2Case2}]{\includegraphics[clip,width=.3\columnwidth]{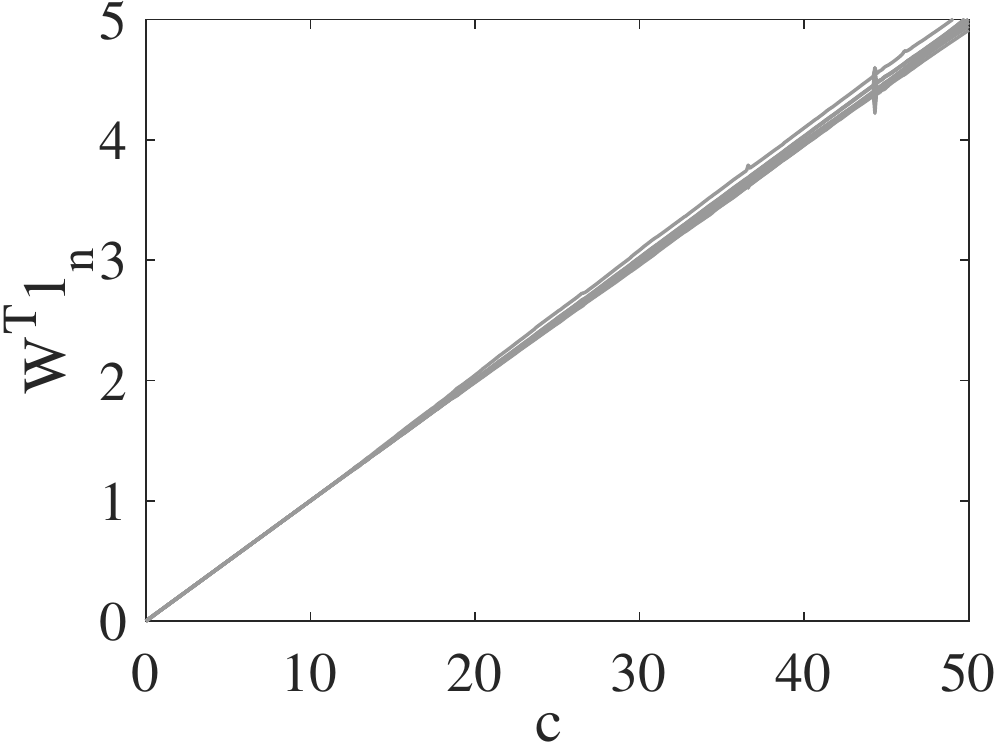}}
	\caption{SDP results for an example of Case 2 in Figure \ref{fig:UniformWeights}: (a) Algebraic connectivity of supra-Laplacian $L$ as function of total budget $c$, (b) optimal interlayer weights assigned to nodes in Layer 1, and (c) optimal interlayer weights assigned to nodes in Layer 2, for two Geo networks with $n=30, m=10, \lambda_2^{(1)}=0.9123,\lambda_2^{(2)}=0.6546$, $c^*=9.1235$.}
	\label{fig:Primal1Case2} 
\end{figure*}
\begin{figure*}
	\centering
	\subfloat[\label{fig:Primal1W1Case2c10}]{\includegraphics[clip,width=.3\columnwidth]{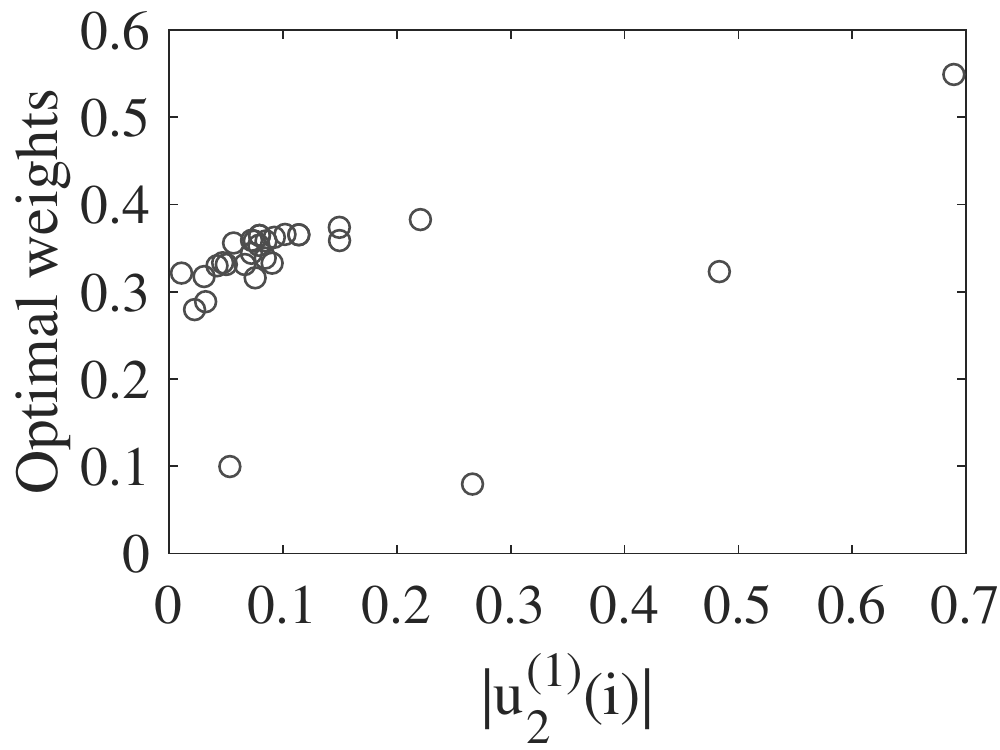}} \ \ \ \ \
	\subfloat[\label{fig:Primal1W1Case2c20}]{\includegraphics[clip,width=.3\columnwidth]{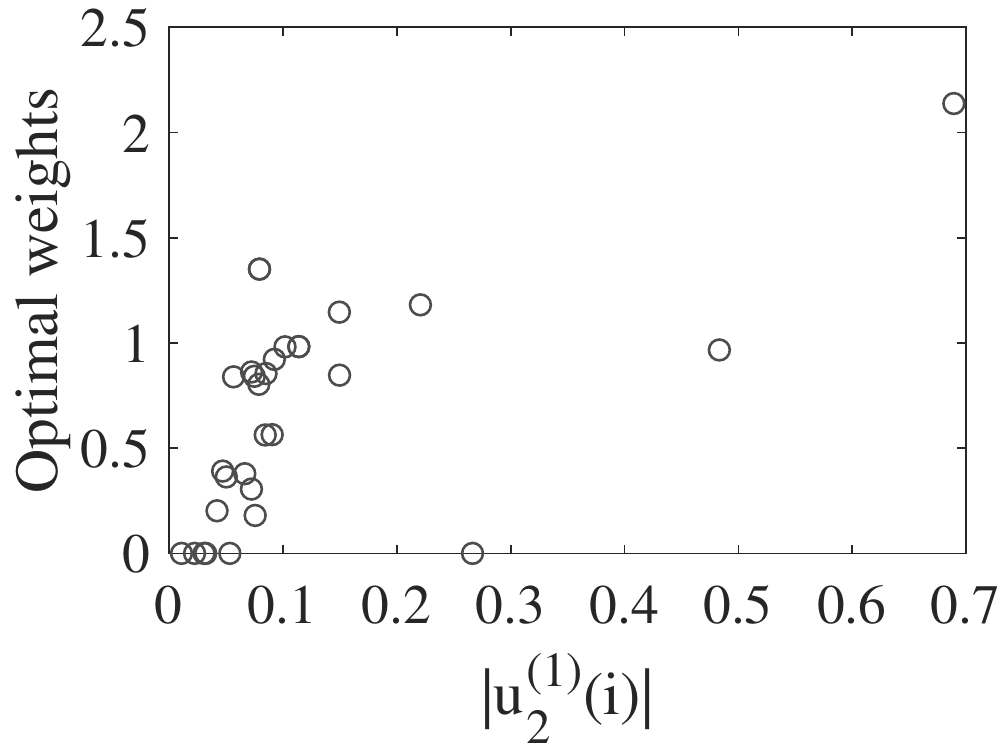}} \ \ \ \ \
	\subfloat[\label{fig:Primal1W1Case2c50}]{\includegraphics[clip,width=.3\columnwidth]{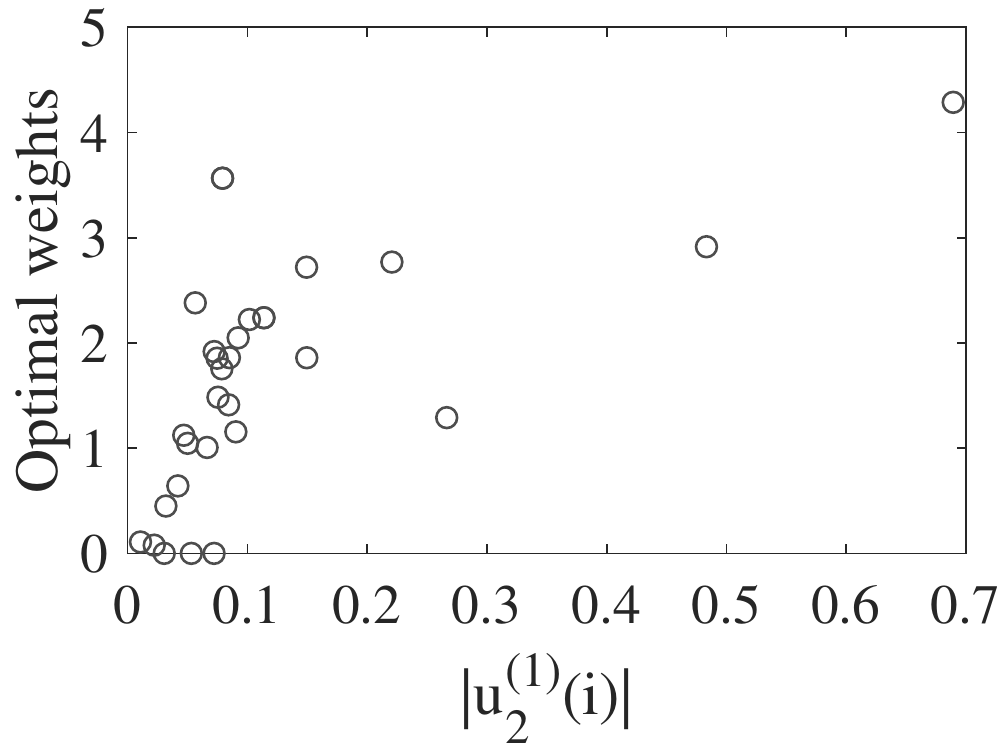}}
	\caption{Optimal weights in Layer 1 as function of Fiedler vector components corresponding to Figure \ref{fig:Primal1Case2} for (a) $c=10$, (b) $c=20$, and (c) $c=50$.}
	\label{fig:Primal1W1Case2} 
\end{figure*}
\begin{figure*}
	\centering
	\subfloat[\label{fig:Primal1LamCase3}]{\includegraphics[clip,width=.3\columnwidth]{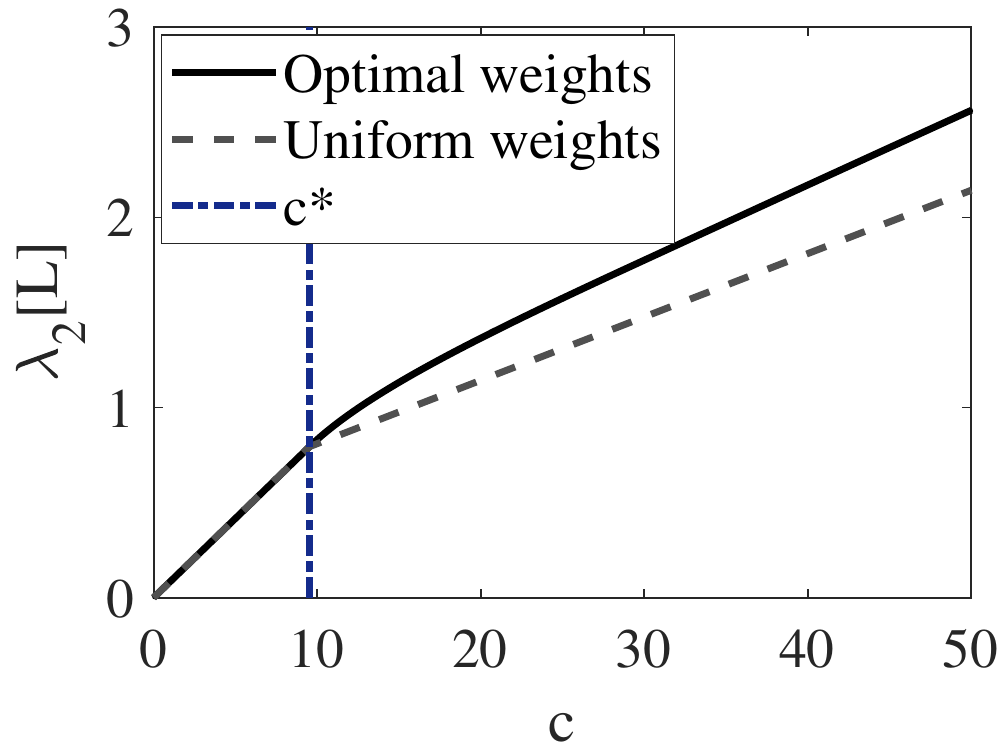}} \ \ \ \ \
	\subfloat[\label{fig:Primal1W1Case3}]{\includegraphics[clip,width=.3\columnwidth]{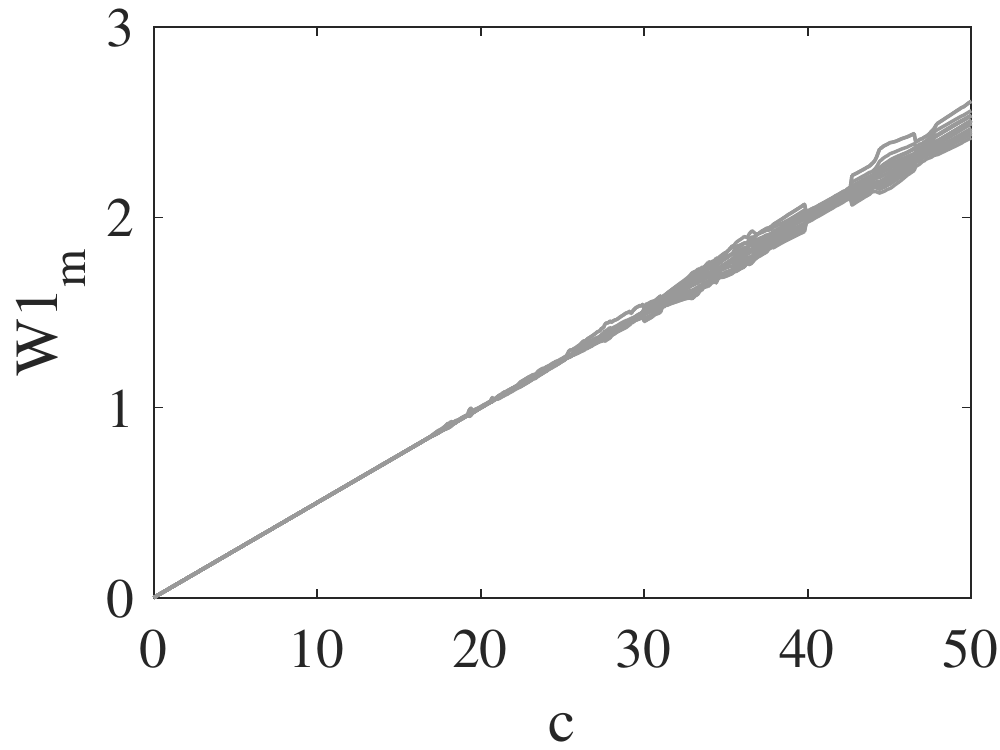}} \ \ \ \ \
	\subfloat[\label{fig:Primal1W2Case3}]{\includegraphics[clip,width=.3\columnwidth]{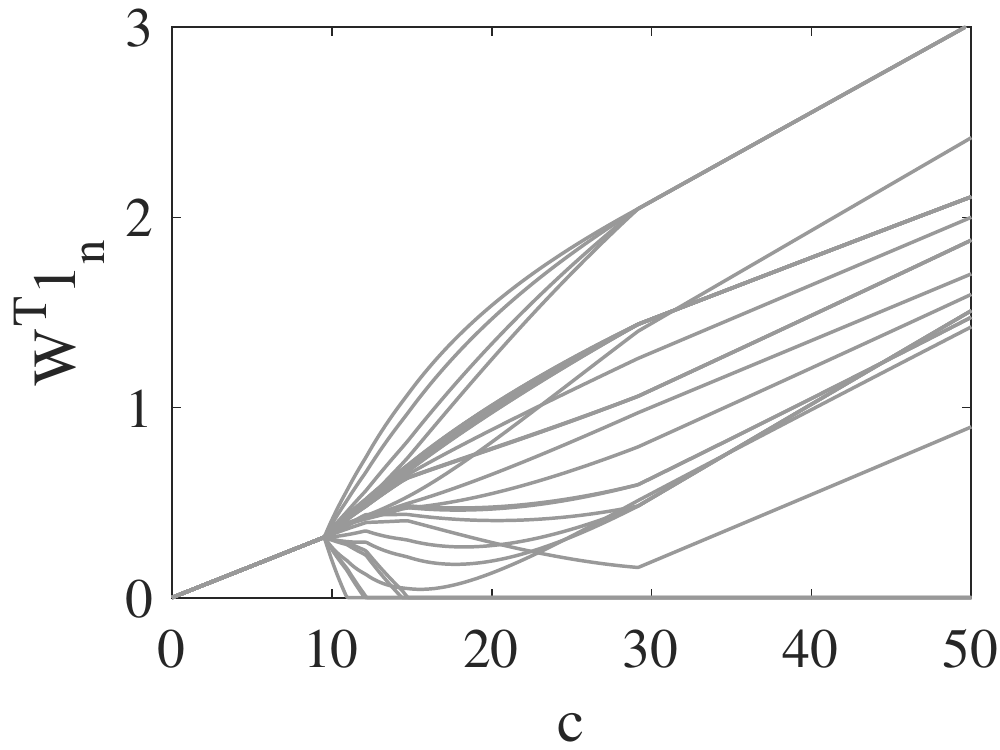}}
	\caption{SDP results for for an example of Case 3 in Figure \ref{fig:UniformWeights}: (a) Algebraic connectivity of supra-Laplacian $L$ as function of total budget $c$, (b) optimal interlayer weights assigned to nodes in Layer 1, and (c) optimal interlayer weights assigned to nodes in Layer 2, for two Geo networks  with $n=20, m=30, \lambda_2^{(1)}=1.3902, \lambda_2^{(2)}=0.4766$, $c^*=9.5320$.}
	\label{fig:Primal1Case3} 
\end{figure*}
\begin{figure*}
	\centering
	\subfloat[\label{fig:Primal1W2Case3c10}]{\includegraphics[clip,width=.3\columnwidth]{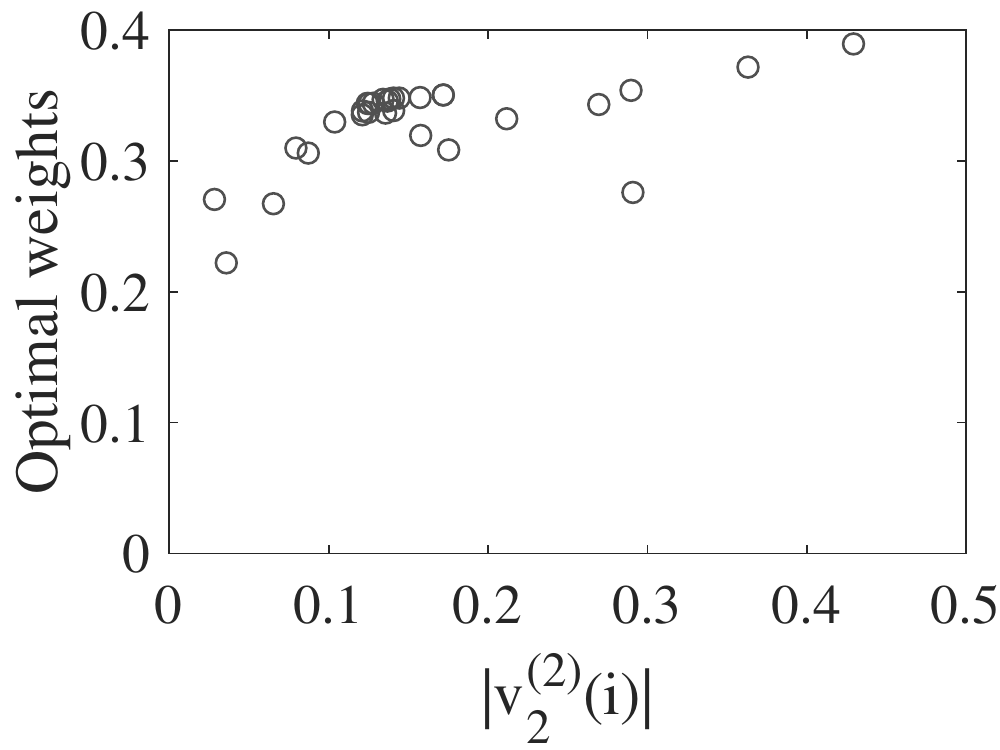}} \ \ \ \ \
	\subfloat[\label{fig:Primal1W2Case3c20}]{\includegraphics[clip,width=.3\columnwidth]{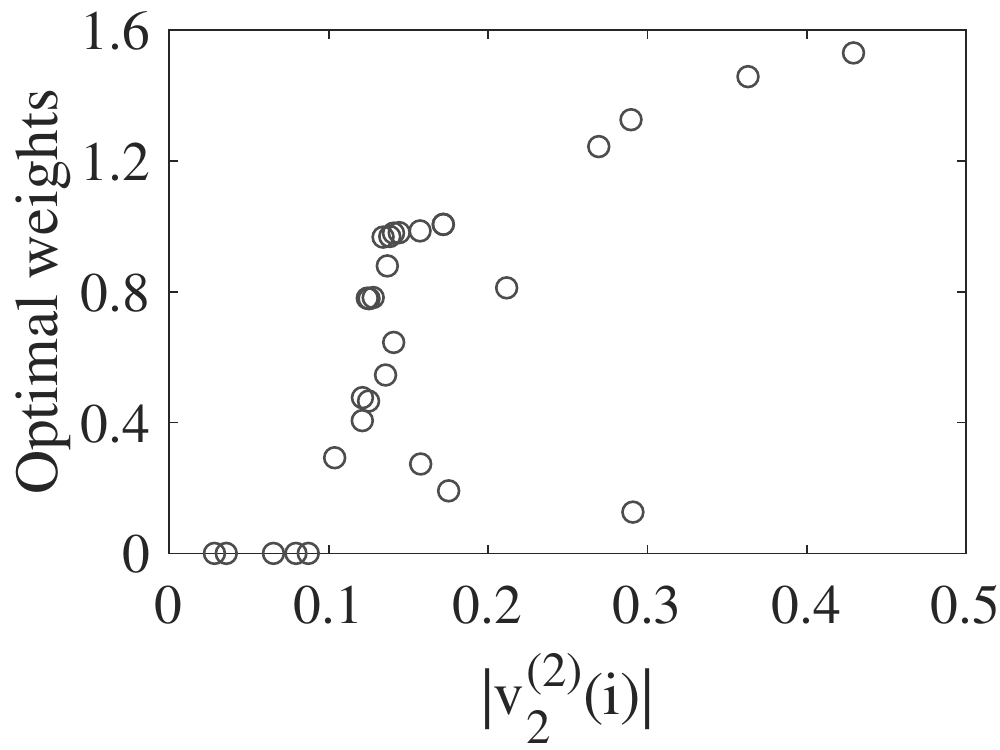}} \ \ \ \ \
	\subfloat[\label{fig:Primal1W2Case3c30}]{\includegraphics[clip,width=.3\columnwidth]{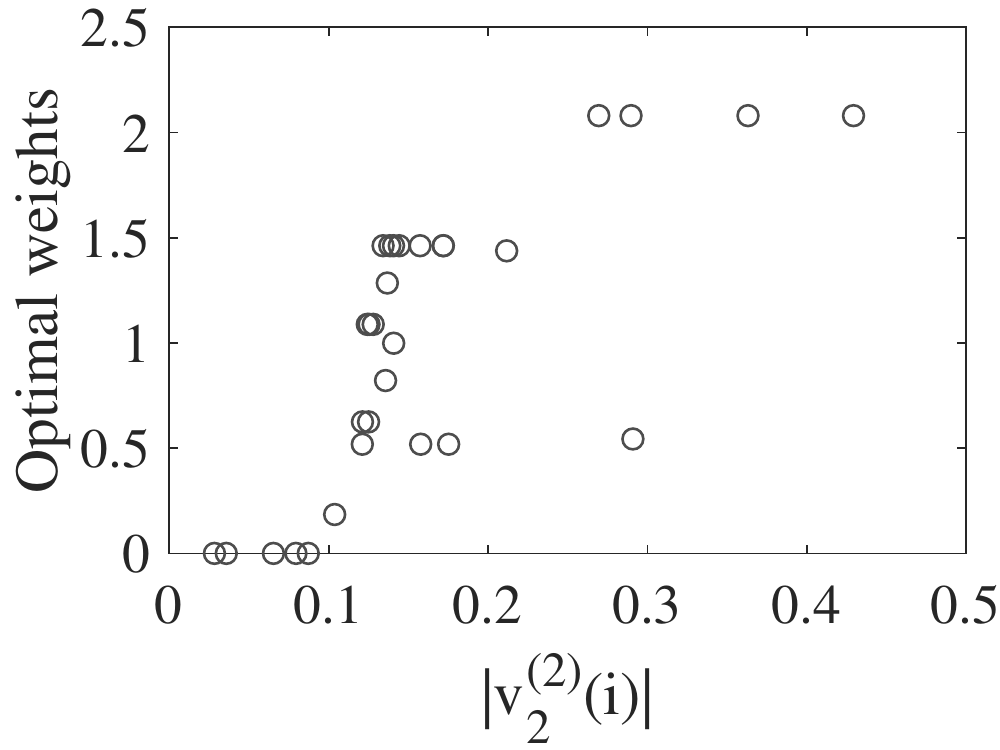}} \ \ \ \ \
	\caption{Optimal weights in Layer 2 as function of Fiedler vector components corresponding to Figure \ref{fig:Primal1Case3} for (a) $c=10$, (b) $c=20$, and (c) $c=30$.}
	\label{fig:Primal1W2Case3} 
\end{figure*}

\subsection{Dual problem and embedding}\label{sec:DuAll}
We investigate  geometric dual problems of each three cases in Figure \ref{fig:UniformWeights} to show different diffusion phases in various regions of optimal weights. For Case 1, Figure \ref{fig:EmbGeoCase1} shows the  embedding of the examples that previously analyzed by solving the primal problem (Figure \ref{fig:Primal1W2Case1} and  Figure \ref{fig:UniformCase1}). For $c<c^*$ in Figure \ref{fig:EmbGeoCase1c1}, we observe the clumped pattern predicted by Lemma \ref{lem:clump}. For the intermediate value $c^*<c<c^{**}$ in Figure \ref{fig:EmbGeoCase1c10}, we note that while the nodes in Layer 1 with larger specific connectivity keep their clumped pattern, the nodes in Layer 2 with smaller specific connectivity become distributed around the clumped nodes. The conditions is reversed for larger $c>c^{**}$ in Figure \ref{fig:EmbGeoCase1c30} where nodes in Layer 2 are clumped together while in Layer 1 are distributed.  Different embedding patterns illustrate different diffusion phases as discussed in Appendix \ref{AppDiffusion}. The results are as follows. For small $c<c^*$ in Figure \ref{fig:EmbGeoCase1c1}, we observe an intralayer phase of optimal diffusion in both layers. For intermediate budgets $c^*<c<c^{**}$ in Figure \ref{fig:EmbGeoCase1c10}, while optimal diffusion in Layer 1 is still in intralayer phase, it is through both intralinks and interlinks in Layer 2. For large budget values $c>c^{**}$ in Figure \ref{fig:EmbGeoCase1c30}, the situation is converse; while Layer 1 undergoes a combination of intralayer and interlayer diffusion phases, Layer 2 undergoes a single phase of interlayer.   \\

Figure \ref{fig:EmbWSCase1} shows embedding results for another example of Case 1 in Figure \ref{fig:UniformCase1}. By Figure \ref{fig:EmbWSCase1Lam}, it is observed the embedding dimension at each $c$ is equal to corresponding multiplicity of supra-Laplacian algebraic connectivity. Embedding for $c\leq c^*$ has a clumped pattern similar to Figure \ref{fig:EmbGeoCase1c1}, thus one dimensional. \\

Examples of embedding results for Cases 2 and 3 in Figure \ref{fig:UniformWeights} are shown in Figures \ref{fig:EmbGeoCase2} and \ref{fig:EmbGeoCase3}, respectively. When $c>c^*$, in both cases, the optimal diffusion in subgraph with larger specific algebraic connectivity is prominently through intralinks, while through interlinks and intralinks in subgraph with smaller specific connectivity.

%Therefore, for large $c$, the optimal diffusion in $G_2$ is mainly due to intralinks while the diffusion in $G_1$ is mainly through interlinks
%For $c<c^*$ the Fiedler vector is $\begin{bmatrix} m\boldsymbol{1}_n \\ -n\boldsymbol{1}_m \end{bmatrix}$, so individual networks both operate distinguishably. For $c^*<c<c^{**}$ the Fiedler vector is $\begin{bmatrix} 0 \\ v_2^{(2)} \end{bmatrix}$, so only Layer 1 operates distinguishably. For $c>c^{**}$ the Fiedler vector is $=\begin{bmatrix} u_2^{(1)} \\ 0 \end{bmatrix}$, so only Layer 2 operates distinguishably.  \\

\begin{figure*}
	\centering
	\subfloat[\label{fig:EmbGeoCase1c1}]{\includegraphics[clip,width=.25\columnwidth]{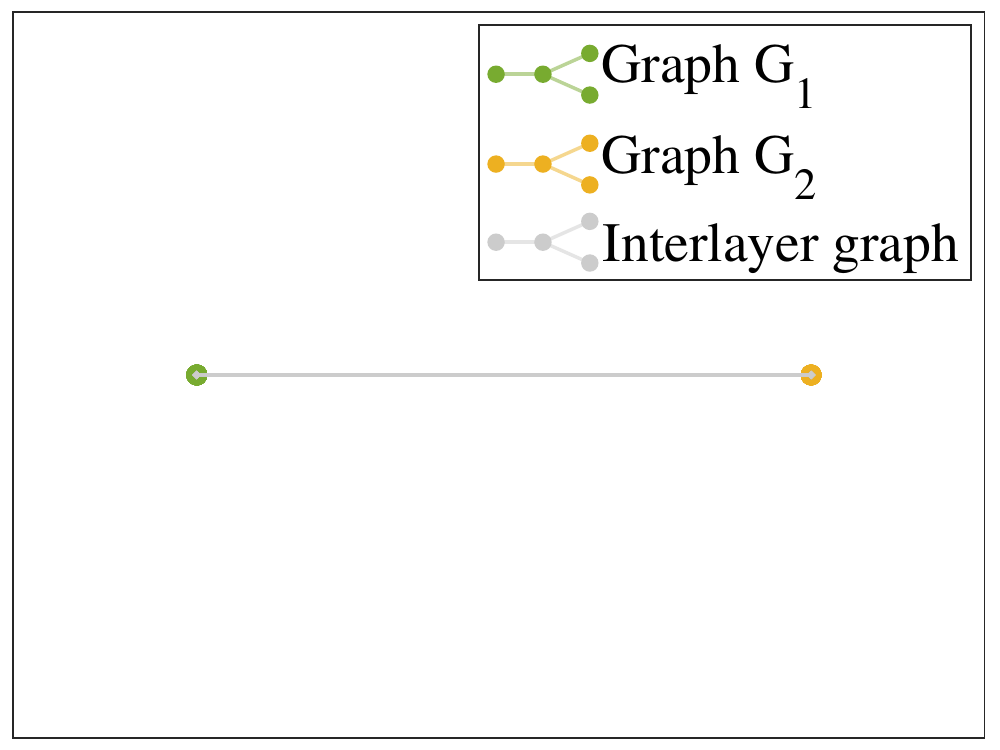}} \ \ \ \ \
	\subfloat[\label{fig:EmbGeoCase1c10}]{\includegraphics[clip,width=.25\columnwidth]{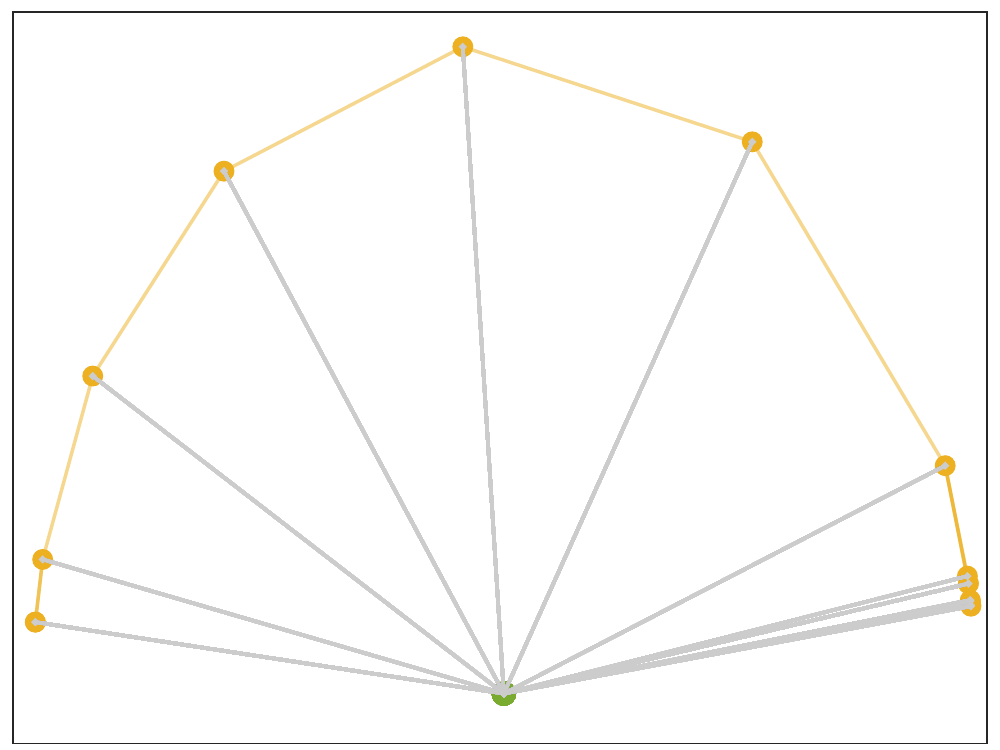}} \ \ \ \ \
	\subfloat[\label{fig:EmbGeoCase1c30}]{\includegraphics[clip,width=.25\columnwidth]{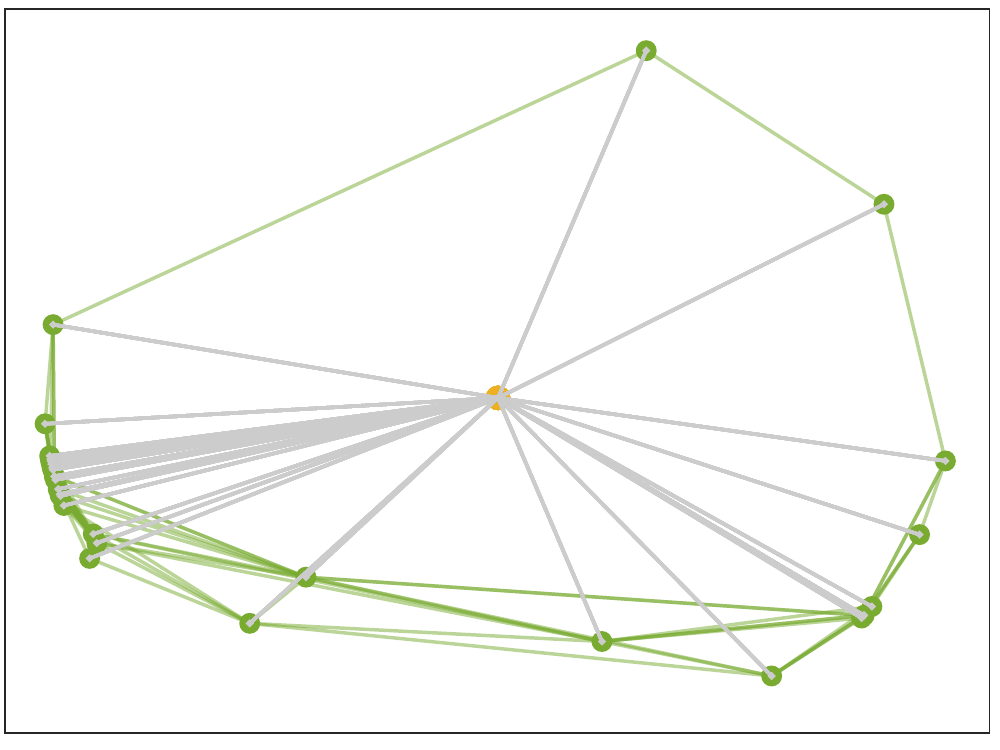}}
	\caption{Graph embeddings corresponding to Figure \ref{fig:Primal1} for (a) $c=1$, (b) $c=10$, and (c) $c=30$.}
	\label{fig:EmbGeoCase1} 
\end{figure*}

\begin{figure*}
	\centering
	\subfloat[\label{fig:EmbWSCase1Lam}]{\includegraphics[clip,width=.25\columnwidth]{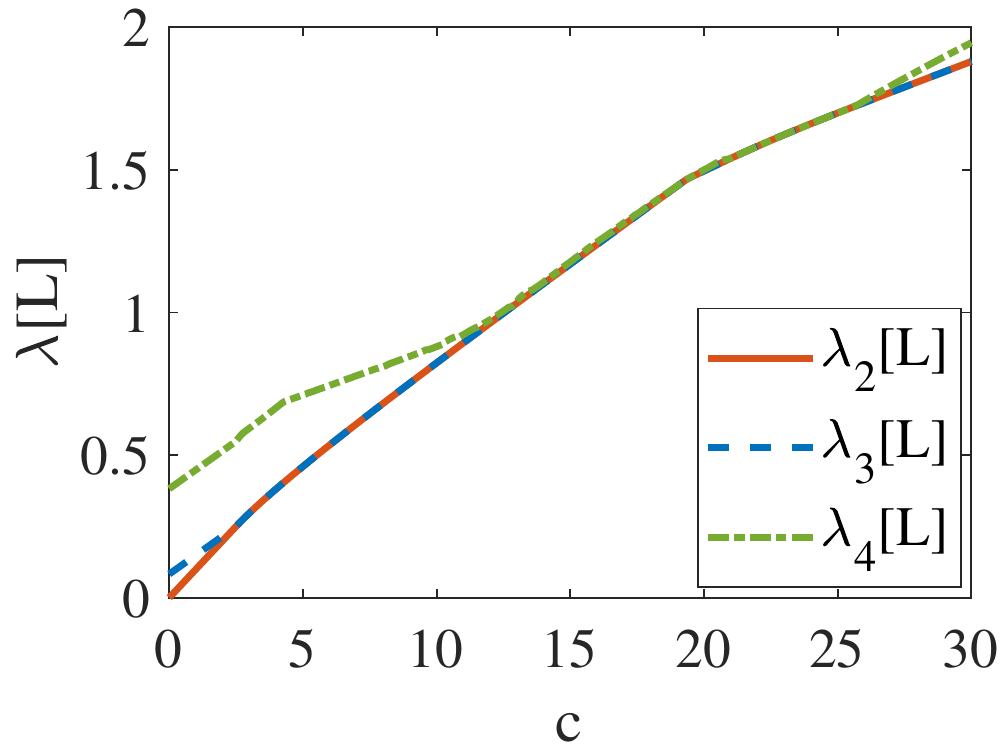}} \ \ \ \ \
	\subfloat[\label{fig:EmbWSCase1c10}]{\includegraphics[clip,width=.25\columnwidth]{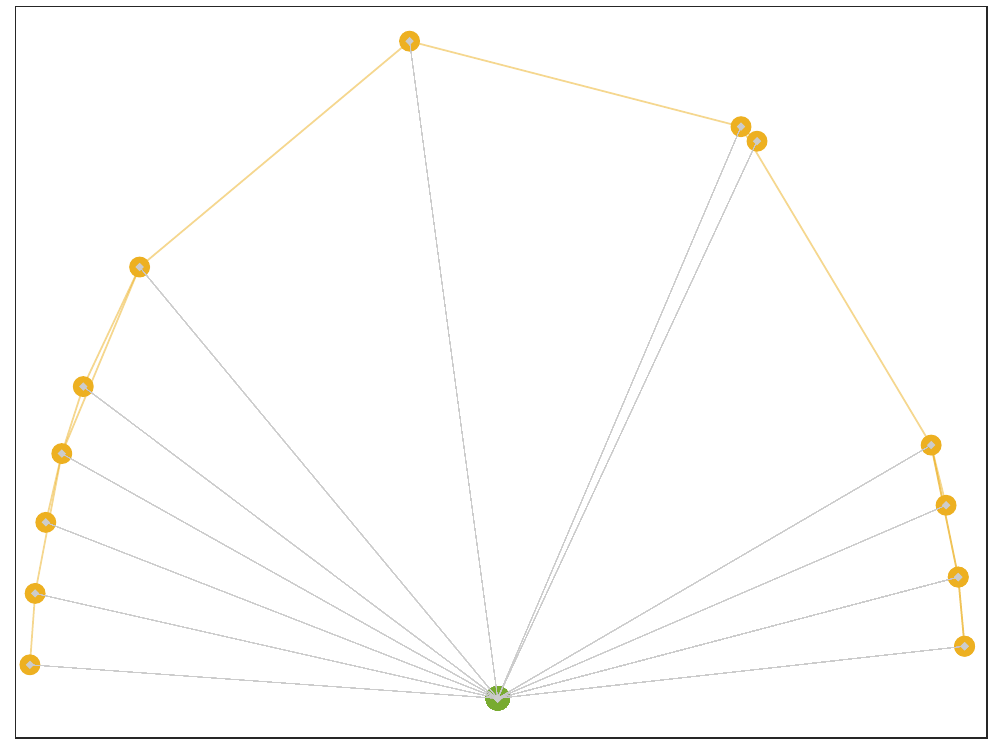}} \ \ \ \ \
	\subfloat[\label{fig:EmbWSCase1c20}]{\includegraphics[clip,width=.25\columnwidth]{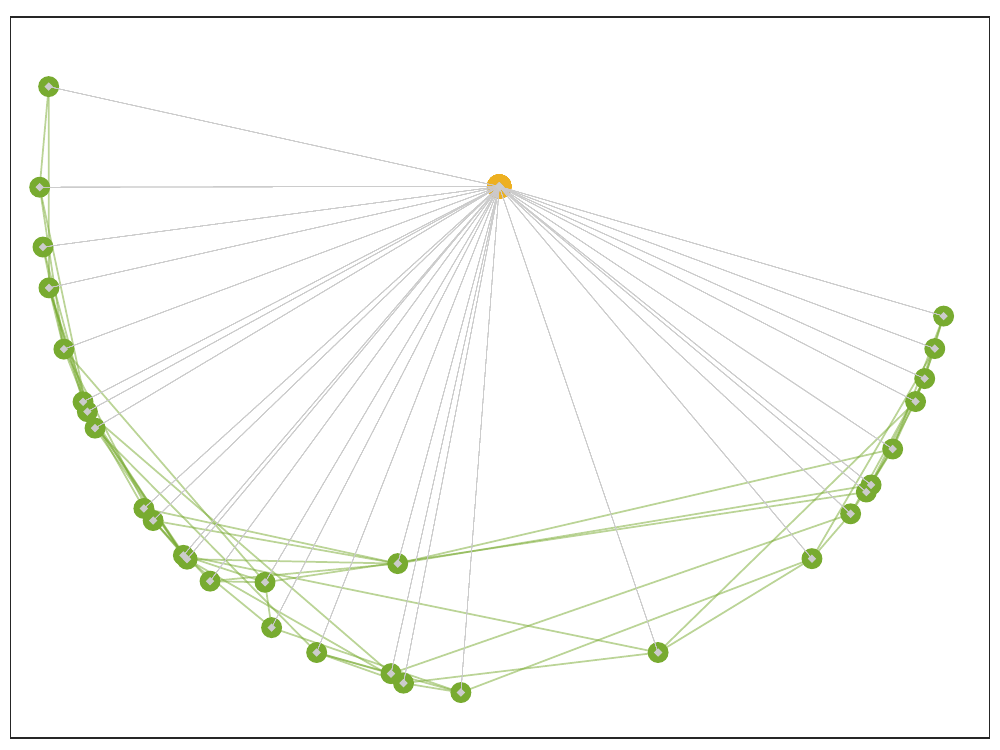}} \\
	\subfloat[\label{fig:EmbWSCase1c24}]{\includegraphics[clip,width=.25\columnwidth]{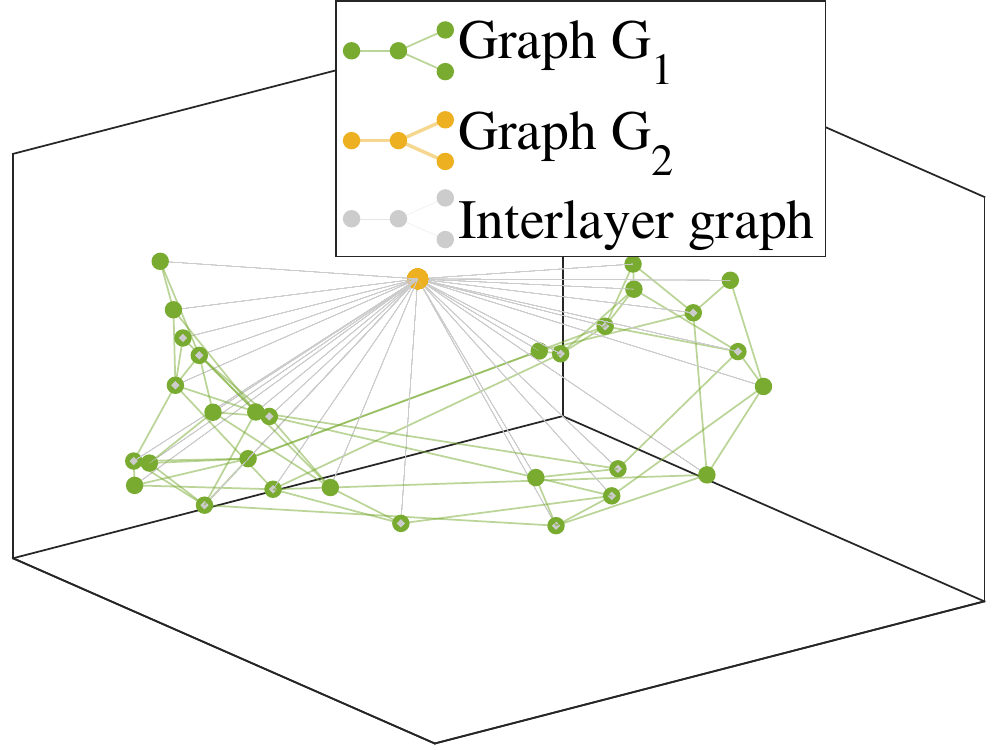}} \ \ \ \ \
	\subfloat[\label{fig:EmbWSCase1c50}]{\includegraphics[clip,width=.25\columnwidth]{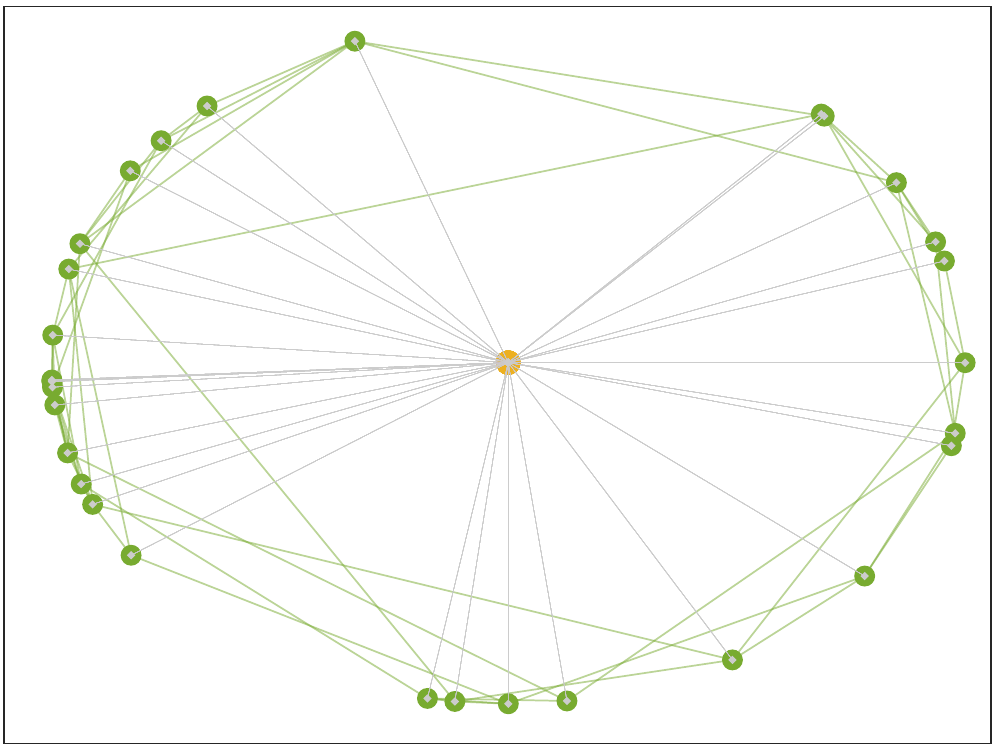}}
	\caption{Graph embedding results for two WS networks as an example of Case 1 in Figure \ref{fig:UniformWeights} with $n=30, m=15, \lambda_2^{(1)}=0.5444, \lambda_2^{(2)}=0.0828$, $c^*=2.4834, c^{**}=13.8486$: (a) the first three eigenvalues for optimal weights, and embedding for (b)  $c=10$, (b) $c=20$, (d) $c=24$, and (e) $c=50$.}
	\label{fig:EmbWSCase1} 
	
\end{figure*}
\begin{figure*}
	\centering
		\subfloat[\label{fig:EmbGeoCase2c10}]{\includegraphics[clip,width=.25\columnwidth]{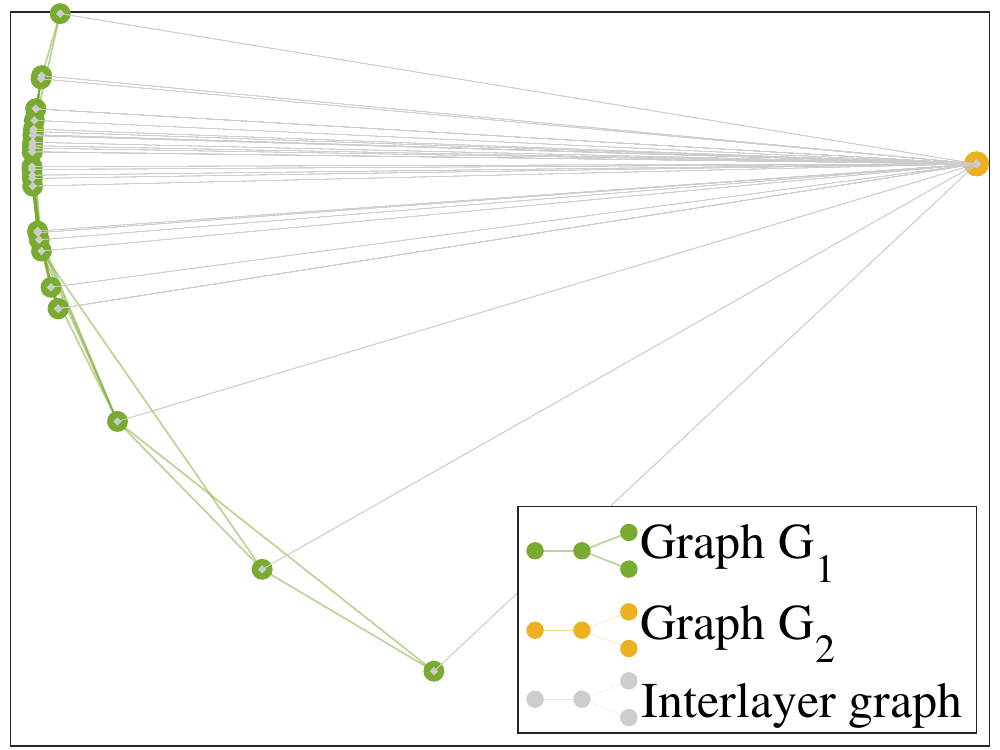}} \ \ \ \ \
		\subfloat[\label{fig:EmbGeoCase2c50}]{\includegraphics[clip,width=.25\columnwidth]{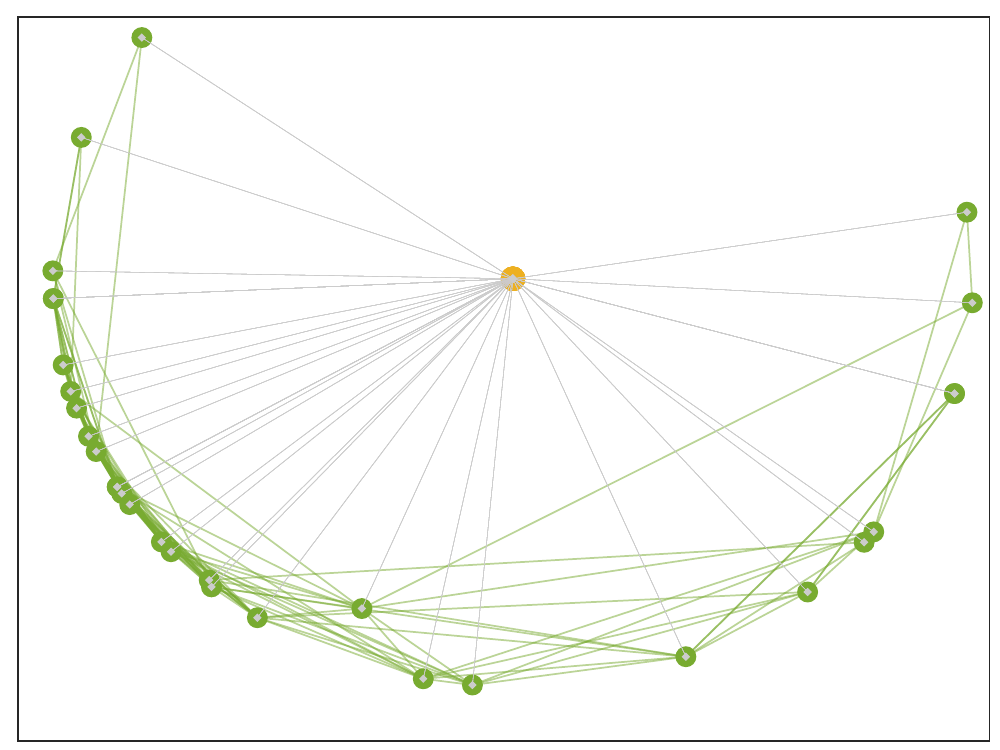}} \ \ \ \ \
		\subfloat[\label{fig:EmbGeoCase2c200}]{\includegraphics[clip,width=.25\columnwidth]{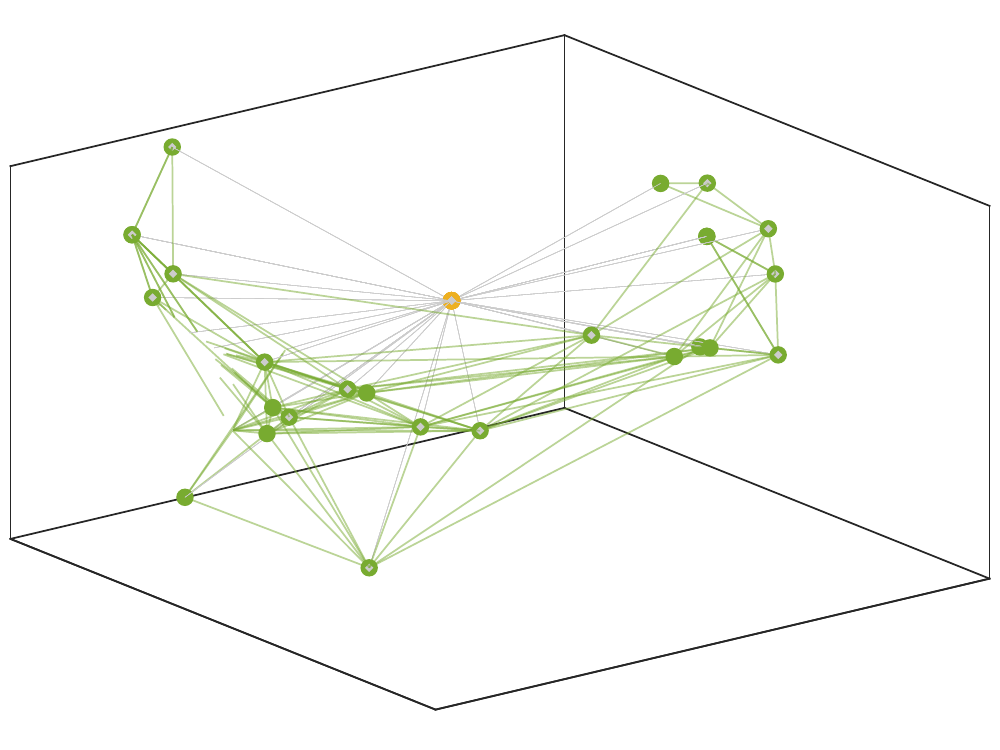}}
		\caption{Graph embeddings corresponding to Figure \ref{fig:Primal1Case2} for (a) $c=10$, (b) $c=50$, and (c) $c=200$.}
		\label{fig:EmbGeoCase2} 
	
\end{figure*}
\begin{figure*}
	\centering
		\subfloat[\label{fig:EmbGeoCase3c10}]{\includegraphics[clip,width=.25\columnwidth]{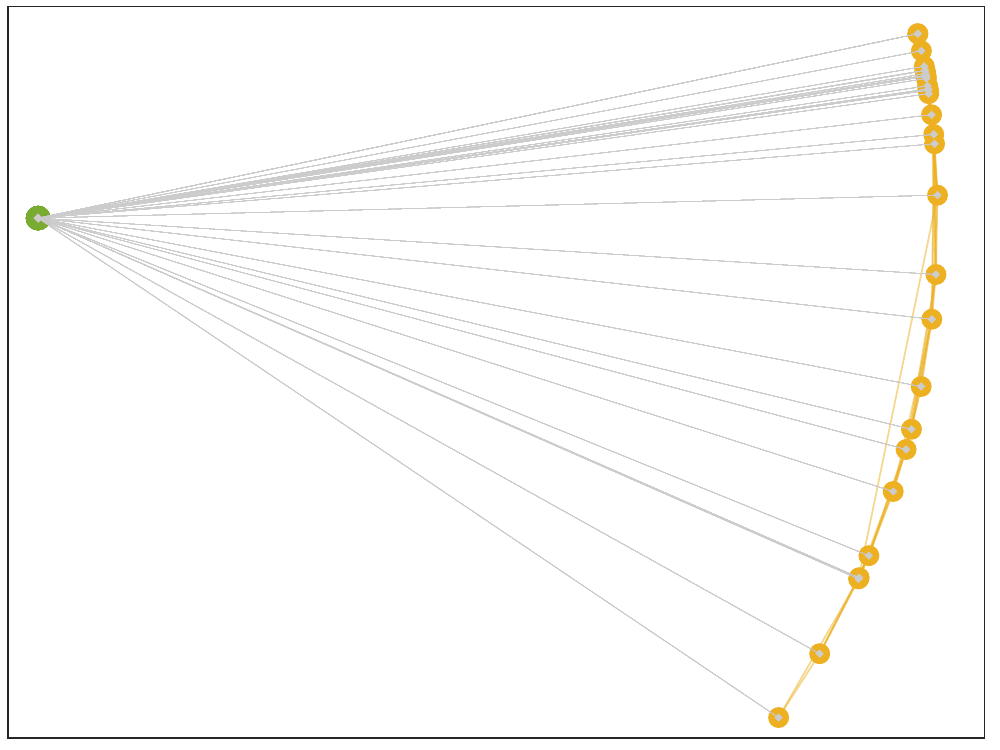}} \ \ \ \ \
		\subfloat[\label{fig:EmbGeoCase3c50}]{\includegraphics[clip,width=.25\columnwidth]{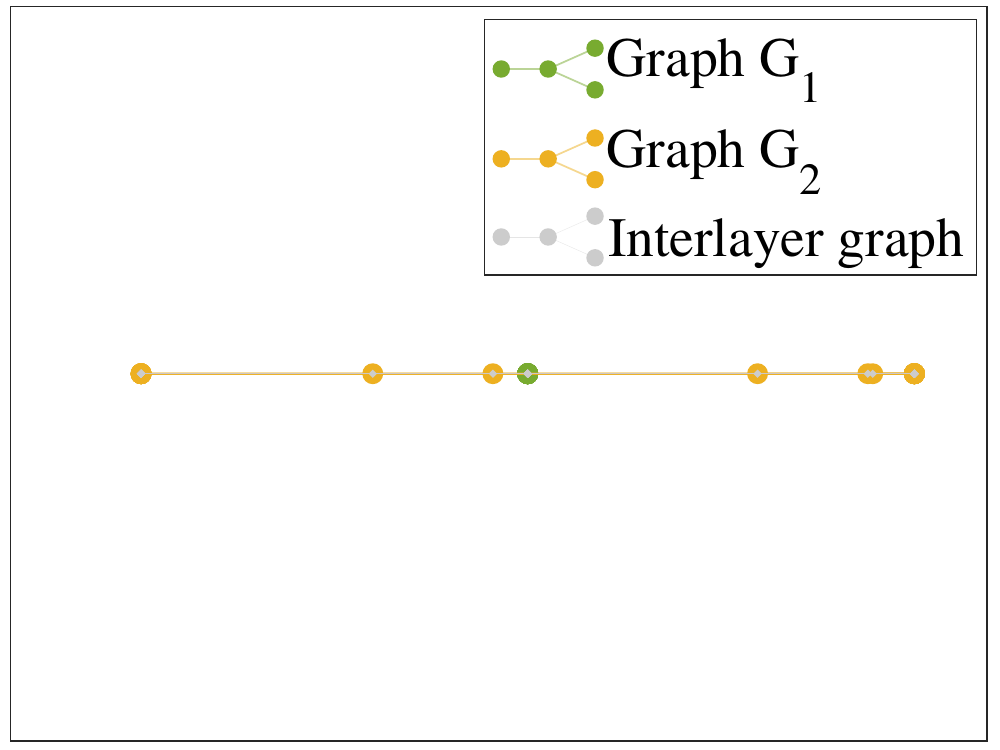}} \ \ \ \ \
		\subfloat[\label{fig:EmbGeoCase3c100}]{\includegraphics[clip,width=.25\columnwidth]{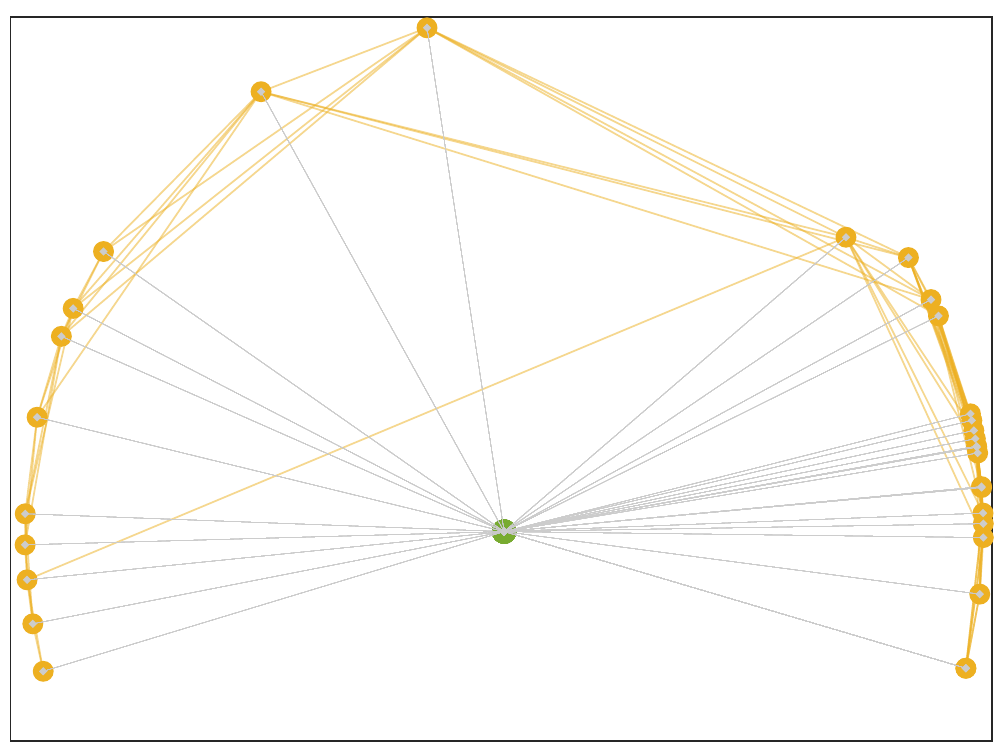}}
		\caption{Graph embeddings corresponding to Figure \ref{fig:Primal1Case3} for (a) $c=10$, (b) $c=50$, and (c) $c=100$.}
		\label{fig:EmbGeoCase3} 
	
\end{figure*}

\section{Interconnections constrained to an admissible set}\label{sec:admissible}
In Section \ref{sec:all-pair}, we had constraint only on the total budget $c$ and no restriction on the number or patterns of interlinks. In this section, we consider a situation where the number of interlinks is limited and we can assign weights only to an admissible set $\left(i,j\right)\in E_a\subset E_3$ and all other edges $E_3/E_a$ have zero weight. \\
%Our first result is that before a threshold $c<c^*$, Lemma \ref{lem:maxLam2} still holds for regular interconnections, where the weights matrix $W$ has constant row and column sum \citep{Mieghem2017Interdepen, Mieghem2019Regular}. Indeed, regular interlinks pattern can be an optimal setting only before $c^*$, and beyond this threshold, regularity generally breaks down by an optimum interconnection. \\

For $c\leq c^*$ and following Lemma \ref{lem:Opt0}, we know that regular weight distribution leads to maximum algebraic connectivity. The regularity condition \eqref{eq:RegulCond0} is generally not feasible for all interconnection patterns. When regular interconnection is not feasible for a given set of admissible interlinks, the bound $\lambda^*_2=\left(\frac{1}{n}+\frac{1}{m}\right)c$ is not attainable and there is generally no region of optimal uniform weights. In such conditions, the region $c\leq c^*$ corresponding to uniform weights, and hence the associated transition fades away  (see Figure \ref{fig:RegulRandInterConec}). This means, with no regularity, there is no region for unified operation of individual network components--Figure \ref{fig:EmbdRegRand} shows the absence of the clumped pattern embedding. As a consequence, interlinks start to contribute to diffusion within each individual layer from small values of coupling strength $c$. %The only difference between the multi layer networks of left and right plots in each of Figures \ref{fig:RegulRandInterConec} and \ref{fig:EmbdRegRand} is the interconnection pattern, and the individual layers, the number of interlinks, and the total budgets are identical for both plots of a figure. 
It should be reemphasized that regularity and uniform weights are different;  Figure \ref{fig:WeightsRegNonuni} illustrates an example with regularity in optimal interlayer weights but nonuniform weights and thus, the bound $\lambda^*_2=\left(\frac{1}{n}+\frac{1}{m}\right)c$ is attained. \\ 
\begin{figure*}
	\centering
	\subfloat[\label{fig:ReguLinkWeights}]{\includegraphics[clip,width=.3\columnwidth]{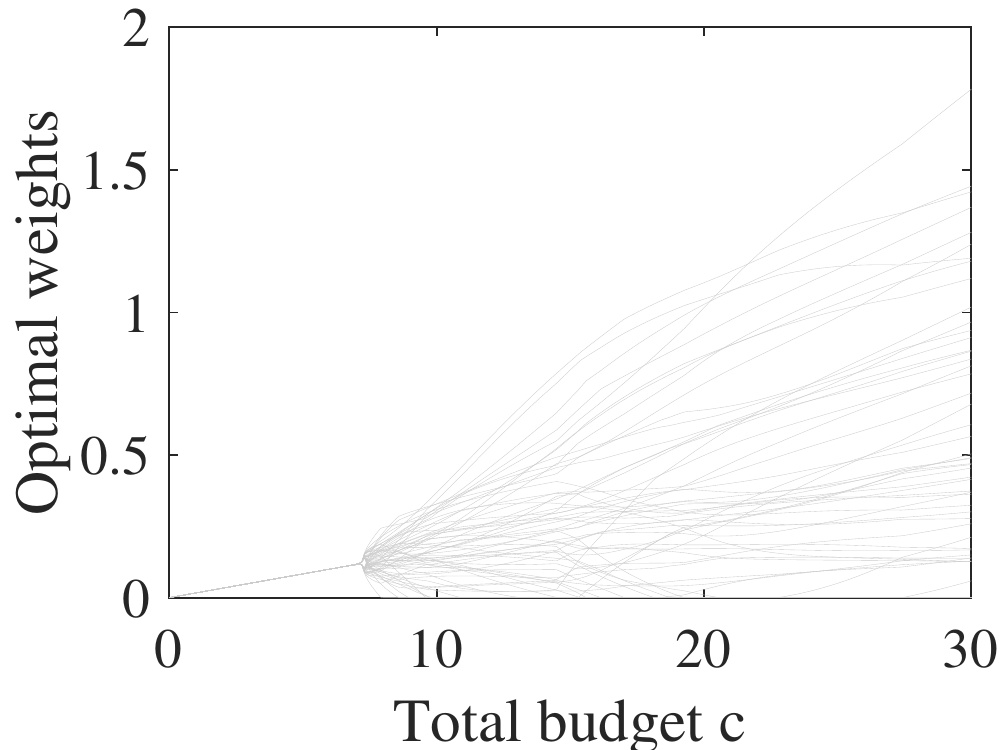}} \ \ \ \ \ \
	\subfloat[\label{fig:RandLinkWeights}]{\includegraphics[clip,width=.3\columnwidth]{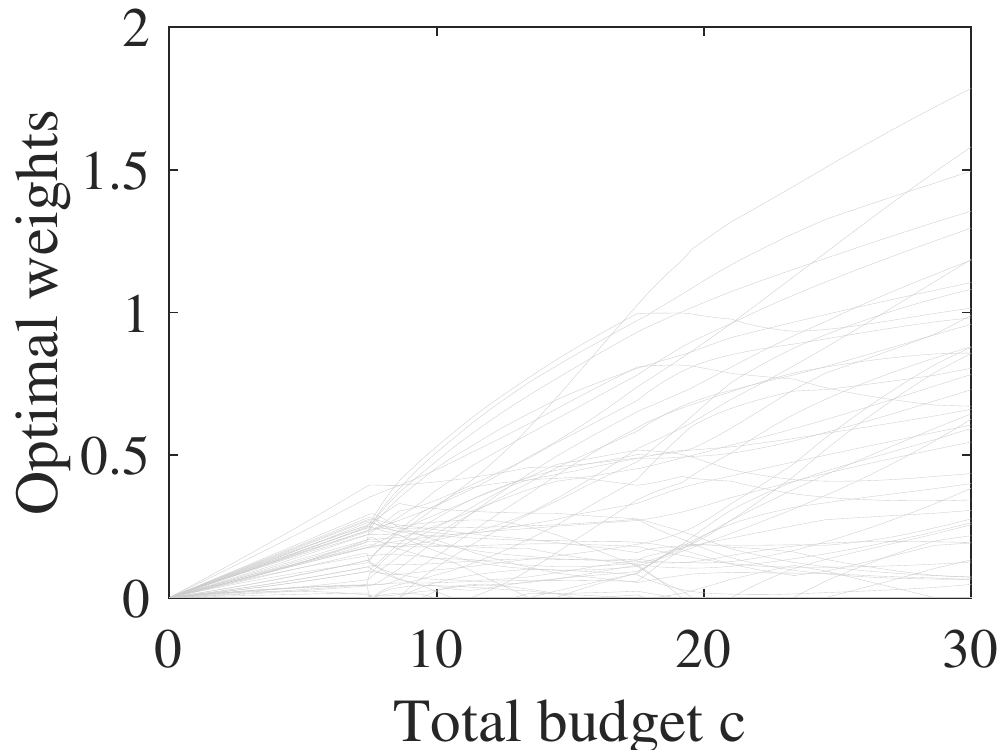}}
	%		\subfloat[\label{fig:RandLinkWeights}]{\includegraphics[clip,width=.35\columnwidth]{Lam2General.pdf}}
	\caption{Optimal weights for maximizing algebraic connectivity in a multilayer including two ER networks, each with 30 nodes, and 60 interlinks for (a) $k$-to-$k$ interconnection with $k=2$, and (b) random interconnections. In (a) regularity is feasible with uniform weights before $c^*$, while in (b), without regularity feasible, optimal weights are always distributive (nonuniform) and there is no uniform optimal weights region. After $c^*$, maximum algebraic connectivity in both patterns is attained by a weight distribution that does not satisfy the regularity conditions. }
	\label{fig:RegulRandInterConec}
	
\end{figure*}
\begin{figure*}
	\centering
	\subfloat[\label{fig:EmbdRegC7}]{\includegraphics[clip,width=.25\columnwidth]{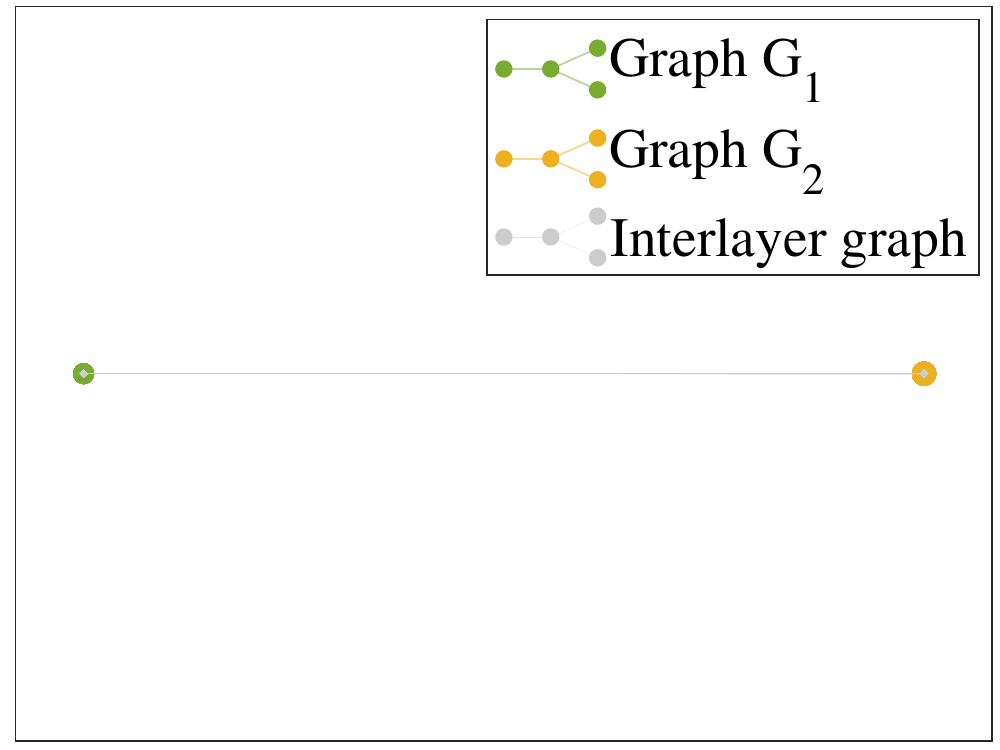}} \ \ \ \ \ \
	\subfloat[\label{fig:EmbdRandC7}]{\includegraphics[clip,width=.25\columnwidth]{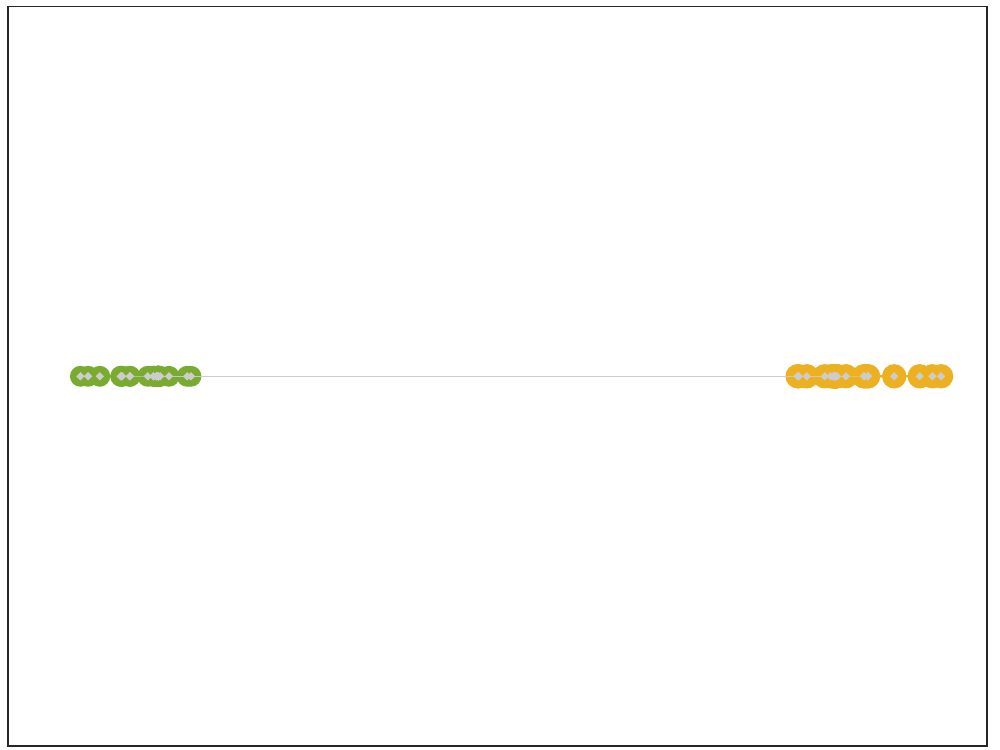}}
	\caption{Graph embedding of a multilayer including two ER networks, each with 30 nodes, and 60 interlinks for $c\leq c^*$ for (a) $k$-to-$k$ interconnection with $k=2$, and (b) random interconnections. In (a), with regularity feasible, the embedding indicates nodes clumped together, while in (b), without regularity feasible, there is no region of clumped nodes embedding.}
	\label{fig:EmbdRegRand}
	
\end{figure*}
\begin{figure*}
	\centering
	\subfloat[\label{fig:WeightsRegNonuni1}]{\includegraphics[clip,width=.3\columnwidth]{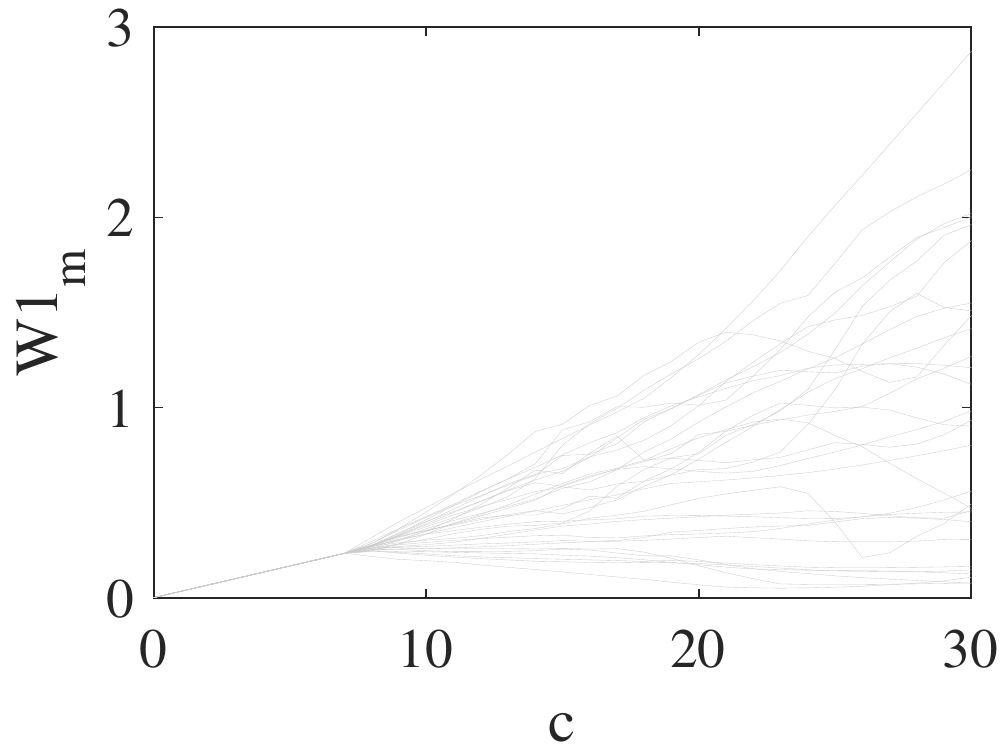}} \ \ \ \ \
	\subfloat[\label{fig:RandLinkWeights2}]{\includegraphics[clip,width=.3\columnwidth]{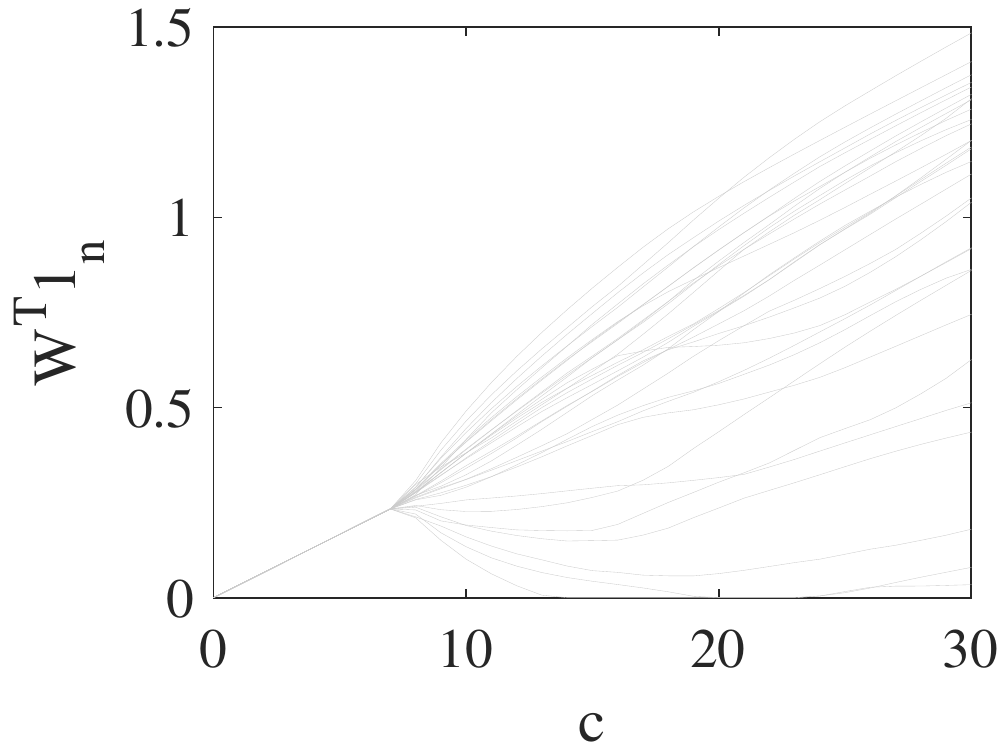}} \ \ \ \ \
	\subfloat[\label{fig:WeightsRegNonuni0}]{\includegraphics[clip,width=.3\columnwidth]{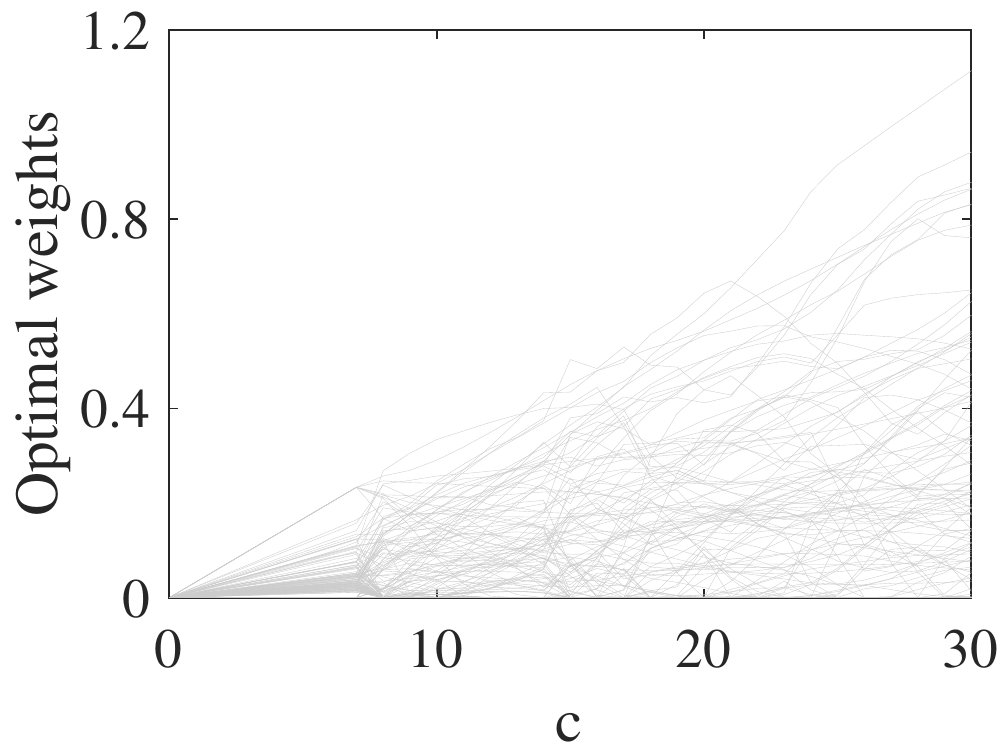}}
	\caption{Optimal weights in a multilayer including two ER networks, each with 30 nodes and total 60 random interlinks, as example of \textit{nonuniform regular optimal weights}. (a) optimal interlayer weights assigned to nodes in Layer 1, (b) optimal interlayer weights assigned to nodes in Layer 2, and (3) optimal weight assigned to each interlink. }
	\label{fig:WeightsRegNonuni}
	
\end{figure*}

Lemma \ref{lem:Opt0} implies that  for a given set of admissible interlinks, regular interconnection is feasible, the maximum algebraic connectivity and the corresponding Fiedler vector are the same for all interconnection patterns for $c\leq c^*$. Therefore, before the threshold there is no dependency of maximum algebraic connectivity on the number of interconnections or admissible set $E_a$. However, what distinguishes between different interconnection patterns is the value of threshold $c^*$ that depends on the  number and pattern of interlinks. We have shown that for the case of all-pairs interconnection, $c^*$ is computed by \eqref{eq:thresholds} and \eqref{eq:threshold2}. However, there is no explicit expression for $c^*$ in general regular interconnections. For particular configurations,  \citet{sahneh2014exact,Mieghem2017Interdepen,Mieghem2019Regular} provide some bounds, and for  one-to-one interconnection pattern, the exact threshold budget is calculated in \citep{sahneh2014exact}, and in \citep{, shakeri2015PRL} from  different approaches. %, to be
%$
%\begin{aligned}
%c^*=\frac{n}{2}\lambda_2\left[\left(L_1^\dagger+L_2^\dagger\right)^\dagger\right]
%\end {aligned}
%$
%where $ ^{\dagger}$ denotes the Moore-Penrose pseudoinverse. 
Moreover, \citet{Mieghem2019Regular} shows that for regular $k$-to-$k$ interconnections, the transition threshold $c^*$ is upper-bounded by the minimum algebraic connectivity of subgraphs times $n$, and lower-bounded by it times $n/2$, i.e. when $n=m$ it follows $\frac{n\lambda_2^{(0)}}{2}\leq c^*\leq n\lambda_2^{(0)}$. The upper-bound is attained when $k=n$. This implies that, the coupling is postponed as much as possible through a complete interconnection; or equivalently, for a given total budget $c$ the weakest coupling occurs with all-pairs interconnection.  \\

Among numerous inter-structures satisfying regularity, the class of $k$-to-$k$ interconnectivity pattern, which is a generalization of the one-to-one scheme, is representative of more real interdependent networks \citep{Mieghem2019Regular}. To unravel more facts about optimized $k$-to-$k$ interconnections, Figures \ref{fig:EmbRegK1}, \ref{fig:EmbRegK2}, and \ref{fig:EmbRegK3} show embedding results for $k=1$, $k=2$, and $k=3$, respectively. These figures use the same individual network components where  $c>c^*$.  Figure \ref{fig:Lam2Reg} shows the corresponding maximum algebraic connectivity values. For the smaller budget $c=10$ in Figures \ref{fig:EmbRegK1C10}, \ref{fig:EmbRegK2C10}, and \ref{fig:EmbRegK3C10}, we observe that the most unified embedding of individual network components is associated with $k=3$, next followed by $k=2$ and $k=1$. 
Therefore, for a given total budget $c$, the larger the number of interconnections, the weaker the coupling. However, Figure \ref{fig:Lam2Reg} illustrates that the algebraic connectivity increases with increasing the number of interconnections for a given total budget. Therefore, despite the networks coupling is becoming weaker by increasing the number of interlinks for fixed total budget $c$, the diffusion becomes faster in such a condition--the speed of diffusion and the modes of diffusion are two distinct properties \citep{sahneh2014exact}. 

For intermediate values of total budget $c$, the main observation in Figures \ref{fig:EmbRegK1C100}, \ref{fig:EmbRegK2C100}, and \ref{fig:EmbRegK3C1000} is that the  one-to-one interconnection holds the minimum embedding space dimension of 2, which increases by 3 in two other cases $k=2$ and $k=3$. For the larger budget $c$ in Figure \ref{fig:EmbRegK1C1000}, when $k=1$, it is seen that each two interlinked nodes are embedded at the same point. Hence, the interlinked nodes are unified due to significant coupling strength $c$. In such conditions, in  one-to-one interconnected networks, the diffusion is connected with an average Laplacian \citep{Radicchi2013}. An average network can also be concluded in Figure \ref{fig:EmbRegK3C1e6} where nodes from different networks are embedded in pairs at same point.  However, here, for $k=3$, the average network associated with intralinks is elaborated with a circle network of interlinks. This circle of interlinks will promote diffusion with respect to  one-to-one interconnected. This is verified by comparing the cases $k=1$ and $k=3$ in Figure \ref{fig:Lam2Reg}. For the odd interconnection pattern $k=2$ in Figure \ref{fig:EmbRegK2C1e5}, the existence of an average network for large $c$ is not directly deducible. \\
 
\begin{figure*}
	\begin{center}
		\subfloat[\label{fig:EmbRegK1C10}]{\includegraphics[clip,width=.25\columnwidth]{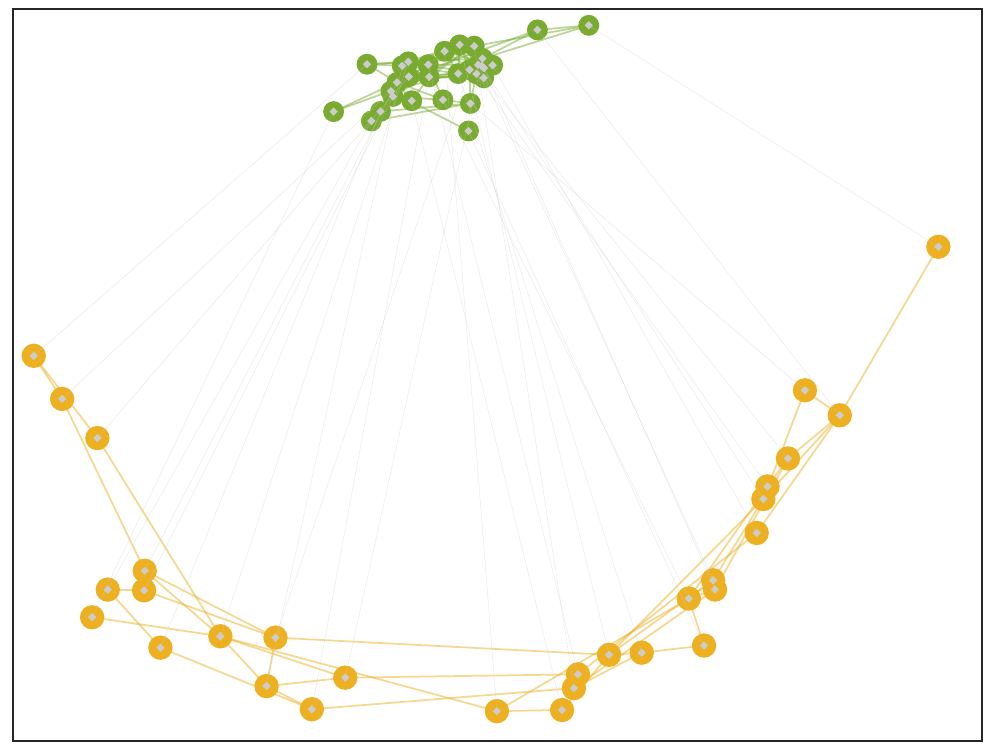}} \ \ \ \ \
		\subfloat[\label{fig:EmbRegK1C100}]{\includegraphics[clip,width=.25\columnwidth]{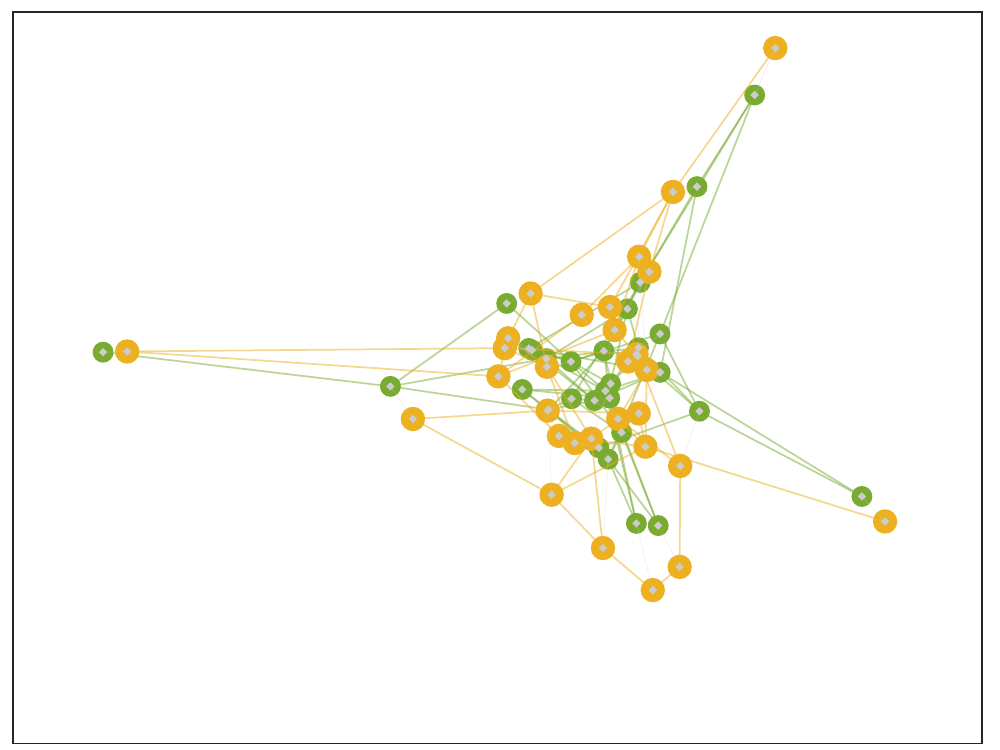}} \ \ \ \ \
		\subfloat[\label{fig:EmbRegK1C1000}]{\includegraphics[clip,width=.25\columnwidth]{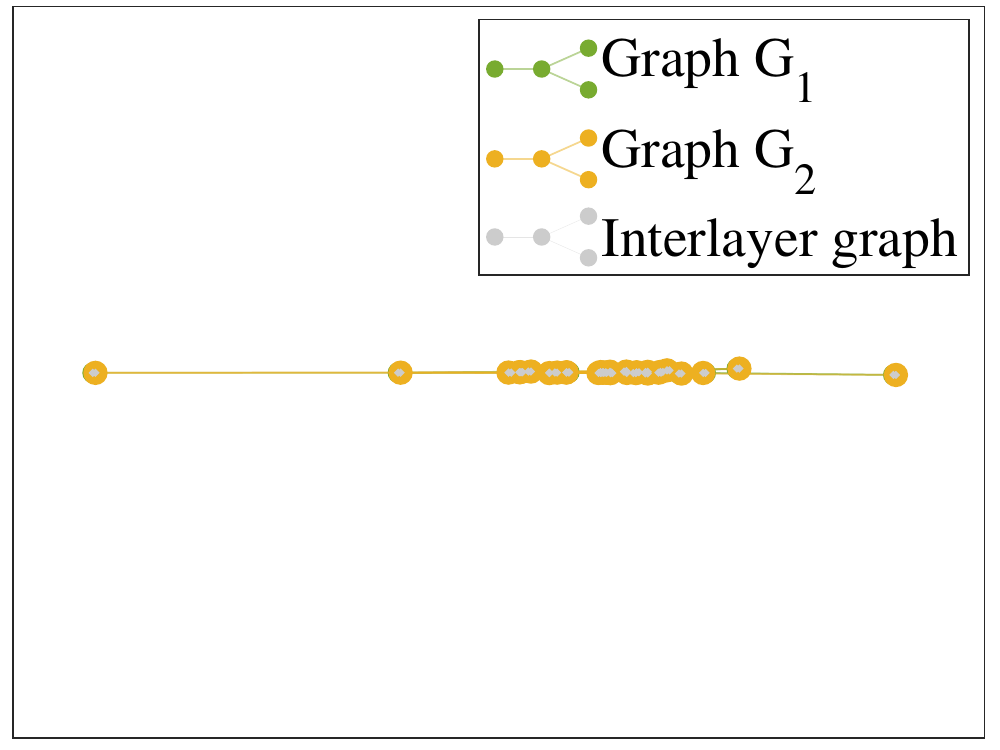}}
		\caption{Graph embedding of a multilayer including two ER networks, each with 30 nodes, and 30 interlinks for $k$-to-$k$ interconnection with $k=1$ and budget values (a) $c=10$, (b) $c=100$, (c) $c=1000$.}
		\label{fig:EmbRegK1}
	\end{center} 
\end{figure*}

\begin{figure*}
	\centering
		\subfloat[\label{fig:EmbRegK2C10}]{\includegraphics[clip,width=.25\columnwidth]{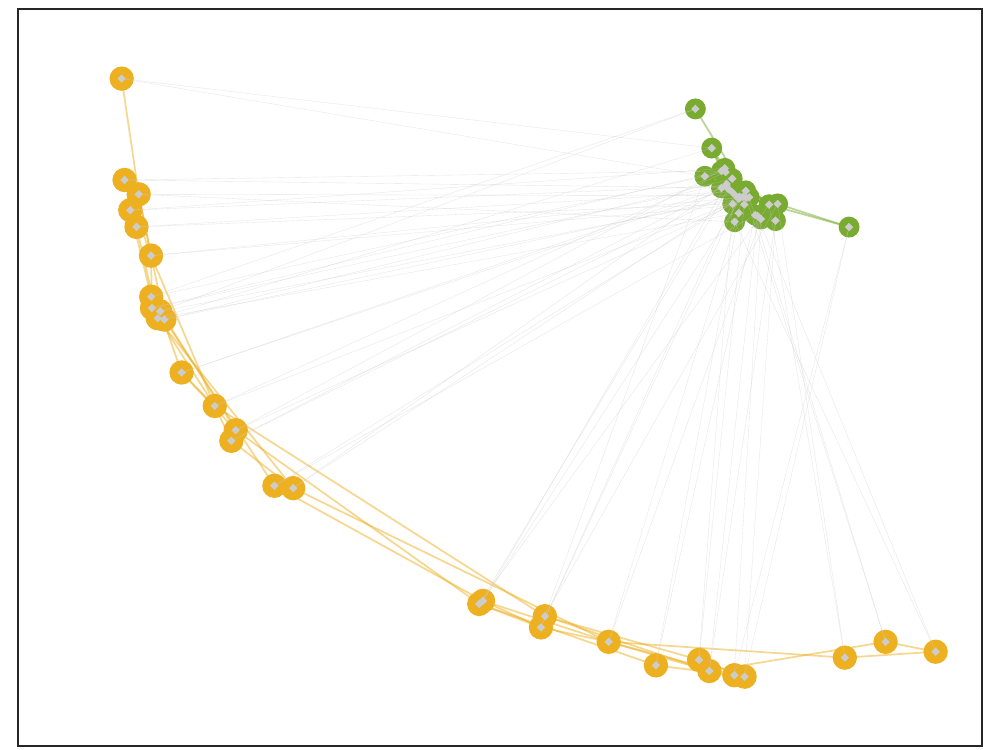}} \ \ \ \ \
%		\subfloat[\label{fig:EmbRegK2C20}]{\includegraphics[clip,width=.35\columnwidth]{EmbRegK2C20.pdf}}
		\subfloat[\label{fig:EmbRegK2C100}]{\includegraphics[clip,width=.25\columnwidth]{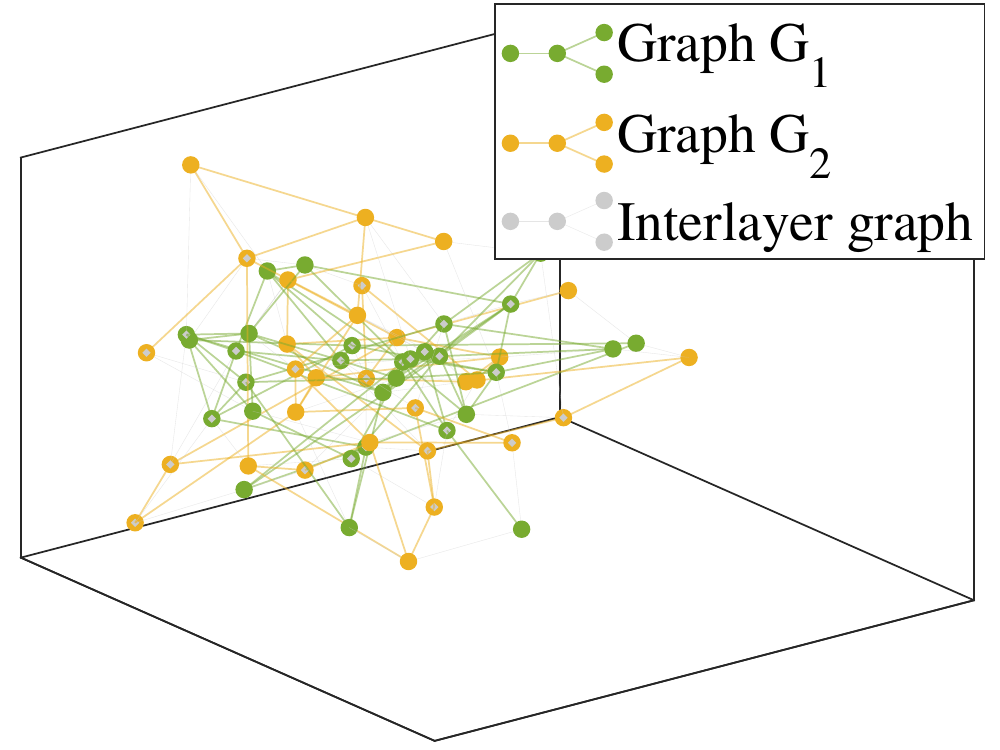}}  \ \ \ \ \
%		\subfloat[\label{fig:EmbRegK2C1000}]{\includegraphics[clip,width=.35\columnwidth]{EmbRegK2C1000.pdf}}
%		\subfloat[\label{fig:EmbRegK2C10000}]{\includegraphics[clip,width=.35\columnwidth]{EmbRegK2C10000.pdf}}
		\subfloat[\label{fig:EmbRegK2C1e5}]{\includegraphics[clip,width=.25\columnwidth]{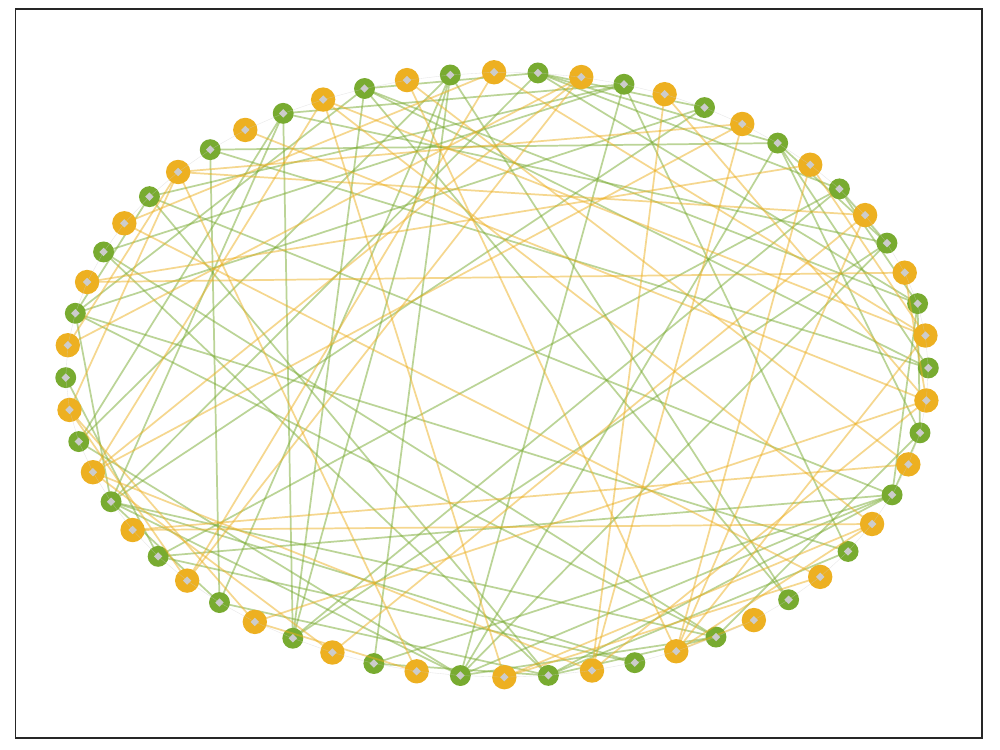}}
	    \caption{Graph embedding of a multilayer including two ER networks, each with 30 nodes, and 60 interlinks for $k$-to-$k$ interconnection with $k=2$ and budget values (a) $c=10$, (b) $c=100$, (c) $c=10^5$.}
		\label{fig:EmbRegK2}
	 
\end{figure*}
\begin{figure*}
	\centering
		\subfloat[\label{fig:EmbRegK3C10}]{\includegraphics[clip,width=.25\columnwidth]{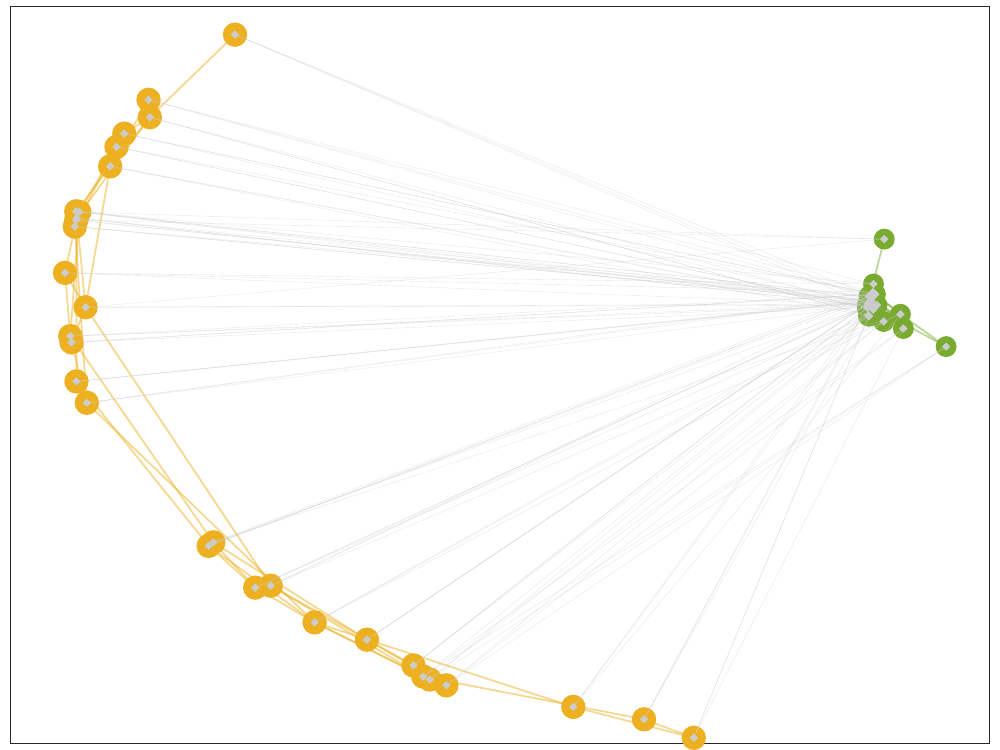}} \ \ \ \ \
		\subfloat[\label{fig:EmbRegK3C1000}]{\includegraphics[clip,width=.25\columnwidth]{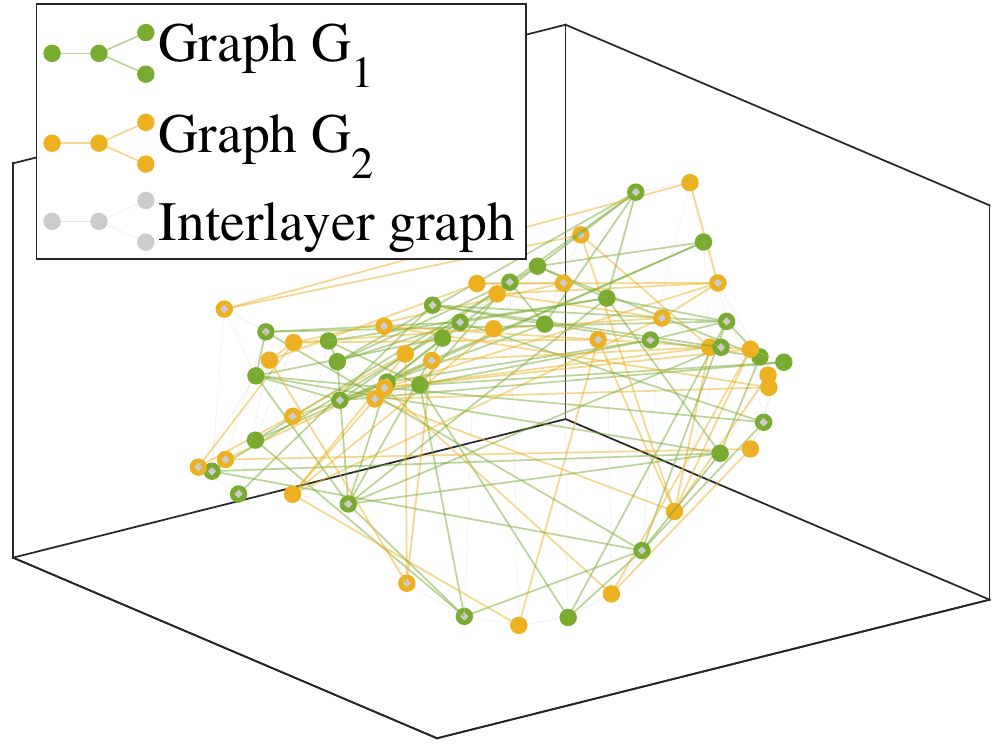}} \ \ \ \ \
		\subfloat[\label{fig:EmbRegK3C1e6}]{\includegraphics[clip,width=.25\columnwidth]{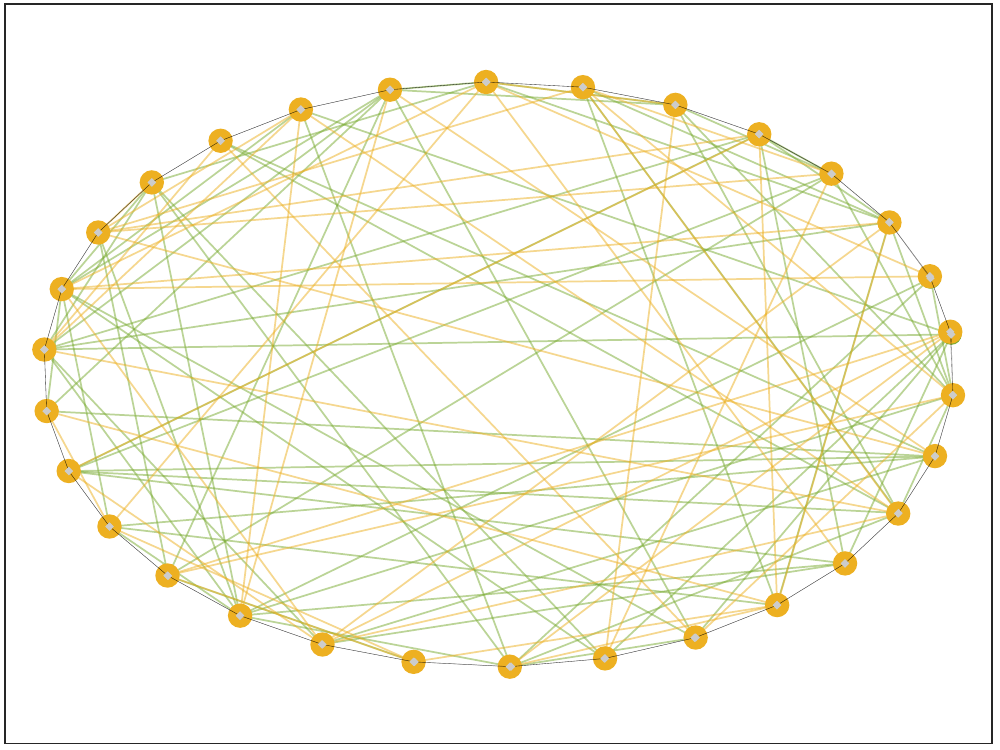}} 
		\caption{Graph embedding of a multilayer including two ER networks, each with 30 nodes, and 90 interlinks for $k$-to-$k$ interconnection with $k=3$ and budget values (a) $c=10$, (b) $c=1000$, (c) $c=10^6$.}
		\label{fig:EmbRegK3}
	 
\end{figure*}
\begin{figure}[!htb]
%	\centering
	\begin{center}
	\includegraphics[clip,width=.3\columnwidth]{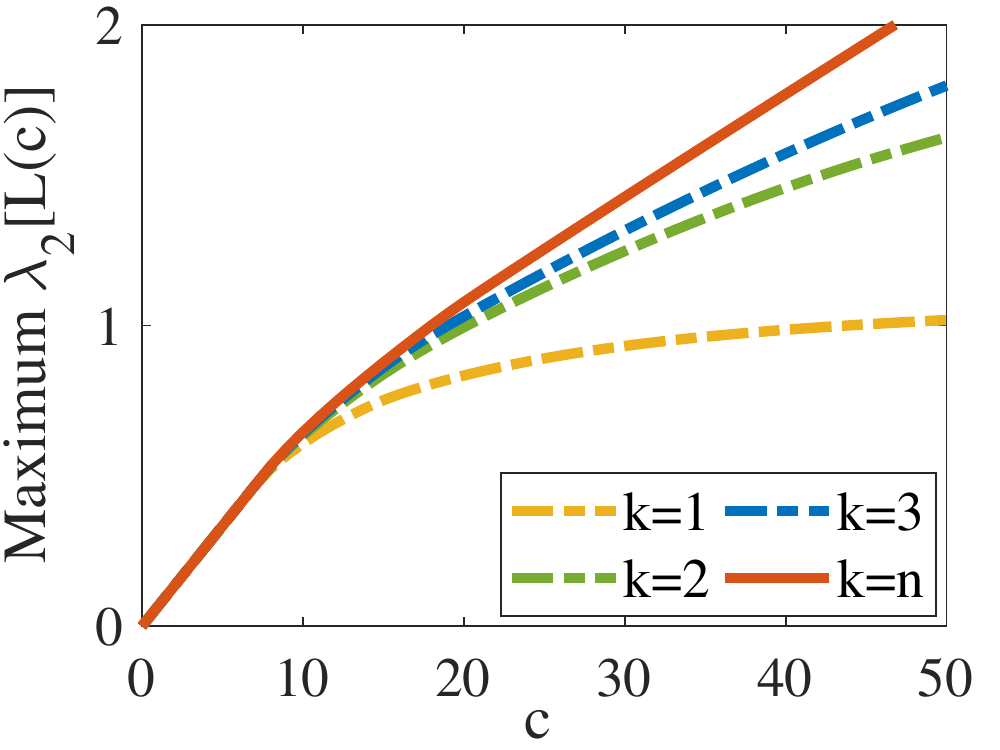} 
	\caption{Maximum algebraic connectivity for $k$-to-$k$ interconnection of two ER networks with $n=m=30$.}
	\label{fig:Lam2Reg}
	\end{center}
\end{figure}

\section{Well-interconnected multilayer networks}\label{sec:Well}
In Section \ref{sec:admissible}, we  observed the importance of the admissible set $E_a$  in final characterization of interdependent network; in this section we ask two questions; 
First, which interconnection pattern leads to the optimal performance. 
For multilayers with one-to-one interconnection, \citet{GomezDiffusion2013} show that the super-diffusion occurs only if some definite condition is satisfied; namely, if the algebraic connectivity of the average Laplacian $L_{ave}=\frac{1}{2}\left(L_1+L_2\right)$ is greater than the algebraic connectivity values of individual network components. 
Second, if there is any other interlink connection strategy, other than the  one-to-one connection strategy that with a given number of interlinks, can lead to super-diffusion without satisfying the condition proposed by Gomez et al. To answer these questions, we use an interlink connection strategy based on a greedy approach by means and show that the approach can result in super-diffusion without requiring the algebraic connectivity of average Laplacian greater than those of individual networks. \\

We investigate the situation where, for constant budget $c$ and number of interlinks $r$, the weights $w_{ij}=w_0=c/r$ assigned to all interlinks are identical, and in turn interconnection pattern is the design parameter. The general optimization problem is combinatorially difficult. A heuristic is following the greedy method in \citep{Boyd2006Growing} that is developed for well-connected single layer networks. By previous analyses, it is evident that, for small budgets $c$, the optimal interconnection pattern is regular, if feasible. However, this is not the case for larger budget values, and we will see that, for a larger given total budget $c$, the optimal interconnection pattern is generally not a regular one, such as a $k$-to-$k$ coupling scheme--although it is feasible. \\

Our design parameter is choosing the inter-structure with given number of interlinks $r$. Based on the greedy approach in \citep{Boyd2006Growing}, we add $r$ edges one at a time, each time choosing the interlink $e\sim\{i,j\}$, between nodes $i\in V_1$ and $j\in V_2$, which has the largest value of $\left(v_i-v_j\right)^2$ with $v$ being a unit Fiedler vector of the current Laplacian. The motivation for this heuristic is that, when $\lambda_2$ is isolated, $w_0\left(v_i-v_j\right)^2$ gives the first order approximation of the increase in $\lambda_2$, if edge $e\sim\{i,j\}$ with weight $w_0$ is added to the graph. \\
%The first edge is selected as the interconnection between the highest-degree nodes in subgraphs $G_1$ and $G_2$. In fact, the high-degree nodes are preferred at first stage due to their ... \cite{Aguirre2014Synchronization}.

Figure \ref{fig:WellConnectSmall1} shows a small well-interconnected network. Here, supper-diffusion is not possible with a one-to-one interconnection since $\text{max}\left(\lambda_2\left[L_1\right],\lambda_2\left[L_2\right]\right)>\lambda_2\left[L_{ave}\right]$, and hence, the condition suggested by \citet{GomezDiffusion2013} is not met. The one-to-one interconnected pattern used as reference case in this section follows a multiplex setting \citep{Mieghem2019Regular}. However, it is observed that under a well-interconnected strategy, the diffusion in multilayer network immediately turns into super-diffusion after small values of budget $c$. In fact, the well-interconnected multilayer is achieved only at the expense of one change with respect to a  one-to-one interconnection pattern. This further  highlights the impact of interconnection pattern on structural properties of interdependent networks.   \\

\begin{figure*}
	\centering
	\subfloat[\label{fig:WelConSmal}]{\includegraphics[clip,width=.3\columnwidth]{WellConnectSmall.pdf}} 
	\subfloat[\label{fig:WelConSmalLam}]{\includegraphics[clip,width=.3\columnwidth]{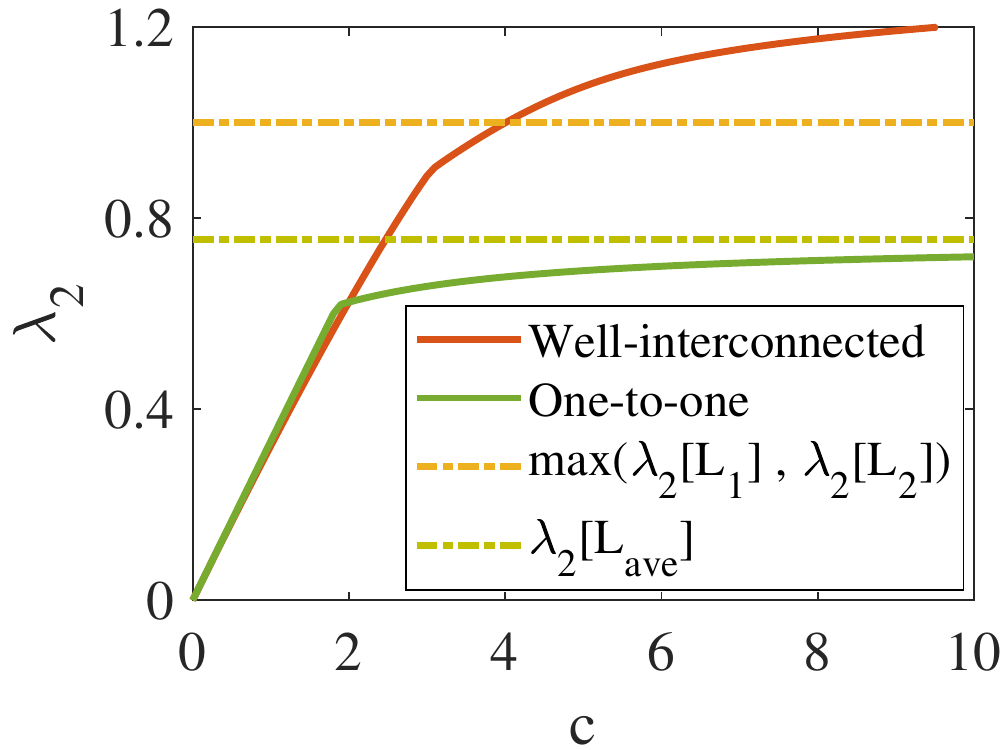}} 
	\caption{(a) Small well-interconnected multilayer network for $c=10$ with $n=m=6$ and $\lambda_2\left[L_1\right]=1, \lambda_2\left[L_2\right]=0.4384$, and (b) $\lambda_2$ as function of total budget $c$.}
%		The well-interconnected network is not one-to-one inter-connected since the Node 4 in Layer 2 holds no interlink while Node 5 in the same layer holds two interlinks. Therefor, with only one change with respect to a  one-to-one interconnection setting the small multilayer network reaches super-diffusion
	\label{fig:WellConnectSmall1} 
\end{figure*}

Figures \ref{fig:WellConnectSmall2} and \ref{fig:WellConnectSmall3} unravel more properties of well-interconnected networks. In Figure \ref{fig:WellConnectSmall2}, we note an interconnection pattern where nodes that are far from each other in an individual network component are bridged by nodes in the other layer. Therefore, overall interconnected network gains increased connectivity among its nodes compared to each isolated component. In a  one-to-one interconnection setting, such condition is possible only by close-to-orthogonal Fiedler vectors \citep{sahneh2014exact}. Moreover, Figure \ref{fig:WellConnectSmall3} unfolds an interconnection pattern that is essentially inhomogeneous as only few nodes in Layer 2 undergo interlinks. On the other hand, all nodes in Layer 1 are (equally) assigned interlink. What marks this situation is the large gap between algebraic connectivity values of individual network components, with Layer 2 holding the (very) larger algebraic connectivity. Furthermore, the Nodes 3 and 5 of Layer 2, i.e. the nodes assigned most interlinks in subgraph with larger connectivity, hold the most negative and positive values in this subgraph Fiedler vector, and each bridges nodes with different signs in Fiedler vector of Layer 1. \\

Figure \ref{fig:WellConnectMedium1} shows the results for two Geo networks each with 30 nodes. The feature is that, while the numbers of interlinks assigned in Layer 2 with smaller algebraic connectivity are more uniform and most nodes are assigned one interlink, the nodes in Layer 1 with larger algebraic connectivity are assigned more different interlinks. In fact, the nodes with most interlinks in Layer 1 are those corresponding to the most negative and positive values in this layer Fiedler vector. \\

\begin{figure*}
	\centering
	\subfloat[\label{fig:WelConSmal2}]{\includegraphics[clip,width=.35\columnwidth]{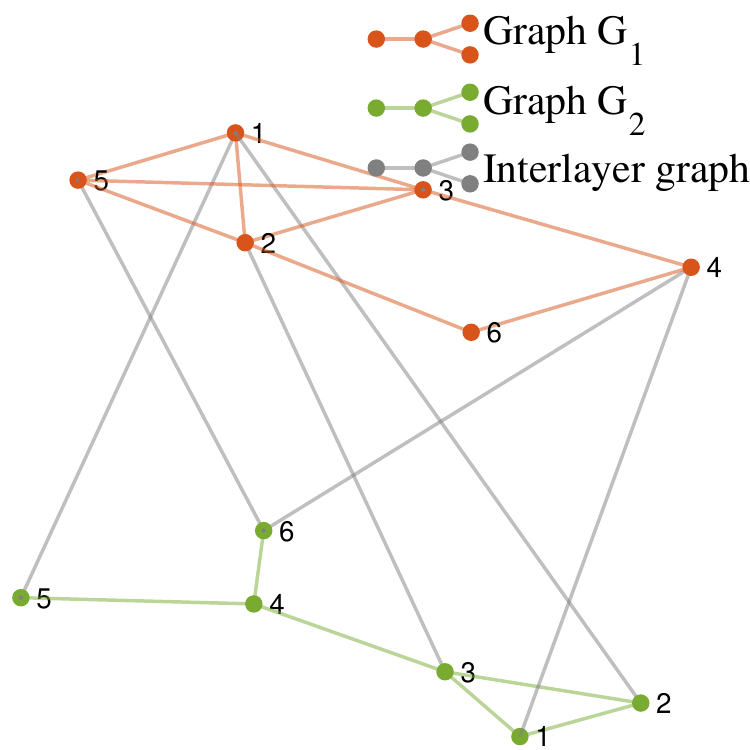}}
	\subfloat[\label{fig:WelConSmalLam2}]{\includegraphics[clip,width=.3\columnwidth]{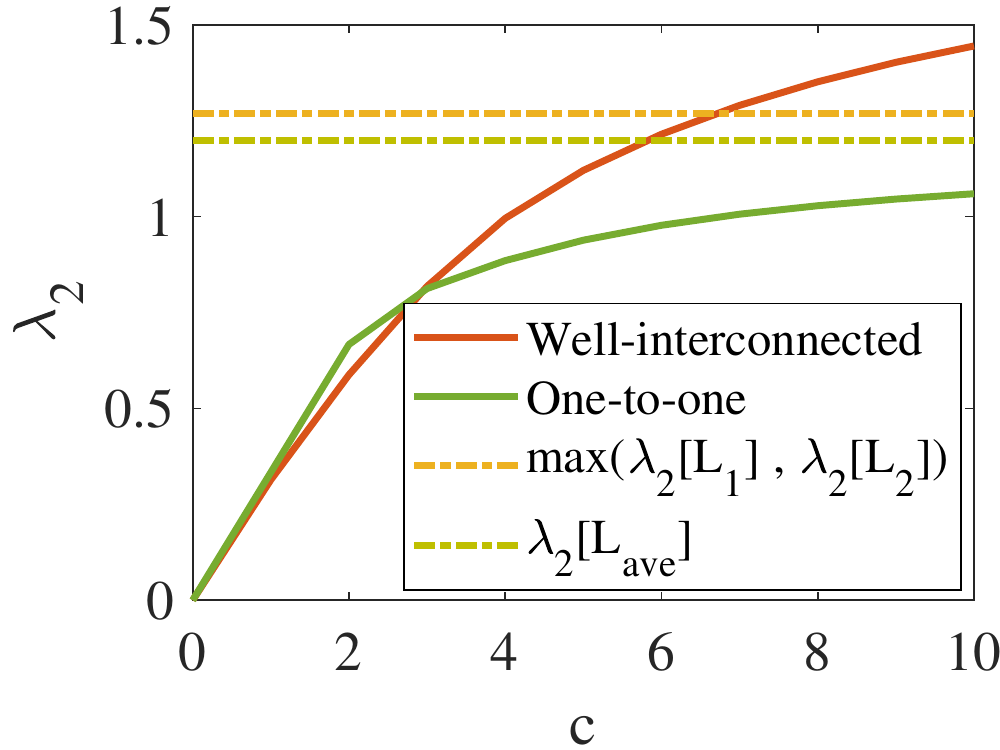}} 
	\caption{(a) Small well-interconnected multilayer network for $c=10$ with $n=m=6$ and $\lambda_2\left[L_1\right]=1.2679, \lambda_2\left[L_2\right]=0.4384$, and (b) $\lambda_2$ as function of total budget $c$. The well-interconnected strategy bridges the nodes that are far from each other in a subgraph. Therefore, Nodes 4 and 5 that are far from each other in Layer 1 are interconnected to the common Node 6 in Layer 2. In the same manner, the Node 1 in Layer 1 is interconnected with two far Nodes 2 and 5 in Layer 2, and the Node 4 in Layer 1 is interconnected to far Nodes 1 and 6 in Layer 2. Moreover, the Nodes 1 and 4 that are far in Layer 1 are interlinked to Nodes 1 and 2 that are close in Layer 1. Therefore, the diffusion between nodes that are far from each other in an individual network component speeds up by interconnecting to a common node, or some closer nodes, within the other layer.}
	\label{fig:WellConnectSmall2} 
\end{figure*}
\begin{figure*}
	\centering
	\subfloat[\label{fig:WelConSmal3}]{\includegraphics[clip,width=.35\columnwidth]{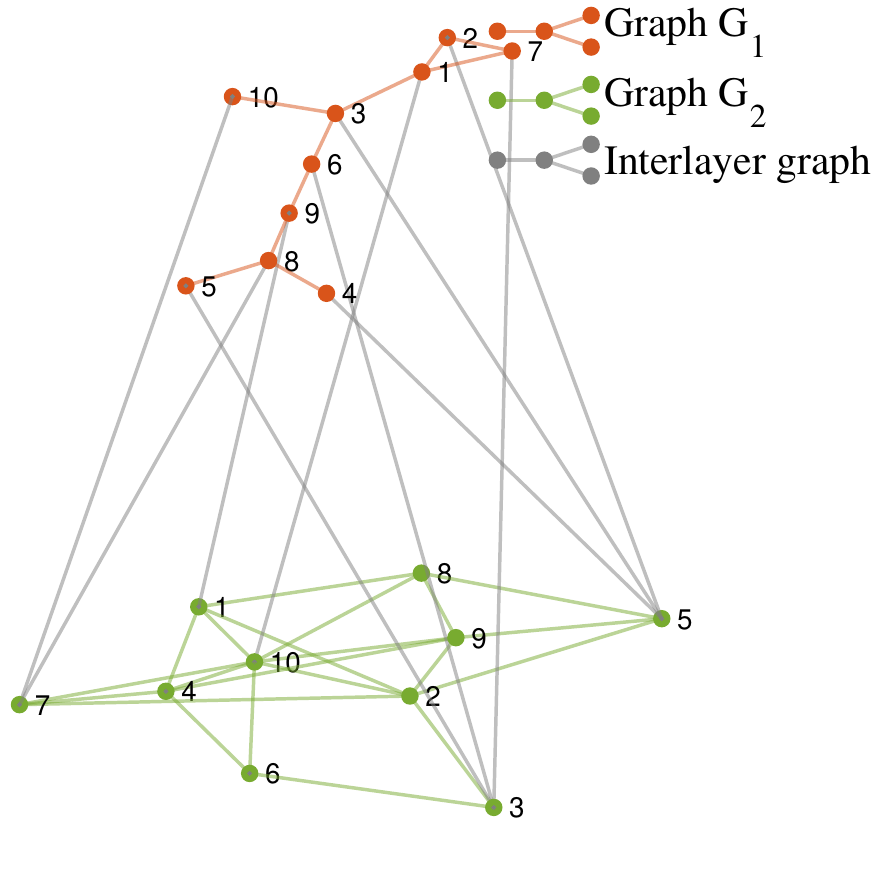}}
	\subfloat[\label{fig:WelConSmalLam3}]{\includegraphics[clip,width=.3\columnwidth]{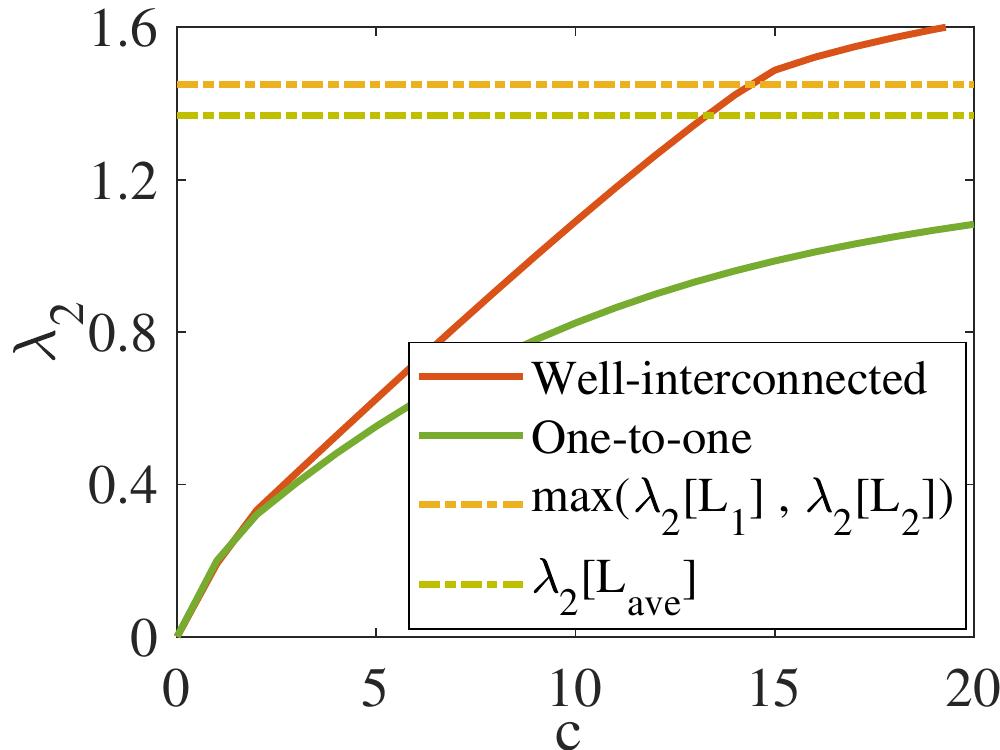}}
	\caption{(a) Small well-interconnected multilayer network for $c=20$ with $n=m=10$ and $\lambda_2\left[L_1\right]=0.1338, \lambda_2\left[L_2\right]=1.4498$, and (b) $\lambda_2$ as function of total budget $c$. Figure indicates an unbalanced interlink assignment strategy where the nodes of the set \{3, 5\} in Layer 2, with the very larger algebraic connectivity, undergo the most interlinks in this layer and bridge the nodes that are far from each other in Layer 1, having the very smaller algebraic connectivity.}
	\label{fig:WellConnectSmall3} 
\end{figure*}
\begin{figure*}
	\centering
	\subfloat[\label{fig:WellConnectMediumLam1}]{\includegraphics[clip,width=.3\columnwidth]{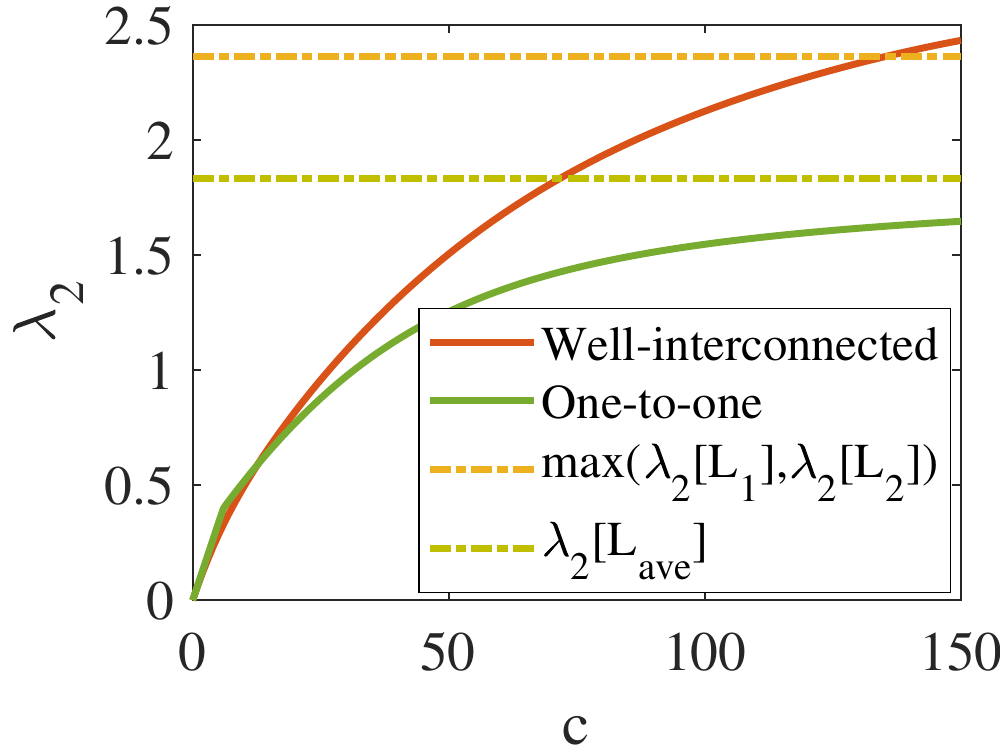}} \ \ \ \ \
	\subfloat[\label{fig:WellConnectMedIntrlnk1_1}]{\includegraphics[clip,width=.3\columnwidth]{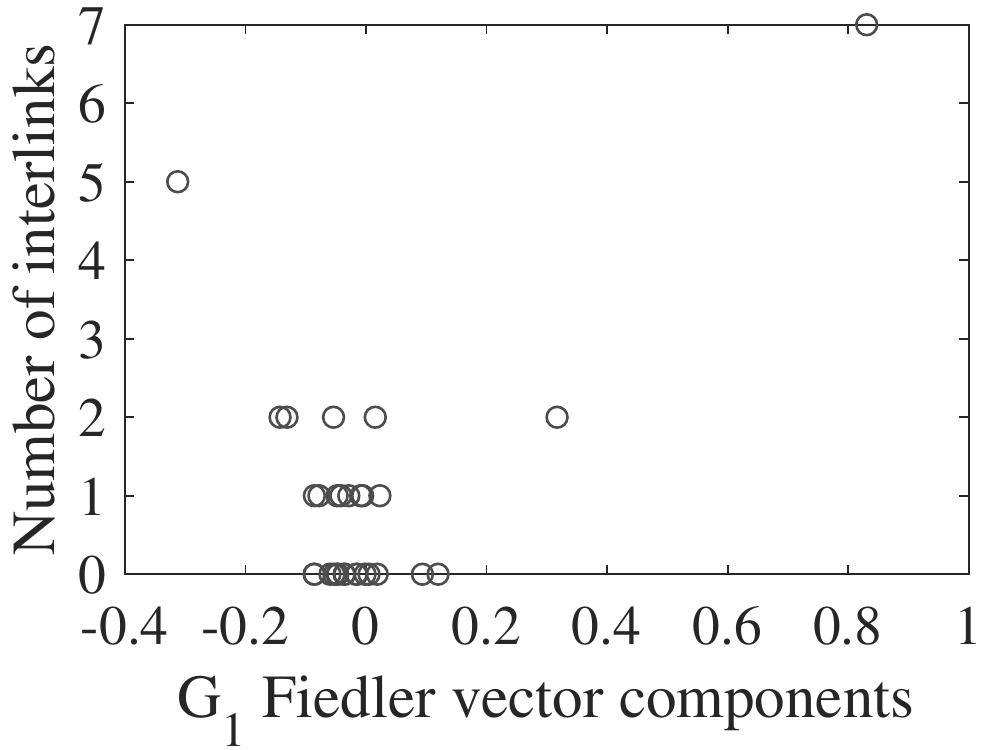}} \ \ \ \ \
	\subfloat[\label{fig:WellConnectMedIntrlnk2_1}]{\includegraphics[clip,width=.3\columnwidth]{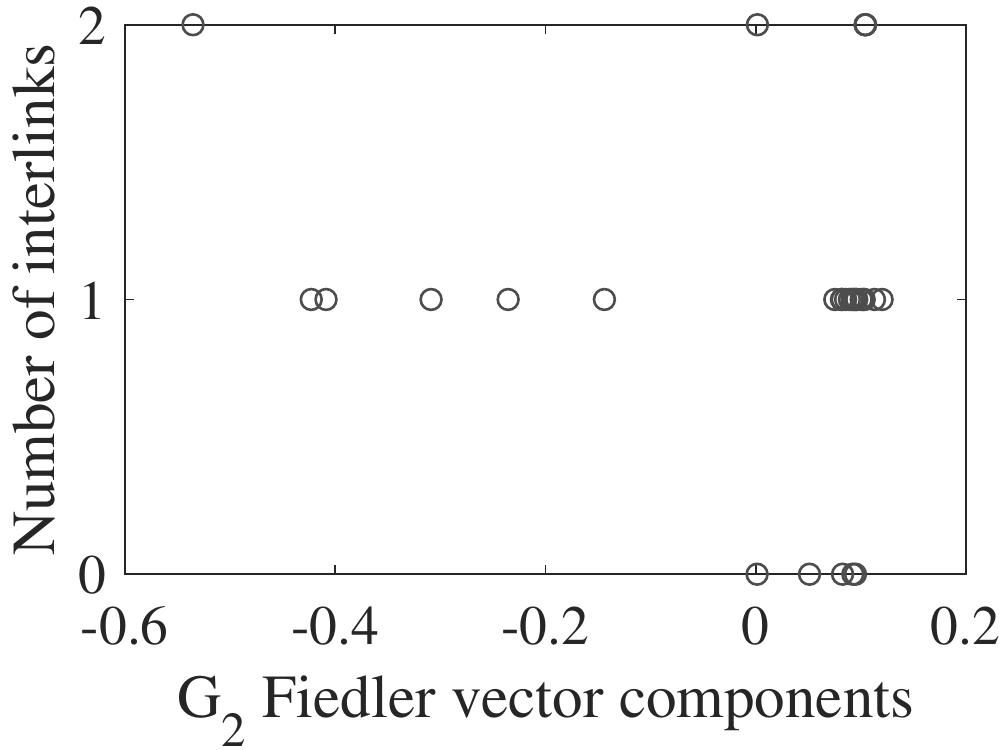}} 
	\caption{Well-interconnection of two Geo networks each with 30 nodes: (a) $\lambda_2$ as function of $c$, and number of interlinks for each node in Layer (b) 1 with larger algebraic connectivity ($\lambda_2=2.3621$), and (c) 2 with smaller algebraic connectivity ($\lambda_2=0.2101$).}
	\label{fig:WellConnectMedium1} 
\end{figure*}

In all above examples, the well-interconnected graph is not regular for the given $c$. To emphasize that the well-interconnected pattern depends on the budget $c$ given, we show in Figure \ref{fig:WellConnectSmall4} the situation where the interlinks may or may not be regular depending on $c$.

\begin{figure*}
	\centering
		\subfloat[\label{fig:WelConSmal4Regu}]{\includegraphics[clip,width=.35\columnwidth]{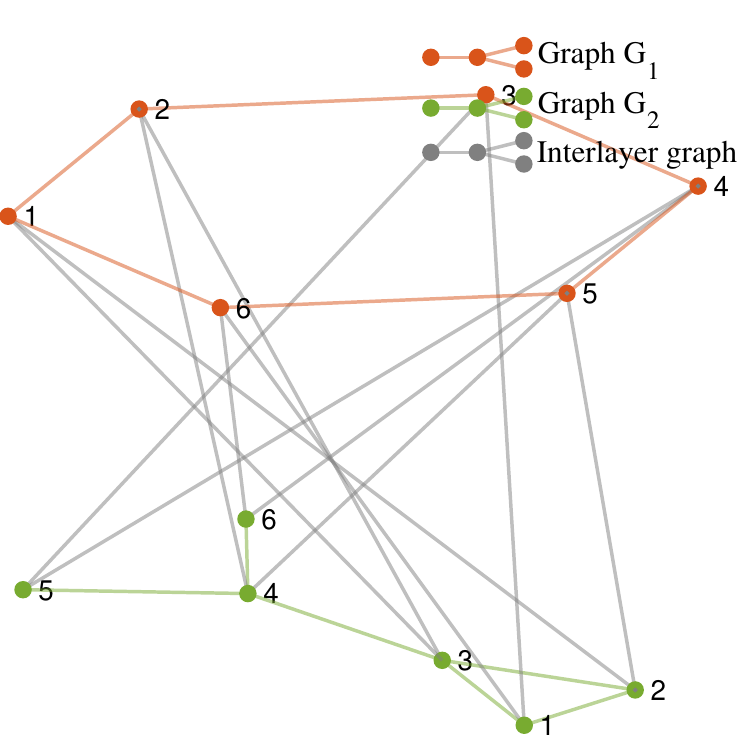}} \ \ \ \ \ \ \ \ \
		\subfloat[\label{fig:WelConSmal4Unregu}]{\includegraphics[clip,width=.35\columnwidth]{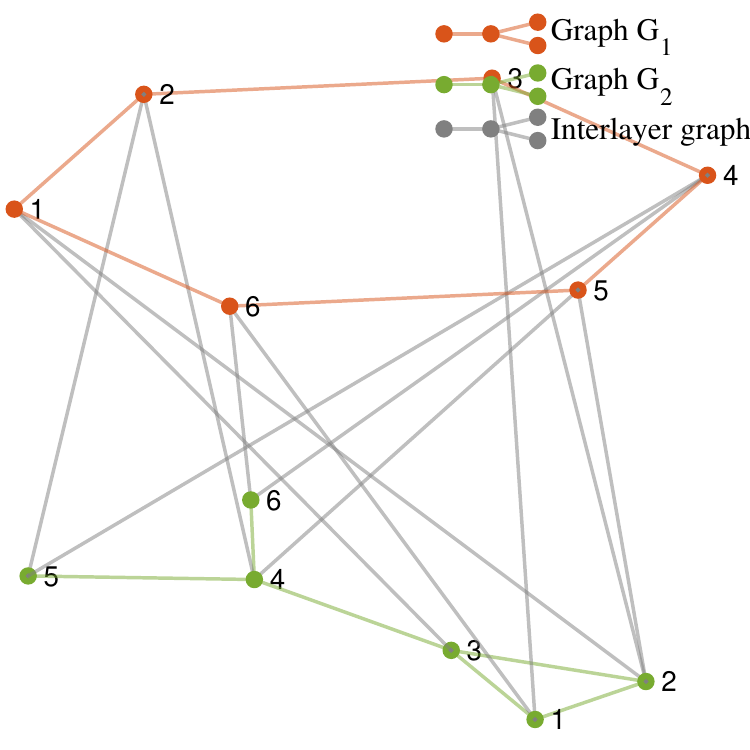}} 	 
		\caption{(a) Small well-interconnected multilayer network in Figure \ref{fig:WellConnectSmall1} revisited with $2n=12$ admissible interlinks: (a) for smaller budget $c=1$ the well-interconnected graph is a regular $k$-to-$k$ interconnection with $k=2$, and (b) for larger budget $c=2$ the well-interconnected graph is not regular.}
		\label{fig:WellConnectSmall4} 
	
\end{figure*}

%\section{Minimum interlinks required to reach a specified algebraic connectivity}
%What is the minimum number of interlinks to achieve a prescribed algebraic connectivity in a multilayer network? In general, this is a combinatorial optimization problem. To get an efficient approximate method to replace the combinatorial problem, we utilize the approach in \citep{Julius2009Gene}.

\section{Conclusion and discussion}\label{sec:Conclusion}
In this paper, we considered optimal interlayer weights and structure determination for maximizing the algebraic connectivity in multilayer networks with arbitrary interconnection pattern. We first investigated optimal weight distribution subject to a total budget $c$. Using an appropriate formulation of maximum algebraic connectivity problem, we showed analytically that before a known threshold budget $c^*$, the maximum algebraic connectivity is attainable subject to a set of regularity conditions, which may or may not lead to uniform weights depending on inter-structure pattern. For efficient numerical solution of larger budget $c>c^*$, we presented a convex formulation of the considered optimization problem. Under a primal-dual setting, we obtained a graph embedding problem that enables easier interpretations of some physical aspects. In particular, we found the graph embedding related to the phase of diffusion, as well as interlayer and intralayer interactions, over the multilayer graph. We showed that when $c\leq c^*$, the graph embedding is one dimensional and implies intralayer phase of optimum diffusion. When regularity is feasible in such a small budget condition, graph embedding involves nodes in same layers clumped together. The most apparent cases of this situation are the case of no restriction on interconnection pattern, i.e. when all nodes in one layer are allowed to be connected to all nodes in the other layer, and the case of $k$-to-$k$ interconnection pattern. For larger budgets beyond $c^*$, interactions between interlayer and intralayer connections result in graph embeddings of higher dimensions, and more diverse versions of intralayer and interlayer phases of optimal diffusion emerge. It was observed that while in an all-to-all interconnection pattern any interlink is possible, the optimal inter-structure may include many interlinks with zero weights. Using a perturbation analysis, we found out a positive correlation between optimal weights and Fiedler vector components of subgraphs.  When sorting several results, we also noted the role of specific algebraic connectivity, i.e. algebraic connectivity divided by number of nodes in a subgraph. Additionally, we investigated determination of optimal inter-structure that, for a given number of interlinks with identical weights, yields the maximum algebraic connectivity. In this regard, we concluded another role of subgraphs Fiedler vectors in identifying optimized inter-structures, which may or may not be regular depending on weight per number of interlinks. 

The findings of this study have far-reaching implications in designing and managing multi-layer systems where resiliency and robustness are of concern. One such system is the supply networks of commodities, which are generally referred to as supply chains. Supply chains are typically composed of a layered sequence of firms whose mission is to serve a final consumer by supplying each other with the material, information, and financial instruments. By nature, and in its generic form, each layer of a supply chain encompasses entities that form a network and perform comparable tasks (e.g., a network of manufacturers versus a network of distributors). A well-studied and common goal in designing supply chains is to serve a market by connecting various layers of firms (i.e., networks) while minimizing logistics costs or the chain's susceptibility to disruptions. This study can be leveraged in devising resilient supply chains that are reasonably immune to supply disturbances and, hence, resistant against catastrophic cascading effects. The recent COVID-19 (SARS-CoV-2) outbreak, which was first documented in Wuhan, China, and the resulting supply disruptions, provide an enlightening example to elucidate the significance of this study in supply chains. In the aftermath of the Coronavirus pandemic, a staggering number of companies reported severe logistics-related disturbances. 94\% of the Fortune 1000 firms described supply chain interruptions \cite{ivanov2020predicting}. Moreover, COVID-19's quarantined areas are home to roughly 12,000 facilities of the world's largest 1000 supply chains \cite{ivanov2020viability}. Considering approximately five million companies with supply roots in Wuhan \cite{ivanov2020predicting}, one begins to realize the significance of the theoretical insights proposed in this study in relation to the supply chains' resilience-driven design.

Besides designing issues, supply chains offer a wealth of tactical and operational problems, structured in a multi-layered fashion, that may benefit from this study. For example, one may consider a distribution problem where a network of warehouses (or fulfillment centers) is to serve a consumer base bound by a transportation network. Warehouses can be either independent or connected through transshipment. A common issue in this scenario is to identify the set of consumers supported by each warehouse, and similarly, the collection of warehouses that may serve a customer. The findings of this study can be utilized to connect the two warehouses and customers' networks, such that the consumers are best served in the face of inventory shortages or commodities' flow interruptions.

Transportation planning problems are another fertile territory for applying the findings of this study. In this context, a decision-maker may be interested in connecting transportation networks of different modes. For example, one widely-studied shipment planning problem is to identify the set of nodes where two distribution networks with varying modes of delivery collide. A few examples of such a problem are connecting a subway system and networks of streets as well as adjoining a truckload and an air freight shipment system. In each of these scenarios, the current study can be utilized to improve the resiliency of the overall connected system by preserving the flow or limiting the disruptions under unforeseen circumstances.

\bibliographystyle{chicago}
\bibliography{refs}

\begin{thebibliography}{}

\bibitem[\protect\citeauthoryear{Aguirre, Sevilla-Escoboza, Guti\'errez, Papo,
  and Buld\'u}{Aguirre et~al.}{2014}]{Aguirre2014Synchronization}
Aguirre, J., R.~Sevilla-Escoboza, R.~Guti\'errez, D.~Papo, and J.~M. Buld\'u
  (2014, Jun).
\newblock Synchronization of interconnected networks: The role of connector
  nodes.
\newblock {\em Phys. Rev. Lett.\/}~{\em 112}, 248701.

\bibitem[\protect\citeauthoryear{Arenas, Díaz-Guilera, Kurths, Moreno, and
  Zhou}{Arenas et~al.}{2008}]{Arena2008Sync}
Arenas, A., A.~Díaz-Guilera, J.~Kurths, Y.~Moreno, and C.~Zhou (2008).
\newblock Synchronization in complex networks.
\newblock {\em Physics Reports\/}~{\em 469\/}(3), 93 -- 153.

\bibitem[\protect\citeauthoryear{Boccaletti, Bianconi, Criado, Del~Genio,
  G{\'o}mez-Gardenes, Romance, Sendina-Nadal, Wang, and Zanin}{Boccaletti
  et~al.}{2014}]{boccaletti2014structure}
Boccaletti, S., G.~Bianconi, R.~Criado, C.~I. Del~Genio, J.~G{\'o}mez-Gardenes,
  M.~Romance, I.~Sendina-Nadal, Z.~Wang, and M.~Zanin (2014).
\newblock The structure and dynamics of multilayer networks.
\newblock {\em Physics Reports\/}~{\em 544\/}(1), 1--122.

\bibitem[\protect\citeauthoryear{Boccaletti, Latorab, Morenod, Chavez, and
  Hwang}{Boccaletti et~al.}{2006}]{Boccalettia2006}
Boccaletti, S., V.~Latorab, Y.~Morenod, M.~Chavez, and D.-U. Hwang (2006).
\newblock Complex networks: Structure and dynamics.
\newblock {\em Physics Reports\/}~{\em 424}, 175--308.

\bibitem[\protect\citeauthoryear{Borgatti and Li}{Borgatti and
  Li}{2009}]{borgatti2009social}
Borgatti, S.~P. and X.~Li (2009).
\newblock On social network analysis in a supply chain context.
\newblock {\em Journal of Supply Chain Management\/}~{\em 45\/}(2), 5--22.

\bibitem[\protect\citeauthoryear{Boyd, Diaconis, and Xiao}{Boyd
  et~al.}{2004}]{boyd2004fastestChain}
Boyd, S., P.~Diaconis, and L.~Xiao (2004).
\newblock Fastest mixing markov chain on a graph.
\newblock {\em SIAM review\/}~{\em 46\/}(4), 667--689.

\bibitem[\protect\citeauthoryear{Boyd and Vandenberghe}{Boyd and
  Vandenberghe}{2004}]{boyd2004convex}
Boyd, S.~P. and L.~Vandenberghe (2004).
\newblock {\em Convex optimization}.
\newblock Cambridge University Press.

\bibitem[\protect\citeauthoryear{Br{\'e}maud}{Br{\'e}maud}{2013}]{bremaud2013markov}
Br{\'e}maud, P. (2013).
\newblock {\em Markov chains: Gibbs fields, Monte Carlo simulation, and
  queues}, Volume~31.
\newblock Springer Science \& Business Media.

\bibitem[\protect\citeauthoryear{Buld{\'u}, Busquets, Mart{\'\i}nez,
  Herrera-Diestra, Echegoyen, Galeano, and Luque}{Buld{\'u}
  et~al.}{2018}]{buldu2018using}
Buld{\'u}, J.~M., J.~Busquets, J.~H. Mart{\'\i}nez, J.~L. Herrera-Diestra,
  I.~Echegoyen, J.~Galeano, and J.~Luque (2018).
\newblock Using network science to analyse football passing networks: Dynamics,
  space, time, and the multilayer nature of the game.
\newblock {\em Frontiers in psychology\/}~{\em 9}, 1900.

\bibitem[\protect\citeauthoryear{Buldyrev, Parshani, Paul, Stanley, and
  Havlin}{Buldyrev et~al.}{2010}]{Buldyrev2010}
Buldyrev, S.~V., R.~Parshani, G.~Paul, H.~E. Stanley, and S.~Havlin (2010).
\newblock Catastrophic cascade of failures in interdependent networks.
\newblock {\em Nature\/}~{\em 464\/}(7291), 1025--1028.

\bibitem[\protect\citeauthoryear{Cencetti and Battiston}{Cencetti and
  Battiston}{2019}]{Cencetti2019Diffusive}
Cencetti, G. and F.~Battiston (2019, mar).
\newblock Diffusive behavior of multiplex networks.
\newblock {\em New Journal of Physics\/}~{\em 21\/}(3), 035006.

\bibitem[\protect\citeauthoryear{{Chattopadhyay}, {Dai}, and
  {Eun}}{{Chattopadhyay} et~al.}{2019}]{Chattopadhyay2019Robustness}
{Chattopadhyay}, S., H.~{Dai}, and D.~Y. {Eun} (2019).
\newblock Maximization of robustness of interdependent networks under budget
  constraints.
\newblock {\em IEEE Transactions on Network Science and Engineering\/}, 1--1.

\bibitem[\protect\citeauthoryear{Cozzo, de~Arruda, Rodrigues, and Moreno}{Cozzo
  et~al.}{2019}]{Cozzo2019LayDegrd}
Cozzo, E., G.~F. de~Arruda, F.~A. Rodrigues, and Y.~Moreno (2019, Jul).
\newblock Layer degradation triggers an abrupt structural transition in
  multiplex networks.
\newblock {\em Phys. Rev. E\/}~{\em 100}, 012313.

\bibitem[\protect\citeauthoryear{Cozzo, Kivelä, Domenico, Sol{\'{e}}-Ribalta,
  Arenas, G{\'{o}}mez, Porter, and Moreno}{Cozzo
  et~al.}{2015}]{Cozzo2015Triadic}
Cozzo, E., M.~Kivelä, M.~D. Domenico, A.~Sol{\'{e}}-Ribalta, A.~Arenas,
  S.~G{\'{o}}mez, M.~A. Porter, and Y.~Moreno (2015, jul).
\newblock Structure of triadic relations in multiplex networks.
\newblock {\em New Journal of Physics\/}~{\em 17\/}(7), 073029.

\bibitem[\protect\citeauthoryear{Darabi~Sahneh, Scoglio, and
  Van~Mieghem}{Darabi~Sahneh et~al.}{2015}]{sahneh2014exact}
Darabi~Sahneh, F., C.~Scoglio, and P.~Van~Mieghem (2015, Oct).
\newblock Exact coupling threshold for structural transition reveals
  diversified behaviors in interconnected networks.
\newblock {\em Phys. Rev. E\/}~{\em 92}, 040801.

\bibitem[\protect\citeauthoryear{de~Arruda, Rodrigues, and Moreno}{de~Arruda
  et~al.}{2018}]{Arruda2018spreading}
de~Arruda, G.~F., F.~A. Rodrigues, and Y.~Moreno (2018).
\newblock Fundamentals of spreading processes in single and multilayer complex
  networks.
\newblock {\em Physics Reports\/}~{\em 756}, 1 -- 59.
\newblock Fundamentals of spreading processes in single and multilayer complex
  networks.

\bibitem[\protect\citeauthoryear{De~Domenico, Sol{\'e}-Ribalta, G{\'o}mez, and
  Arenas}{De~Domenico et~al.}{2014}]{Domenico2014Navigability}
De~Domenico, M., A.~Sol{\'e}-Ribalta, S.~G{\'o}mez, and A.~Arenas (2014).
\newblock Navigability of interconnected networks under random failures.
\newblock {\em Proceedings of the National Academy of Sciences\/}~{\em
  111\/}(23), 8351--8356.

\bibitem[\protect\citeauthoryear{Dickison, Havlin, and Stanley}{Dickison
  et~al.}{2012}]{Dickison2012Epidemics}
Dickison, M., S.~Havlin, and H.~E. Stanley (2012, Jun).
\newblock Epidemics on interconnected networks.
\newblock {\em Phys. Rev. E\/}~{\em 85}, 066109.

\bibitem[\protect\citeauthoryear{Diestel}{Diestel}{2017}]{Diestel2017}
Diestel, R. (2017).
\newblock {\em Graph Theory\/} (5 ed.).
\newblock Springer Graduate Texts in Mathematics (GTM). Hamburg, Germany.

\bibitem[\protect\citeauthoryear{Estrada and G\'omez-Garde\~nes}{Estrada and
  G\'omez-Garde\~nes}{2014}]{Estrada2014}
Estrada, E. and J.~G\'omez-Garde\~nes (2014, Apr).
\newblock Communicability reveals a transition to coordinated behavior in
  multiplex networks.
\newblock {\em Phys. Rev. E\/}~{\em 89}, 042819.

\bibitem[\protect\citeauthoryear{Fiedler}{Fiedler}{1973}]{Fiedler1973}
Fiedler, M. (1973).
\newblock Algebraic connectivity of graphs.
\newblock {\em Czechoslovak mathematical journal\/}~{\em 23\/}(2), 298--305.

\bibitem[\protect\citeauthoryear{Gao, Buldyrev, Havlin, and Stanley}{Gao
  et~al.}{2012}]{Gao2012Robustness}
Gao, J., S.~V. Buldyrev, S.~Havlin, and H.~E. Stanley (2012, Jun).
\newblock Robustness of a network formed by $n$ interdependent networks with a
  one-to-one correspondence of dependent nodes.
\newblock {\em Phys. Rev. E\/}~{\em 85}, 066134.

\bibitem[\protect\citeauthoryear{Ghosh and Boyd}{Ghosh and
  Boyd}{2006}]{Boyd2006Growing}
Ghosh, A. and S.~Boyd (2006).
\newblock Growing well-connected graphs.
\newblock In {\em Decision and Control, 2006 45th IEEE Conference on}, pp.\
  6605--6611. IEEE.

\bibitem[\protect\citeauthoryear{Gomez, Diaz-Guilera, Gomez-Gardenes,
  Perez-Vicente, Moreno, and Arenas}{Gomez et~al.}{2013}]{GomezDiffusion2013}
Gomez, S., A.~Diaz-Guilera, J.~Gomez-Gardenes, C.~J. Perez-Vicente, Y.~Moreno,
  and A.~Arenas (2013).
\newblock Diffusion dynamics on multiplex networks.
\newblock {\em Physical Review Letters\/}~{\em 110}, 028701(5).

\bibitem[\protect\citeauthoryear{Gomez, Diaz-Guilera, Gomez-Garde{\~n}es,
  Perez-Vicente, Moreno, and Arenas}{Gomez et~al.}{2013}]{gomez2013diffusion}
Gomez, S., A.~Diaz-Guilera, J.~Gomez-Garde{\~n}es, C.~J. Perez-Vicente,
  Y.~Moreno, and A.~Arenas (2013).
\newblock Diffusion dynamics on multiplex networks.
\newblock {\em Physical review letters\/}~{\em 110\/}(2), 028701.

\bibitem[\protect\citeauthoryear{Goring, Helmberg, and Wappler}{Goring
  et~al.}{2008}]{GoringShadowSeperator}
Goring, F., C.~Helmberg, and M.~Wappler (2008).
\newblock Embedded in the shadow of the separator.
\newblock {\em SIAM Journal on Optimization\/}~{\em 19\/}(1), 472--501.

\bibitem[\protect\citeauthoryear{Goring, Helmberg, and Wappler}{Goring
  et~al.}{2011}]{goring2011rotational}
Goring, F., C.~Helmberg, and M.~Wappler (2011).
\newblock The rotational dimension of a graph.
\newblock {\em Journal of Graph Theory\/}~{\em 66\/}(4), 283--302.

\bibitem[\protect\citeauthoryear{Grant and Boyd}{Grant and
  Boyd}{2009}]{grant2009cvx}
Grant, M. and S.~Boyd (2009).
\newblock Cvx: Matlab software for disciplined convex programming (web page and
  software).

\bibitem[\protect\citeauthoryear{Helmberg and Reiss}{Helmberg and
  Reiss}{2010}]{Helmberg2010}
Helmberg, C. and S.~Reiss (2010).
\newblock A note on fiedler vectors interpreted as graph realizations.
\newblock {\em Operations Research Letters\/}~{\em 38}, 320--321.

\bibitem[\protect\citeauthoryear{Helmberg, Rendl, Mohar, and Poljak}{Helmberg
  et~al.}{1995}]{Helmberg1995Bounds}
Helmberg, C., F.~Rendl, B.~Mohar, and S.~Poljak (1995).
\newblock A spectral approach to bandwidth and separator problems in graphs.
\newblock {\em Linear and Multilinear Algebra\/}~{\em 39\/}(1-2), 73--90.

\bibitem[\protect\citeauthoryear{Ivanov}{Ivanov}{2020}]{ivanov2020predicting}
Ivanov, D. (2020).
\newblock Predicting the impacts of epidemic outbreaks on global supply chains:
  A simulation-based analysis on the coronavirus outbreak (covid-19/sars-cov-2)
  case.
\newblock {\em Transportation Research Part E: Logistics and Transportation
  Review\/}~{\em 136}, 101922.

\bibitem[\protect\citeauthoryear{Ivanov and Dolgui}{Ivanov and
  Dolgui}{2020}]{ivanov2020viability}
Ivanov, D. and A.~Dolgui (2020).
\newblock Viability of intertwined supply networks: extending the supply chain
  resilience angles towards survivability. a position paper motivated by
  covid-19 outbreak.
\newblock {\em International Journal of Production Research\/}~{\em 58\/}(10),
  2904--2915.

\bibitem[\protect\citeauthoryear{Jadbabaie, Lin, and Morse}{Jadbabaie
  et~al.}{2003}]{Jadbabaie2003NearestNeighbor}
Jadbabaie, A., J.~Lin, and A.~S. Morse (2003).
\newblock Coordination of groups of mobile autonomous agents using nearest
  neighbor rules.
\newblock {\em IEEE Transactions on Automatic Control\/}~{\em 48\/}(6),
  1520--1533.

\bibitem[\protect\citeauthoryear{Jamakovic and Uhlig}{Jamakovic and
  Uhlig}{2007}]{Jamakovic2007}
Jamakovic, A. and S.~Uhlig (2007).
\newblock On the relationship between the algebraic connectivity and graph's
  robustness to node and link failures.
\newblock In {\em Next Generation Internet Networks, 3rd EuroNGI Conference
  on}, pp.\  96--102. IEEE.

\bibitem[\protect\citeauthoryear{Juvan and Mohar}{Juvan and
  Mohar}{1993}]{Juvan1993Bandwidth}
Juvan, M. and B.~Mohar (1993).
\newblock Laplace eigenvalues and bandwidth-type invariants of graphs.
\newblock {\em Journal of Graph Theory\/}~{\em 17\/}(3), 393--407.

\bibitem[\protect\citeauthoryear{Kim, Choi, Yan, and Dooley}{Kim
  et~al.}{2011}]{kim2011structural}
Kim, Y., T.~Y. Choi, T.~Yan, and K.~Dooley (2011).
\newblock Structural investigation of supply networks: A social network
  analysis approach.
\newblock {\em Journal of Operations Management\/}~{\em 29\/}(3), 194--211.

\bibitem[\protect\citeauthoryear{Kivel{\"a}, Arenas, Barthelemy, Gleeson,
  Moreno, and Porter}{Kivel{\"a} et~al.}{2014}]{Arenas2014multilayer}
Kivel{\"a}, M., A.~Arenas, M.~Barthelemy, J.~P. Gleeson, Y.~Moreno, and M.~A.
  Porter (2014).
\newblock Multilayer networks.
\newblock {\em Journal of Complex Networks\/}~{\em 2\/}(3), 203--271.

\bibitem[\protect\citeauthoryear{Kryven and Bianconi}{Kryven and
  Bianconi}{2019}]{Kryven2019Percolation}
Kryven, I. and G.~Bianconi (2019, Aug).
\newblock Enhancing the robustness of a multiplex network leads to multiple
  discontinuous percolation transitions.
\newblock {\em Phys. Rev. E\/}~{\em 100}, 020301.

\bibitem[\protect\citeauthoryear{Li, Wu, Scoglio, and Gruenbacher}{Li
  et~al.}{2015}]{Li2015RobustAllocation}
Li, X., H.~Wu, C.~Scoglio, and D.~Gruenbacher (2015).
\newblock Robust allocation of weighted dependency links in cyber–physical
  networks.
\newblock {\em Physica A: Statistical Mechanics and its Applications\/}~{\em
  433\/}(1), 316--327.

\bibitem[\protect\citeauthoryear{Mart{\'\i}n-Hern{\'a}ndez, Wang, Van~Mieghem,
  and D’Agostino}{Mart{\'\i}n-Hern{\'a}ndez
  et~al.}{2014}]{Gregorio2014algebraic}
Mart{\'\i}n-Hern{\'a}ndez, J., H.~Wang, P.~Van~Mieghem, and G.~D’Agostino
  (2014).
\newblock Algebraic connectivity of interdependent networks.
\newblock {\em Physica A: Statistical Mechanics and its Applications\/}~{\em
  404}, 92--105.

\bibitem[\protect\citeauthoryear{Mesbahi and Egerstedt}{Mesbahi and
  Egerstedt}{2010}]{Mesbahi2010}
Mesbahi, M. and M.~Egerstedt (2010).
\newblock {\em Graph Theoretic Methods in Multiagent Networks}.
\newblock Princeton, New Jersey: Princeton University Press.

\bibitem[\protect\citeauthoryear{Min, Yi, Lee, and Goh}{Min
  et~al.}{2014}]{Min2014Robustness}
Min, B., S.~D. Yi, K.-M. Lee, and K.-I. Goh (2014, Apr).
\newblock Network robustness of multiplex networks with interlayer degree
  correlations.
\newblock {\em Phys. Rev. E\/}~{\em 89}, 042811.

\bibitem[\protect\citeauthoryear{Moothedath, Chaporkar, and Belur}{Moothedath
  et~al.}{2019}]{Moothedath2019Controllability}
Moothedath, S., P.~Chaporkar, and M.~N. Belur (2019).
\newblock Optimal selection of essential interconnections for structural
  controllability in heterogeneous subsystems.
\newblock {\em Automatica\/}~{\em 103}, 424 -- 434.

\bibitem[\protect\citeauthoryear{Olfati-Saber}{Olfati-Saber}{2006}]{Olfati-SaberFlock}
Olfati-Saber, R. (2006).
\newblock Flocking for multi-agent dynamic systems: algorithms and theory.
\newblock {\em IEEE Transactions on Automatic Control\/}~{\em 51\/}(3),
  401--420.

\bibitem[\protect\citeauthoryear{Olfati-Saber and Murray}{Olfati-Saber and
  Murray}{2004}]{Olfati-Saber2004Consensus}
Olfati-Saber, R. and R.~M. Murray (2004).
\newblock Consensus problems in networks of agents with switching topology and
  time-delays.
\newblock {\em IEEE Transactions on Automatic Control\/}~{\em 49\/}(9),
  1520--1533.

\bibitem[\protect\citeauthoryear{Pan, Wang, Cai, and Zhou}{Pan
  et~al.}{2019}]{Pan2019Spreading}
Pan, L., W.~Wang, S.~Cai, and T.~Zhou (2019, Aug).
\newblock Optimal interlayer structure for promoting spreading of the
  susceptible-infected-susceptible model in two-layer networks.
\newblock {\em Phys. Rev. E\/}~{\em 100}, 022316.

\bibitem[\protect\citeauthoryear{Radicchi}{Radicchi}{2014}]{Radicchi2014Supercritical}
Radicchi, F. (2014, Apr).
\newblock Driving interconnected networks to supercriticality.
\newblock {\em Phys. Rev. X\/}~{\em 4}, 021014.

\bibitem[\protect\citeauthoryear{Radicchi and Arenas}{Radicchi and
  Arenas}{2013}]{Radicchi2013}
Radicchi, F. and A.~Arenas (2013).
\newblock Abrupt transition in the structural formation of interconnected
  networks.
\newblock {\em Nature Physics\/}~{\em 9\/}(11), 717--720.

\bibitem[\protect\citeauthoryear{Rapisardi, Arenas, Caldarelli, and
  Cimini}{Rapisardi et~al.}{2018}]{Arena2018Multiple}
Rapisardi, G., A.~Arenas, G.~Caldarelli, and G.~Cimini (2018, Jul).
\newblock Multiple structural transitions in interacting networks.
\newblock {\em Phys. Rev. E\/}~{\em 98}, 012302.

\bibitem[\protect\citeauthoryear{Sahneh, Chowdhury, and Scoglio}{Sahneh
  et~al.}{2012}]{Sahneh2012}
Sahneh, F.~D., F.~N. Chowdhury, and C.~M. Scoglio (2012, September).
\newblock On the existence of a threshold for preventive behavioral responses
  to suppress epidemic spreading.
\newblock {\em Sci. Rep.\/}~{\em 2}, --.

\bibitem[\protect\citeauthoryear{Saumell-Mendiola, Serrano, and
  Bogu\~n\'a}{Saumell-Mendiola et~al.}{2012}]{Saumell2012Epidemic}
Saumell-Mendiola, A., M.~A. Serrano, and M.~Bogu\~n\'a (2012, Aug).
\newblock Epidemic spreading on interconnected networks.
\newblock {\em Phys. Rev. E\/}~{\em 86}, 026106.

\bibitem[\protect\citeauthoryear{Shakeri, Albin, Darabi~Sahneh,
  Poggi-Corradini, and Scoglio}{Shakeri et~al.}{2016}]{shakeri2015PRL}
Shakeri, H., N.~Albin, F.~Darabi~Sahneh, P.~Poggi-Corradini, and C.~Scoglio
  (2016, Mar).
\newblock Maximizing algebraic connectivity in interconnected networks.
\newblock {\em Phys. Rev. E\/}~{\em 93}, 030301.

\bibitem[\protect\citeauthoryear{Shakeri, Tavasoli, Ardjmand, and
  Poggi-Corradini}{Shakeri et~al.}{2020}]{shakeri2020designing}
Shakeri, H., A.~Tavasoli, E.~Ardjmand, and P.~Poggi-Corradini (2020).
\newblock Designing optimal multiplex networks for certain laplacian spectral
  properties.
\newblock {\em Physical Review E\/}~{\em 102\/}(2), 022302.

\bibitem[\protect\citeauthoryear{Shames}{Shames}{1996}]{Shames1996}
Shames, I.~H. (1996).
\newblock {\em Engineering mechanics: statics and dynamics\/} (4 ed.).
\newblock Prentice Hall. New Jersey.

\bibitem[\protect\citeauthoryear{Son, Bizhani, Christensen, Grassberger, and
  Paczuski}{Son et~al.}{2012}]{Son2012Percolation}
Son, S.-W., G.~Bizhani, C.~Christensen, P.~Grassberger, and M.~Paczuski (2012,
  jan).
\newblock Percolation theory on interdependent networks based on epidemic
  spreading.
\newblock {\em {EPL} (Europhysics Letters)\/}~{\em 97\/}(1), 16006.

\bibitem[\protect\citeauthoryear{Sonntag, Borgnakke, and Wylen}{Sonntag
  et~al.}{2009}]{Sonntag2009}
Sonntag, R.~E., C.~Borgnakke, and G.~J.~V. Wylen (2009).
\newblock {\em Fundamentals of Thermodynamics\/} (7 ed.).
\newblock New Jersey: Wiley.

\bibitem[\protect\citeauthoryear{Strogatz}{Strogatz}{2001}]{Strogatz2001}
Strogatz, S.~H. (2001).
\newblock Exploring complex networks.
\newblock {\em Nature (London)\/}~{\em 410}, 268–276.

\bibitem[\protect\citeauthoryear{Sun, Boyd, Xiao, and Diaconis}{Sun
  et~al.}{2006}]{Boyd2006FastestMP}
Sun, J., S.~Boyd, L.~Xiao, and P.~Diaconis (2006).
\newblock The fastest mixing markov process on a graph and a connection to a
  maximum variance unfolding problem.
\newblock {\em SIAM review\/}~{\em 48\/}(4), 681--699.

\bibitem[\protect\citeauthoryear{Tanner, Jadbabaie, and Pappas}{Tanner
  et~al.}{2007}]{Jadbabaie2007Flocking}
Tanner, H.~G., A.~Jadbabaie, and G.~J. Pappas (2007).
\newblock Flocking in fixed and switching networks.
\newblock {\em IEEE Transactions on Automatic control\/}~{\em 52\/}(5),
  863--868.

\bibitem[\protect\citeauthoryear{Tejedor, Longjas, Foufoula-Georgiou, Georgiou,
  and Moreno}{Tejedor et~al.}{2018}]{Tejedor2018Directed}
Tejedor, A., A.~Longjas, E.~Foufoula-Georgiou, T.~T. Georgiou, and Y.~Moreno
  (2018, Sep).
\newblock Diffusion dynamics and optimal coupling in multiplex networks with
  directed layers.
\newblock {\em Phys. Rev. X\/}~{\em 8}, 031071.

\bibitem[\protect\citeauthoryear{Van~Mieghem}{Van~Mieghem}{2010}]{PietBook}
Van~Mieghem, P. (2010).
\newblock {\em Graph spectra for complex networks}.
\newblock Cambridge University Press.

\bibitem[\protect\citeauthoryear{Van~Mieghem}{Van~Mieghem}{2016}]{Mieghem2017Interdepen}
Van~Mieghem, P. (2016, Apr).
\newblock Interconnectivity structure of a general interdependent network.
\newblock {\em Phys. Rev. E\/}~{\em 93}, 042305.

\bibitem[\protect\citeauthoryear{Vandenberghe and Boyd}{Vandenberghe and
  Boyd}{1996}]{Boyd96SDP}
Vandenberghe, L. and S.~Boyd (1996).
\newblock Semidefinite programming.
\newblock {\em SIAM Review\/}~{\em 38\/}(1), 49--95.

\bibitem[\protect\citeauthoryear{{Wang}, {Wen}, {Yu}, {Yu}, and {Huang}}{{Wang}
  et~al.}{2019}]{Wang2019Synchronization}
{Wang}, P., G.~{Wen}, X.~{Yu}, W.~{Yu}, and T.~{Huang} (2019, March).
\newblock Synchronization of multi-layer networks: From node-to-node
  synchronization to complete synchronization.
\newblock {\em IEEE Transactions on Circuits and Systems I: Regular
  Papers\/}~{\em 66\/}(3), 1141--1152.

\bibitem[\protect\citeauthoryear{Wang, Kooij, Moreno, and Van~Mieghem}{Wang
  et~al.}{2019}]{Mieghem2019Regular}
Wang, X., R.~E. Kooij, Y.~Moreno, and P.~Van~Mieghem (2019, Jan).
\newblock Structural transition in interdependent networks with regular
  interconnections.
\newblock {\em Phys. Rev. E\/}~{\em 99}, 012311.

\bibitem[\protect\citeauthoryear{Yagan and Gligor}{Yagan and
  Gligor}{2012}]{Yagan2012Contagion}
Yagan, O. and V.~Gligor (2012, Sep).
\newblock Analysis of complex contagions in random multiplex networks.
\newblock {\em Phys. Rev. E\/}~{\em 86}, 036103.

\bibitem[\protect\citeauthoryear{{Yagan}, {Qian}, {Zhang}, and
  {Cochran}}{{Yagan} et~al.}{2012}]{Yagan2012OptAlc}
{Yagan}, O., D.~{Qian}, J.~{Zhang}, and D.~{Cochran} (2012, Sep.).
\newblock Optimal allocation of interconnecting links in cyber-physical
  systems: Interdependence, cascading failures, and robustness.
\newblock {\em IEEE Transactions on Parallel and Distributed Systems\/}~{\em
  23\/}(9), 1708--1720.

\bibitem[\protect\citeauthoryear{Yang, Tu, Li, and Guo}{Yang
  et~al.}{2019}]{Yang2019Synch}
Yang, Y., L.~Tu, K.~Li, and T.~Guo (2019).
\newblock Optimized inter-structure for enhancing the synchronizability of
  interdependent networks.
\newblock {\em Physica A: Statistical Mechanics and its Applications\/}~{\em
  521}, 310 -- 318.

\bibitem[\protect\citeauthoryear{Zhang, Boccaletti, Guan, and Liu}{Zhang
  et~al.}{2015}]{Zhang2015Explosive}
Zhang, X., S.~Boccaletti, S.~Guan, and Z.~Liu (2015, Jan).
\newblock Explosive synchronization in adaptive and multilayer networks.
\newblock {\em Phys. Rev. Lett.\/}~{\em 114}, 038701.

\bibitem[\protect\citeauthoryear{{Zhang} and {Yagan}}{{Zhang} and
  {Yagan}}{2019}]{Yagan2019Robustness}
{Zhang}, Y. and O.~{Yagan} (2019).
\newblock Robustness of interdependent cyber-physical systems against cascading
  failures.
\newblock {\em IEEE Transactions on Automatic Control\/}, 1--1.

\bibitem[\protect\citeauthoryear{{Zhuang} and {Yagan}}{{Zhuang} and
  {Yagan}}{2019}]{Yagan2019Contagion}
{Zhuang}, Y. and O.~{Yagan} (2019).
\newblock Multi-stage complex contagions in random multiplex networks.
\newblock {\em IEEE Transactions on Control of Network Systems\/}, 1--1.

\end{thebibliography}

\appendix
\section{Supplemental materials }\label{App}

\subsection{Further bounds on algebraic connectivity}\label{AppBounds}
Decomposing $v_1$ and $v_2$ in \eqref{eq:v1v2} in various ways can yield various bounds for algebraic connectivity. For instance, consider the decomposition $v_1=\alpha u_2^{(1)}+y_1, \|u_2^{(1)}\|=1, y_1\in\mathbb R^n, y_1^Tu_2^{(1)}=0$, where $u_2^{(1)}$ is the eigenvector associated with the second smallest eigenvalue $\lambda_2^{(1)}$, or Fiedler vector, of $L_1$, $L_1u_2^{(1)}=\lambda_2^{(1)}u_2^{(1)}, \boldsymbol{1}_n^T{u_2^{(1)}}=0$. To ensure $v_1^T\boldsymbol{1}_n=-v_2^T\boldsymbol{1}_m$, we should have $y_1^T\boldsymbol{1}_n=-v_2^T\boldsymbol{1}_m$. Inserting this decomposition of $v_1$ into \eqref{eq:v1v2}, we can get
\begin{equation}\label{vdecomp2}
\begin{gathered}
\alpha^2\left(\lambda_2^{(1)}+{u_2^{(1)}}^T\text{diag}\left(W\boldsymbol 1_m\right)u_2^{(1)}-\lambda_2\right)+2\alpha\left({u_2^{(1)}}^T\text{diag}\left(W\boldsymbol 1_m\right)y_1-{u_2^{(1)}}^TWv_2\right) \\ -2y_1^TWv_2+v_2^T\left(L_2+\text{diag}\left(W^T\boldsymbol 1_n\right)\right)v_2-\lambda_2\left(\|y_1\|^2+\|v_2\|^2\right)\geq 0, \\ \forall \ y_1^Tu_2^{(1)}=0, \ y_1^T\boldsymbol{1}_n=-v_2^T\boldsymbol{1}_m
\end{gathered}
\end{equation}
Since the condition in \eqref{vdecomp2} must hold for every $\alpha$, the coefficient of $\alpha^2$ must be positive, which yields the following bound for algebraic connectivity
%Setting $\alpha=0, u_1=u_2^{(1)}, u_2=0$ in \eqref{eq:lam2main}, we can get the following bound for algebraic connectivity of supra-Laplacian
\begin{equation}\label{eq:lam2bound1}
\begin{gathered}
\lambda_2\leq\lambda_2^{(1)}+{u_2^{(1)}}^T\text{diag}\left(W\boldsymbol 1_m\right)u_2^{(1)}
\end{gathered}
\end{equation}

\begin{lemma}\label{lem:UpBoundCond1}
	The upper-bound in \eqref{eq:lam2bound1} is not attainable.
	%	\begin{equation}\label{eq:set1}
	%	\begin{gathered}
	%	\mathcal S^1=\{{u_2^{(1)}}^TWu_m=0, \ {u_2^{(1)}}^T\text{diag}\left(W\boldsymbol 1_m\right)y_1=0, \ \ \forall \ y_1\in\mathbb R^n, \ y_1^Tu_2^{(1)}=y_1^T\boldsymbol{1}_n=0\}
	%	\end{gathered} 
	%	\end{equation}
	%	which implies $Wu_m\perp\text{span}\{u_2^{(1)}\}$, and $\text{diag}\left(W\boldsymbol 1_m\right)u_2^{(1)} \in \text{span}\{\boldsymbol 1_n, u_2^{(1)}\}$.
\end{lemma}
\begin{proof}
	Considering the condition in \eqref{eq:lam2bound1} as equality, we see that the coefficient of $\alpha^2$ in \eqref{vdecomp2} becomes zero. Then, since \eqref{vdecomp2} must hold for every $\alpha$, the coefficient of $\alpha$ must vanish as well:
	\begin{equation}\label{eq:lamCond1}
	\begin{gathered}
	{u_2^{(1)}}^T\text{diag}\left(W\boldsymbol 1_m\right)y_1-{u_2^{(1)}}^TWv_2=0, \\ \forall \ y_1^Tu_2^{(1)}=0, \ y_1^T\boldsymbol{1}_n=-v_2^T\boldsymbol{1}_m
	\end{gathered}
	\end{equation}
	Setting $y_1=\frac{1}{n}\boldsymbol 1_n, v_2=-\frac{1}{m}\boldsymbol 1_m$ in \eqref{eq:lamCond1}, we get ${\frac{1}{n}u_2^{(1)}}^T\text{diag}\left(W\boldsymbol 1_m\right)\boldsymbol 1_n+\frac{1}{m}{u_2^{(1)}}^TW\boldsymbol 1_m=0$. Since $\text{diag}\left(W\boldsymbol 1_m\right)\boldsymbol 1_n=W\boldsymbol 1_m$, it follows the weights matrix $W$ should belong to $W\in\mathcal S^1=\{{u_2^{(1)}}^TW\boldsymbol{1}_m=0\}=\{{u_2^{(1)}}^TW\perp\text{spam}\{\boldsymbol{1}_m\}\}$. On the other hand, setting $y_1=0$ in \eqref{eq:lamCond1}, we note that $W\in\mathcal S^2=\{{u_2^{(1)}}^TWv_2=0, \ v_2^T\boldsymbol{1}_m=0\}=\{{u_2^{(1)}}^TW\in\text{spam}\{\boldsymbol{1}_m\}\}$. The proof follows since $\mathcal S^1\cap\mathcal S^2=\{\emptyset\}$ for $c\neq 0$. 
\end{proof}

Similarly, denoting $\lambda_2^{(2)}$ the second smallest eigenvalue of $L_2$ and $v_2^{(2)}$ the corresponding Fiedler vector, we can get the following upper-bound in terms of spectral properties of $L_2$
\begin{equation}\label{eq:lam2bound2}
\begin{gathered}
\lambda_2\leq\lambda_2^{(2)}+{v_2^{(2)}}^T\text{diag}\left(W^T\boldsymbol 1_n\right)v_2^{(2)}
\end{gathered}
\end{equation}
which is not attainable as well. Although the bounds in \eqref{eq:lam2bound1} and \eqref{eq:lam2bound2} are not attainable, they can give some intuition about the maximum algebraic connectivity problem under investigation. 

\subsection{Different Diffusion phases}\label{AppDiffusion}
Figure \ref{fig:TimeRes} shows different diffusion phases corresponding to the conditions considered in Figure \ref{fig:EmbGeoCase1}. To better realize the interlinks effect, we have assumed identical initial conditions of nodes in same layers. In Figure \ref{fig:TimeResC1}, when $c<c^*$, each network operation is distinguishably unified. Indeed, for small $c$, the interconnection strength is too weak to affect the intralayer connections and break the individual networks unity. In such condition, the optimal diffusion process within each network component is prominently through its intralinks. For the intermediate value $c=10$ in Figure \ref{fig:TimeResC10}, while the network $G_1$ with larger specific connectivity still operates as a unity, the network $G_2$, with smaller specific connectivity, loses its operation unity due to being interconnected with $G_1$. Thus, the subgraph with larger specific connectivity is more robust against intermediate couplings while the other with smaller specific connectivity is more vulnerable and loses unity in this region. As such, for intermediate $c^*<c<c^{**}$ in Figure \ref{fig:TimeResC10}, the optimal diffusion within $G_1$ is mostly due to its intralinks while in $G_2$ it is through intralinks and interlinks. The situation turns conversely for the larger value $c=30>c^{**}$ in Figure \ref{fig:TimeResC30} where $G_1$ loses unity while $G_2$ becomes again unified. This time, the optimal interlinks strength is so high that it completely overcomes the intralink effects within $G_2$. In fact, the strong interlinks may be thought of as powerful attracting forces that pull the nodes in $G_2$ toward each other from all sides, due to all-to-all interconnection possibility, thus making them unified. However, the interlinks are only strong enough to destroy $G_1$ unity but not strong enough to completely overcome the intralink effects in this subgraph. As such, diffusion within $G_2$ forms prominently through interlinks while it is through interlinks and intralinks in $G_1$.
\begin{figure*}
	\centering
	\subfloat[\label{fig:TimeResC1}]{\includegraphics[clip,width=.3\columnwidth]{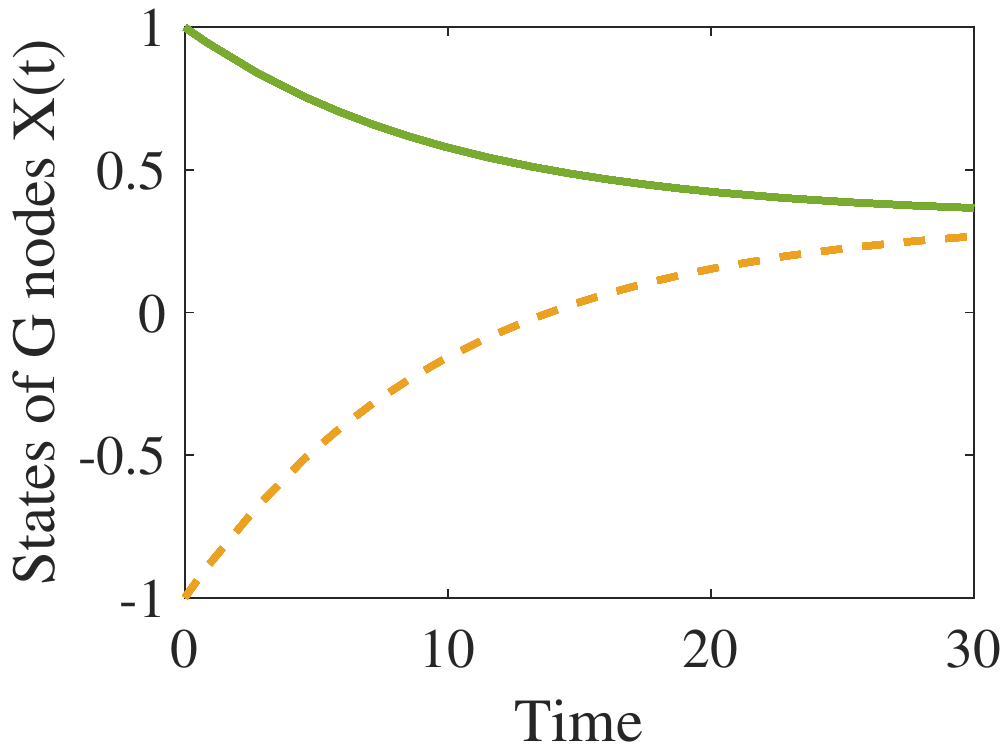}} \ \ \ \ \
	\subfloat[\label{fig:TimeResC10}]{\includegraphics[clip,width=.3\columnwidth]{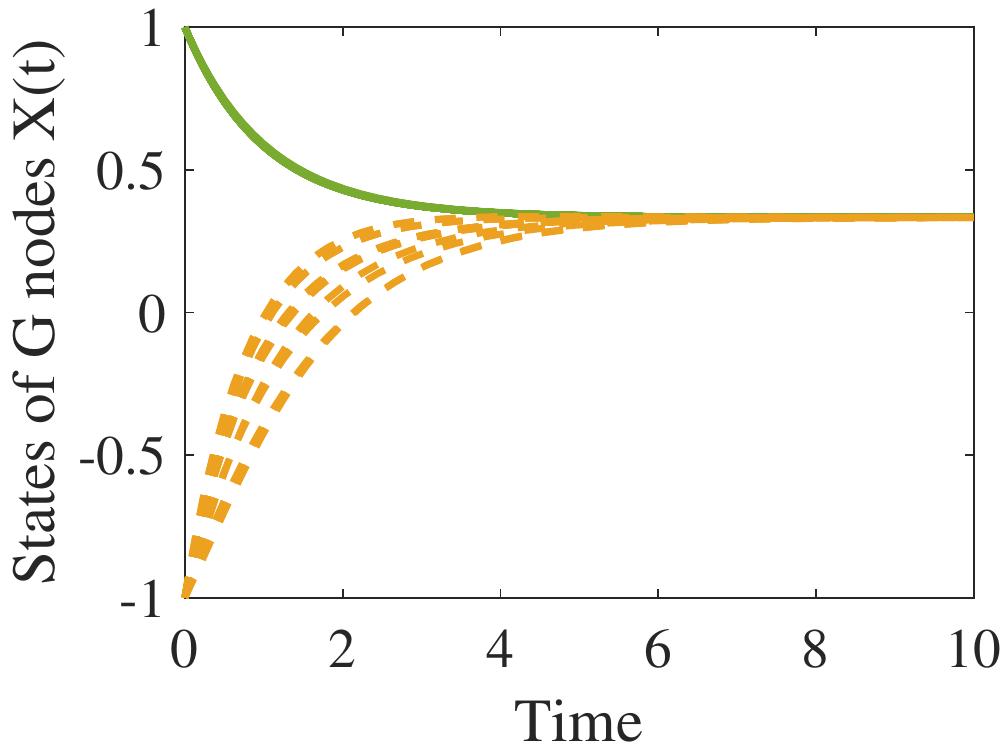}} \ \ \ \ \
	\subfloat[\label{fig:TimeResC30}]{\includegraphics[clip,width=.3\columnwidth]{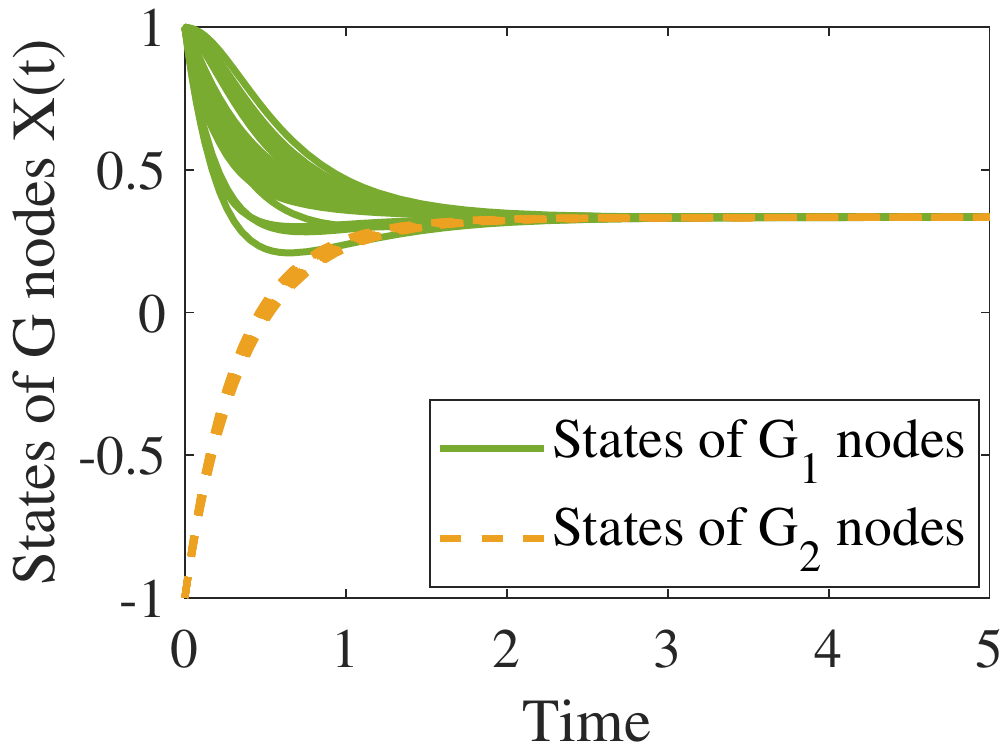}} 
	%	\subfloat[\label{fig:TimeResC100}]{\includegraphics[clip,width=.4\columnwidth]{TimeResC100.pdf}}
	\caption{Diffusion process $\dot X=-LX$ corresponding to Figure \ref{fig:EmbGeoCase1} for (a) $c=1$, (b) $c=10$, and (c) $c=30$.}
	\label{fig:TimeRes} 
	
\end{figure*}
%\begin{lemma}\label{lem:UpBoundCond02}
%	The upper-bound in \eqref{eq:lam2bound1} is achieved only if $W\boldsymbol 1_m\in \text{span}\{\boldsymbol 1_n\}$, or equivalently only if all nodes in  $G_1$ are assigned identical total interlink weight. Then, the maximum algebraic connectivity associated with this weight distribution is
%	\begin{equation}
%	\begin{gathered}
%	F\left(c\right)=\lambda_2^{(1)}+\frac{c}{n}
%	\end{gathered} 
%	\end{equation}
%\end{lemma}
%\begin{proof}
%	Considering the condition in \eqref{eq:lam2bound1} as equality, we see that the coefficient of $\alpha^2$ in \eqref{vdecomp2} becomes zero. Then, since \eqref{vdecomp2} must hold for every $\alpha$, the coefficient of $\alpha$ must vanish as well:
%	\begin{equation*}
%	\begin{gathered}
%	{u_2^{(1)}}^T\text{diag}\left(W\boldsymbol 1_m\right)y_1-{u_2^{(1)}}^TWv_2=0, \\ \forall \ y_1^Tu_2^{(1)}=0, \ y_1^T\boldsymbol{1}_n=-v_2^T\boldsymbol{1}_m
%	\end{gathered}
%	\end{equation*}
%	Setting $v_2=0$,
%	\begin{equation*}
%	\begin{gathered}
%	{u_2^{(1)}}^T\text{diag}\left(W\boldsymbol 1_m\right)y_1=0, \ \ \forall \ y_1\in\mathbb R^n, \ y_1^Tu_2^{(1)}=y_1^T\boldsymbol{1}_n=0
%	\end{gathered}
%	\end{equation*}
%	which implies
%	\[
%	\text{diag}\left(W\boldsymbol 1_m\right)u_2^{(1)} \in \text{span}\{\boldsymbol 1_n, u_2^{(1)}\}, \ \ \boldsymbol 1_n^Tu_2^{(1)}=0
%	\]
%	This shows $W\boldsymbol 1_m\in \text{span}\{\boldsymbol 1_n\}\boldsymbol{!!!}$, which states all rows of $W$ have the same sum, or equivalently all nodes in  $G_1$ are assigned the same total interlink weight. 
%\end{proof}
\end{document}